\documentclass[a4paper,11pt]{article}

\usepackage[a4paper, total={6.5in, 10in}]{geometry}
\usepackage{setspace} 
\usepackage{amsmath,amsthm,amssymb,graphicx,mathtools,tikz,hyperref}
\usepackage{algorithm}
\usepackage{booktabs}
\usepackage{algpseudocode}
\usepackage{float}
\theoremstyle{plain}

\newtheorem{proposition}{Proposition}
\newtheorem{lemma}{Lemma}

\newtheorem{corollary}{Corollary}
\newtheorem{assumption}{Assumption}[section]
\theoremstyle{remark}

\usepackage[round]{natbib}
\usepackage{authblk}
\usetikzlibrary{positioning}

\hypersetup{
    colorlinks,
    linkcolor={blue!50!black},
    citecolor={blue!50!black},
    urlcolor={blue!80!black}
}
\begin{document}
\date{October 5, 2026}

\title{Impulse Response Inference for Matrix Autoregressions}
\author{Alain Hecq}
\author{Ivan Ricardo}
\author{Ines Wilms}
\affil{Maastricht University, Department of Quantitative Economics} 

\maketitle

\begin{abstract}
Matrix autoregressive (MAR) models offer a parsimonious framework for modeling matrix-valued time series, yet tools for estimation and inference for their impulse response functions are lacking.
We develop asymptotic and bootstrap-based inference for impulse responses of stable MAR($p$) models.
We derive the joint asymptotic distribution of the coefficient and covariance estimators, which permits closed-form delta-method standard errors.
To address finite-sample bias, we propose ProBAB-MAR, a bias-corrected bootstrap that projects the corrected coefficients back onto the Kronecker parameter space, and prove that it attains asymptotically correct coverage.
Monte Carlo simulations show that delta-method intervals undercover in small samples, while ProBAB-MAR achieves near-nominal coverage with intervals considerably narrower than those from an unrestricted VAR.
An application to euro area inflation illustrates how shocks transmit across countries and inflation categories.

\medskip

\end{abstract}

\newpage

\section{Introduction}
\label{sec:intro}
A growing share of time series data are naturally represented as matrix-valued time series, with observations indexed by two cross-sectional dimensions rather than by a single vector of series. 
Matrix autoregressive (MAR) models have become increasingly popular to model such series. 
In this paper, we develop asymptotic and bootstrap-based impulse response inference for stable MAR models.
We derive the joint asymptotic distribution of the MAR coefficient and covariance estimators, which yields closed-form standard errors for impulse response coefficients via the delta-method.
We also propose a bias-corrected bootstrap procedure for MAR impulse responses, labeled ProBAB-MAR, and prove that it attains asymptotically correct coverage.

Macroeconomic and financial applications provide many natural examples of matrix-valued time series where one dimension indexes economic indicators and the other indexes the units (e.g., countries, sectors, or firms) to which they belong.
A growing literature on matrix-valued time series \citep{chen_autoregressive_2021, tsay_matrix_2024, zhang_additive_2024, samadi_matrix_2025} models such data directly, keeping the matrix structure that vectorization (e.g. the well known VAR($p$)) would discard.
For example, the application we consider in this paper records year-on-year inflation for $N_1 = 5$ euro area countries across $N_2 = 3$ categories (energy, food, and core), so that each month delivers an $N_1 \times N_2$ matrix of observations.
Vectorizing such data and fitting a vector autoregression (VAR; \citealp{lutkepohl_new_2005, kilian_structural_2017}) discards the row- and column-specific structure of the matrix and quickly becomes difficult to estimate as the two dimensions grow, since the number of coefficients scales with $(N_1 N_2)^2 p$ for lag order $p$.
Panel VARs \citep{holtz_estimating_1988, koop_forecasting_2019} accommodate multiple cross-sections but do not model the two-way dependence that is intrinsic to matrix-valued data.
The matrix autoregressive model (MAR) of \citet{chen_autoregressive_2021} instead imposes a Kronecker structure on the autoregressive coefficients, separating the dynamics along the rows from those along the columns and reducing the parameter count from $(N_1 N_2)^2 p$ to $(N_1^2 + N_2^2 - 1)p$.
In our inflation application, with $N_1 = 5$ countries, $N_2 = 3$ categories, and $p = 2$ lags, this is the difference between estimating $66$ and $450$ coefficients.
That parsimony is what makes the MAR attractive.

\citet{chen_autoregressive_2021} establish the asymptotic normality of the coefficient estimators for the MAR($1$), providing a basis for inference on the row and column interaction matrices, with \citet{li_multilinear_2021} extending this to tensors with general lag order.
While the coefficient matrices may themselves be of interest, much empirical work ultimately cares about impulse responses, which trace how a shock to one entry of the matrix propagates across both dimensions over time \citep{kilian_structural_2017}.
Despite the extensive literature on estimating MAR models, research on impulse response analysis for MAR models remains scarce, with notable exceptions including the recent work by \citet{celani_matrix_2024, bucci_structural_2026, lara_structural_2026}.
Moreover, available asymptotic results in the literature typically stop at the coefficient matrices of the MAR. 
Impulse responses in the matrix setting are therefore typically reported as point estimates alone \citep{chen_autoregressive_2021}, or accompanied by intervals from residual bootstrap schemes carried over from the unrestricted VAR, whose validity under the Kronecker restriction has not been established \citep{bucci_structural_2026}.

In this paper, we fill this gap by developing inference for the impulse responses of a MAR($p$) and make three contributions.
First, we obtain the joint asymptotic distribution of the coefficient and covariance estimators of the MAR($p$), which the existing results for matrix and tensor autoregressions, stated for the coefficients alone, do not cover \citep{li_multilinear_2021}.
We then apply the delta method \citep[Ch. 11]{hamilton_time_1994} to obtain the asymptotic distribution of the impulse responses, yielding closed-form standard errors at any fixed horizon.
Second, we propose ProBAB-MAR, a bias-corrected bootstrap procedure for MAR impulse responses, and prove that it attains asymptotically correct coverage.
ProBAB-MAR takes its inspiration from the bootstrap-after-bootstrap of \citet{kilian_small_1998} for VAR models, which uses a first-stage bootstrap to estimate the small-sample bias of the coefficient estimators, generates the bootstrap samples from the bias-corrected model, and extends it to the MAR setting.
This extension requires care, because the estimated bias does not itself have a Kronecker form, so subtracting it from the estimated coefficients leaves a coefficient matrix outside the MAR parameter space.
ProBAB-MAR therefore projects the bias-corrected coefficients back onto the space of Kronecker products before the bootstrap samples are generated, so that the bootstrap design and the estimated model coincide.
Third, and of independent interest, we characterize the rank structure of the matricized impulse response for MAR($p$) models.
For a MAR(1), a shock induces proportional co-movement: every unit responds with the same pattern, scaled only by a position-specific factor, so that the matricized response is rank one, a result established already by \cite{chen_autoregressive_2021}.
This proportional response is a co-movement restriction of the kind studied in the common-features literature \citep{engle_testing_1993, cubadda_studying_2009}.
For a MAR($p$) ($p \geq 2$), we show that this proportionality breaks down gradually as the horizon grows, thereby admitting an increasingly rich pattern of heterogeneous responses across the two dimensions.

We compare the asymptotic-based inference with the bootstrap-based inference in a Monte Carlo study and find that the asymptotic delta-method intervals tend to undercover in small samples, a distortion that the bootstrap corrects.
The MAR-based bootstrap interval, by contrast, maintains nominal coverage while being substantially narrower than the intervals from an unrestricted VAR bootstrap, a gain that reflects the parsimony of the Kronecker structure.
We then illustrate the framework with an application to euro area inflation, tracing how a common energy shock transmits across countries and into food and core inflation.

Our work builds on the matrix autoregressive literature and on bootstrap methods for impulse response inference, but differs from both in focus.
\citet{chen_autoregressive_2021} introduce the MAR, develop alternating least squares and maximum likelihood estimation, and establish the asymptotic distribution of the MAR($1$) coefficients; subsequent work refines its estimation and representation \citep{celani_matrix_2024, zhang_additive_2024, samadi_envelope_2026}.
Dimension reduction for matrix-valued time series has also been pursued through factor models that extract common row and column factors \citep{wang_factor_2019, chen_factor_2022, chen_statistical_2023, gao_transformed_2023} and through reduced-rank autoregressions \citep{xiao_reduced_forth, hecq_decomposing_2026} and reduced-rank regression more broadly \citep{cubadda_dimension_2022, cubadda_reduced_2022}.
This literature targets estimation and the underlying coefficient or factor structure, whereas we target the impulse responses these models imply and the inference required to report them.
On the inferential side, the bootstrap-after-bootstrap is the standard device for impulse-response intervals in unrestricted VARs, and it underlies the VAR intervals in recent large-scale comparisons of VAR and local projection inference \citep{kilian_reliable_2011, li_local_2024, olea_primer_2025}.
Our procedure is tailored to the Kronecker-restricted coefficient space of the MAR, exploits the Kronecker structure to deliver narrower intervals, and preserves the row--column decomposition that makes the matrix model interpretable.
We focus throughout on the stationary case; extending impulse-response inference to cointegrated matrix-valued systems \citep{li_cointegrated_2024, chen_inference_2025, hecq_cointegrated_2025, lopetuso_cointegrated_2026} is a natural direction for future work.

The remainder of this paper is structured as follows.
Section \ref{sec:model} introduces the MAR model and its estimation, and derives its impulse responses and their rank structure.
Section \ref{sec:inference} develops inference for these impulse responses, first through the joint asymptotic distribution of the coefficient and covariance estimators and the delta method, and then through ProBAB-MAR, whose asymptotic validity we establish.
Section \ref{sec:simulations} reports the Monte Carlo evidence on coverage and interval width.
Section \ref{sec:application} presents the empirical application to euro area inflation.
Section \ref{sec:conclusion} concludes and discusses avenues for future research.

A word on the notation.
We denote scalars by lowercase letters $x$, vectors by boldface lowercase letters $\mathbf{x}$, and matrices by boldface capital letters $\mathbf{X}$.
For a generic matrix $\mathbf{X}$, we write $\mathbf{X}^\top$ for its transpose, $\|\mathbf{X}\|_F$ for its Frobenius norm, and $\operatorname{vec}(\mathbf{X})$ for its column-wise vectorization, with $\operatorname{unvec}_{N_1,N_2}(\cdot)$ denoting the inverse operation that reshapes an $N_1 N_2$-vector back into an $N_1 \times N_2$ matrix.
For a symmetric $N \times N$ matrix $\mathbf{S}$, we write $\operatorname{vech}(\mathbf{S})$ for its half vectorization, which stacks the elements on and below the main diagonal column by column into an $N(N+1)/2$-vector, so that each distinct entry of $\mathbf{S}$ appears only once.
We use $\otimes$ for the Kronecker product and $\mathbf{I}_N$ for the $N \times N$ identity matrix.
Finally, we write $\operatorname{blkdiag}(\mathbf{Z}_1, \ldots, \mathbf{Z}_p)$ for the block-diagonal matrix with $\mathbf{Z}_1, \ldots, \mathbf{Z}_p$ on the diagonal and zero blocks elsewhere.

\section{Impulse Responses for Matrix Autoregressive Models}
\label{sec:model}

This section starts by reviewing the matrix autoregressive (MAR) model and its estimation (Section \ref{subsec:MAR}) from which the impulse responses are obtained (Section \ref{subsec:irf}). 

\subsection{The Matrix Autoregressive Model}
\label{subsec:MAR}
Consider a time series in which at each time $t = -p+1, \ldots, T$, an $N_1 \times N_2$ matrix $\mathbf{Y}_t$ is observed, so that $\mathbf{Y}_{-p+1}, \ldots, \mathbf{Y}_{0}$ serve as presample values and $\mathbf{Y}_{1}, \ldots, \mathbf{Y}_{T}$ as the estimation sample.
The matrix autoregressive model of order $p$, denoted MAR($p$), for $\mathbf{Y}_t$  is given by
\begin{equation}\label{eq:mar}
    \mathbf{Y}_t = \sum_{j=1}^{p} \mathbf{A}_j \mathbf{Y}_{t-j} \mathbf{B}_j^\top + \mathbf{U}_t,
\end{equation}
where $\mathbf{A}_j \in \mathbb{R}^{N_1 \times N_1}$ and $\mathbf{B}_j \in \mathbb{R}^{N_2 \times N_2}$ are coefficient matrices (for $j=1,\ldots, p$).
The left coefficient matrix $\mathbf{A}_j$ captures interactions along the row dimension, while the right coefficient matrix $\mathbf{B}_j$ captures interactions along the column dimension, both at lag $j$; see \citet{chen_autoregressive_2021} for detailed interpretations.
We assume that all series are mean-centered such that no intercept is included.
Furthermore, $\mathbf{U}_t$ is an $N_1 \times N_2$ matrix innovation process whose covariance is separable, $\operatorname{Cov}(\operatorname{vec}(\mathbf{U}_t)) = \boldsymbol{\Sigma} = \boldsymbol{\Sigma}_2 \otimes \boldsymbol{\Sigma}_1 \in \mathbb{R}^{N_{1} N_{2} \times N_{1} N_{2}}$, so that $\boldsymbol{\Sigma}_1 \in \mathbb{R}^{N_{1} \times N_{1}}$ governs contemporaneous dependence across rows and $\boldsymbol{\Sigma}_2 \in \mathbb{R}^{N_{2} \times N_{2}}$ across columns. 

Model \eqref{eq:mar} can be written as the restricted VAR($p$)
\begin{equation} \label{eq:restrictedvar}
    \mathbf{y}_t = \sum_{j=1}^{p} \underbrace{(\mathbf{B}_j \otimes \mathbf{A}_j)}_{\mathbf{C}_j} \mathbf{y}_{t-j} + \mathbf{u}_t,
\end{equation}
by applying the $\operatorname{vec}(\cdot)$ operator to the matrix-valued series in equation \eqref{eq:mar} with $\mathbf{y}_t = \operatorname{vec}(\mathbf{Y}_t) \in \mathbb{R}^{N_1 N_2}$ and $\mathbf{u}_t = \operatorname{vec}(\mathbf{U}_t) \in \mathbb{R}^{N_1 N_2}$, and using the identity $\operatorname{vec}(\mathbf{A}\mathbf{X}\mathbf{B}^\top) = (\mathbf{B} \otimes \mathbf{A})\operatorname{vec}(\mathbf{X})$ on the lagged terms.
Each lag coefficient $\mathbf{C}_j = \mathbf{B}_j \otimes \mathbf{A}_j \in \mathbb{R}^{N_{1} N_{2} \times N_{1} N_{2}}$ is constrained to have Kronecker structure, which reduces the number of free parameters from $(N_1 N_2)^2 p$ in an unrestricted VAR($p$) to $(N_1^2 + N_2^2 - 1)p$ in the MAR($p$).
Note the subtraction accounts for the scale indeterminacy between $\mathbf{A}_j$ and $\mathbf{B}_j$.
Following the standard practice in the MAR literature, we impose $\|\mathbf{A}_j\|_F = 1$ together with a fixed sign convention.
Under this normalization, the decomposition $\mathbf{C}_j = \mathbf{B}_j \otimes \mathbf{A}_j$ is uniquely defined for  $\mathbf{A}_j \neq \mathbf{0}$ and $\mathbf{B}_j \neq \mathbf{0}$ (for each $j=1, \ldots, p$).  
Similarly, we impose the normalization $\|\boldsymbol{\Sigma}_1\|_F = 1$ to resolve the scale indeterminacy between the covariance matrices.

For estimation and asymptotic analysis, it is convenient to rewrite model \eqref{eq:mar} in stacked form.
Define $\mathbf{A} = [\mathbf{A}_1 \cdots \mathbf{A}_p] \in \mathbb{R}^{N_1 \times N_1 p}$, $\mathbf{B} = [\mathbf{B}_1 \cdots \mathbf{B}_p] \in \mathbb{R}^{N_2 \times N_2 p}$, and the block-diagonal regressor matrix
$ \mathbf{X}_t = \operatorname{blkdiag}(\mathbf{Y}_{t-1}, \ldots, \mathbf{Y}_{t-p}) \in \mathbb{R}^{N_1 p \times N_2 p}$.
Then \eqref{eq:mar} becomes 
$$
\mathbf{Y}_t = \mathbf{A} \mathbf{X}_t \mathbf{B}^\top + \mathbf{U}_t.
$$ 
Collecting the coefficient and covariance parameters in 
\begin{equation*}
\begin{aligned}
    \boldsymbol{\beta} &= \big(\operatorname{vec}(\mathbf{A})^{\top}, \operatorname{vec}(\mathbf{B}^{\top})^{\top}\big)^{\top} \in \mathbb{R}^{p(N_1^2+N_2^2)}, \\
    \boldsymbol{\sigma} &= \big(\operatorname{vech}(\boldsymbol{\Sigma}_1)^\top, \operatorname{vech}(\boldsymbol{\Sigma}_2)^\top\big)^\top \in \mathbb{R}^{N_1(N_1+1)/2 + N_2(N_2+1)/2}, \\
    \boldsymbol{\zeta} &= \big(\boldsymbol{\beta}^\top,\boldsymbol{\sigma}^\top\big)^\top \in \mathbb{R}^{d},
\end{aligned}
\end{equation*}
where $d = p(N_1^2+N_2^2) + N_1(N_1+1)/2 + N_2(N_2+1)/2$, the Gaussian log-likelihood is, up to an additive constant,
\begin{equation}\label{eq:loglik}
    \ell_T(\boldsymbol{\zeta})
    = -\frac{T}{2}\Big[N_2 \log|\boldsymbol{\Sigma}_1| + N_1 \log|\boldsymbol{\Sigma}_2|\Big]
      - \frac{1}{2}\sum_{t=1}^{T} \operatorname{tr}\big(\boldsymbol{\Sigma}_1^{-1}\mathbf{U}_t\boldsymbol{\Sigma}_2^{-1}\mathbf{U}_t^\top\big),
\end{equation}
where we use $|\boldsymbol{\Sigma}_2 \otimes \boldsymbol{\Sigma}_1| = |\boldsymbol{\Sigma}_1|^{N_2}|\boldsymbol{\Sigma}_2|^{N_1}$ and $\operatorname{vec}(\mathbf{U}_t)^\top \boldsymbol{\Sigma}^{-1}\operatorname{vec}(\mathbf{U}_t) = \operatorname{tr}(\boldsymbol{\Sigma}_1^{-1}\mathbf{U}_t\boldsymbol{\Sigma}_2^{-1}\mathbf{U}_t^\top)$.
We estimate the MAR parameters by maximum likelihood estimation (MLE), thereby maximizing \eqref{eq:loglik} subject to the above discussed normalizations on $\mathbf{A}_j$ and $\boldsymbol{\Sigma}_1$.
We hereby use the alternating maximum likelihood algorithm (see e.g., \citealp{chen_autoregressive_2021, xiao_reduced_forth, li_multilinear_2021}), initialized at the projection estimator of \citet{chen_autoregressive_2021}.

\subsection{MAR Impulse Responses}
\label{subsec:irf}
We now present the impulse responses of the MAR model and show how its Kronecker structure translates into a rank structure of the impulse responses.
Proposition \ref{prop:rank} gives a rank bound for MAR($p$) models, thereby extending \citet{chen_autoregressive_2021}, who show that the matricized impulse response of a MAR($1$) has rank one.

To obtain the matrix moving average (MMA) representation of the MAR, we exploit the fact that the MAR($p$) is a restricted VAR($p$), as given in equation \eqref{eq:restrictedvar}.
Its moving average representation therefore follows from the usual VAR recursions under standard conditions on the coefficients and the innovations.
Let
\begin{equation*}
    \mathbf{F} =
    \begin{bmatrix}
        \mathbf{C}_1 & \mathbf{C}_2 & \cdots & \mathbf{C}_{p-1} & \mathbf{C}_p \\
        \mathbf{I}_{N_1 N_2} & \mathbf{0} & \cdots & \mathbf{0} & \mathbf{0} \\
        \mathbf{0} & \mathbf{I}_{N_1 N_2} & \cdots & \mathbf{0} & \mathbf{0} \\
        \vdots & & \ddots & \vdots & \vdots \\
        \mathbf{0} & \mathbf{0} & \cdots & \mathbf{I}_{N_1 N_2} & \mathbf{0}
    \end{bmatrix}
    \in \mathbb{R}^{N_1 N_2 p \times N_1 N_2 p},
\end{equation*}
denote the companion matrix of \eqref{eq:restrictedvar}.
We impose the following assumptions.

\begin{assumption}[Stability]\label{ass:stationary}
All eigenvalues of the companion matrix $\mathbf{F}$ lie strictly inside the unit circle.
\end{assumption}

\begin{assumption}[Innovations]\label{ass:wn}
The innovations $\{\mathbf{U}_t\}$ form a matrix white noise process, that is, they are serially uncorrelated with $\mathbb{E}(\mathbf{U}_t) = \mathbf{0}$ and have a separable covariance $\operatorname{Cov}(\operatorname{vec}(\mathbf{U}_t)) = \boldsymbol{\Sigma}_2 \otimes \boldsymbol{\Sigma}_1$, with $\boldsymbol{\Sigma}_1$ and $\boldsymbol{\Sigma}_2$ positive definite.
\end{assumption}

Assumption \ref{ass:stationary} restricts our focus to stable MAR models.
For $p = 1$, the eigenvalues of $\mathbf{F} = \mathbf{B}_1 \otimes \mathbf{A}_1$ are the products of the eigenvalues of $\mathbf{A}_1$ and $\mathbf{B}_1$, so that the condition reduces to $\rho(\mathbf{A}_1)\rho(\mathbf{B}_1) < 1$, with $\rho(\cdot)$ denoting the spectral radius \citep[Proposition 1]{chen_autoregressive_2021}.
Assumption \ref{ass:wn} only restricts the first two moments of the innovations.
Positive definiteness of $\boldsymbol{\Sigma}_1$ and $\boldsymbol{\Sigma}_2$ implies that $\boldsymbol{\Sigma}$ is positive definite, which guarantees that the impact matrices introduced at the end of this section exist.
Separability is nonetheless restrictive, as it requires the contemporaneous covariance across rows to be proportional across columns, and vice versa.
When it fails, a separable model may misrepresent the dependence among the innovations \citep{hoff_core_2023}, but separability can be assessed with dedicated tests \citep{roy_implementation_2005, srivastava_models_2008, sung_testing_2025}.
The inference in Section \ref{sec:inference} requires stronger conditions on the innovations, namely independence and finite fourth moments, which we introduce and discuss there.

Under Assumptions \ref{ass:stationary} and \ref{ass:wn}, $\mathbf{y}_t$ has a matrix moving average (MMA) representation
\begin{equation*}
    \mathbf{y}_{t} = \sum_{h=0}^{\infty} \boldsymbol{\Theta}_h \mathbf{u}_{t-h},
\end{equation*}
where the coefficient matrices $\boldsymbol{\Theta}_h \in \mathbb{R}^{N_1 N_2 \times N_1 N_2}$ follow the recursion
\begin{equation}\label{eq:mma}
    \boldsymbol{\Theta}_0 = \mathbf{I}_{N_1 N_2}, \qquad
    \boldsymbol{\Theta}_h = \sum_{j=1}^{\min(p,h)} \mathbf{C}_j \boldsymbol{\Theta}_{h-j}, \qquad h \geq 1.
\end{equation}
Throughout, we index entries of $\mathbf{Y}_t$ either by their position $(r,s)$ in the matrix or by the corresponding position $r + (s-1)N_1$ in the vectorization, similarly for indexing entries in $\mathbf{U}_t$.
Now consider a unit shock to entry $(r,s)$ of $\mathbf{U}_t$.
This shock corresponds to the basis vector $\mathbf{e}_k$ with $k = r + (s-1)N_1$ where $\mathbf{e}_{r + (s-1)N_1} = \mathbf{e}_s \otimes \mathbf{e}_r \in \mathbb{R}^{N_1 N_2}$.
The reduced-form responses at horizon $h$ of each component series in $\mathbf{y}_t$ to the given unit shock are then collected in the vector $\boldsymbol{\Theta}_h \mathbf{e}_k$.
For a single element, the horizon $h$ response at vectorized position $i = m + (n-1)N_1$ in  $\mathbf{y}_t$ (equivalently row unit $m$ in column unit $n$ of $\mathbf{Y}_t$), to the shock at position $k$ in $\mathbf{u}_t$ is 
\begin{equation*}
    [\boldsymbol{\Theta}_{h}]_{ik}  = \theta_{ik, h} = \mathbf{e}_i^\top \boldsymbol{\Theta}_h \mathbf{e}_k.
\end{equation*}

We next examine how the Kronecker structure of the MAR model translates into structure in its impulse responses.
Consider first the case of a MAR($1$). The Kronecker structure in $\mathbf{C}_1$ yields a factorization of the impulse response into separate row and column components.
Since $\boldsymbol{\Theta}_h = (\mathbf{B}_{1} \otimes \mathbf{A}_{1})^h = \mathbf{B}_{1}^h \otimes \mathbf{A}_{1}^h$ for $p = 1$, we have
\begin{equation}\label{eq:rc_decomp}
    \operatorname{unvec}_{N_1, N_2}(\boldsymbol{\Theta}_h \mathbf{e}_k)
    = \underbrace{(\mathbf{A}_1^h \mathbf{e}_r)}_{\text{row response}}
      \underbrace{(\mathbf{B}_1^h \mathbf{e}_s)^\top}_{\text{column response}}.
\end{equation}
We refer to $\mathbf{A}_1^h \mathbf{e}_r$ as the row response, which traces how the shock propagates across the $N_1$ row units independently of the column dimension, and to $\mathbf{B}_1^h \mathbf{e}_s$ as the column response, tracing propagation across the $N_2$ column units.
This matrix has rank at most one, a property noted by \citet{chen_autoregressive_2021} in the context of their shock-first impulse response function under a Kronecker-structured error covariance.
All units respond proportionally to the shock, differing only by a scalar that depends on their position in the row or column dimension.
In our inflation example, this would imply that all countries respond to a shock with the same pattern across inflation categories, scaled by a country-specific factor.

Now consider the case of a MAR($p$) model.
For $p \geq 2$, the rank-one factorization in \eqref{eq:rc_decomp} no longer holds at every horizon, because the MMA recursion \eqref{eq:mma} sums multiple Kronecker products.
The following proposition, proved in Appendix \ref{app:rank}, characterizes the resulting rank structure.

\begin{proposition}\label{prop:rank}
    Let Assumption \ref{ass:stationary} hold and let $\boldsymbol{\Theta}_h$ be the MMA
    coefficients of the MAR($p$), generated by the recursion \eqref{eq:mma} with
    $\mathbf{C}_j = \mathbf{B}_j \otimes \mathbf{A}_j$ for $j = 1, \ldots, p$.
    Let $\mathbf{v} = \mathbf{v}_2 \otimes \mathbf{v}_1$ with $\mathbf{v}_1 \in \mathbb{R}^{N_1}$
    and $\mathbf{v}_2 \in \mathbb{R}^{N_2}$, and define the integer sequence $c_0 = 1$ and
    $c_h = \sum_{j=1}^{\min(p,h)} c_{h-j}$ for $h \geq 1$.
    Then for each $h \geq 0$ there exist vectors
    $\mathbf{a}^{(\ell)} \in \mathbb{R}^{N_1}$ and $\mathbf{b}^{(\ell)} \in \mathbb{R}^{N_2}$,
    $\ell = 1, \ldots, c_h$, such that
    \begin{equation}\label{eq:kron_expansion}
        \boldsymbol{\Theta}_h \mathbf{v}
        = \sum_{\ell=1}^{c_h} \mathbf{b}^{(\ell)} \otimes \mathbf{a}^{(\ell)},
        \qquad\text{equivalently}\qquad
        \operatorname{unvec}_{N_1,N_2}(\boldsymbol{\Theta}_h \mathbf{v})
        = \sum_{\ell=1}^{c_h} \mathbf{a}^{(\ell)} \big(\mathbf{b}^{(\ell)}\big)^{\top},
    \end{equation}
    and consequently
    \begin{equation}\label{eq:rank_bound}
        \operatorname{rank} \left(\operatorname{unvec}_{N_1, N_2}(\boldsymbol{\Theta}_h \mathbf{v})\right)
        \leq \min\{c_h, N_1, N_2\}.
    \end{equation}
\end{proposition}

The reduced-form response takes $\mathbf{v} = \mathbf{e}_k = \mathbf{e}_s \otimes \mathbf{e}_r$.
The sequence $c_h$ counts the number of Kronecker product terms in the expansion of $\boldsymbol{\Theta}_h \mathbf{v}$.
For $p = 1$, $c_h = 1$ for all $h$, recovering the rank-one structure of \eqref{eq:rc_decomp}.
For $p = 2$, the sequence follows the Fibonacci recursion: $c_h = 1, 1, 2, 3, 5, 8, \ldots$, so the rank bound grows with the horizon until it saturates at $\min(N_1, N_2)$.

\begin{corollary}\label{cor:rank}
    Let $\{\mathbf{a}^{(\ell)}\}_{\ell=1}^{c_h}$ and $\{\mathbf{b}^{(\ell)}\}_{\ell=1}^{c_h}$ be the factor vectors
    of the expansion \eqref{eq:kron_expansion}.
    If each of the two sets is linearly independent, then $c_h \leq \min(N_1, N_2)$ and the bound
    \eqref{eq:rank_bound} holds with equality,
    \begin{equation*}
        \operatorname{rank} \left(\operatorname{unvec}_{N_1, N_2}(\boldsymbol{\Theta}_h \mathbf{v})\right)
        = \min\{c_h, N_1, N_2\}.
    \end{equation*}
\end{corollary}

The rank bound has a concrete economic interpretation.
At $p = 1$, all units co-move proportionally at every horizon.
At $p \geq 2$, this co-movement restriction is progressively relaxed as the horizon increases.
The matricized impulse responses then admit an increasingly rich pattern of heterogeneous responses across the two dimensions.

Finally, the impulse responses presented so far are reduced-form impulse responses, that is, responses of the system to a shock in the reduced-form innovations $\mathbf{u}_t$.
To permit causal interpretation of the impulse responses, structural assumptions about the data generating process need to be imposed.
Consider, for instance, identification through an impact matrix $\mathbf{P}$ satisfying $\boldsymbol{\Sigma} = \mathbf{P}\mathbf{P}^\top$, which exists by Assumption \ref{ass:wn}.
The horizon-$h$ response at vectorized position $i$ in $\mathbf{y}_t$ to the identified shock at position $k$ is then
\begin{equation}\label{eq:theta_def}
    \theta_{ik,h} = g_{ik,h}(\boldsymbol{\zeta})
    = \mathbf{e}_i^\top \boldsymbol{\Theta}_h(\boldsymbol{\beta}) \mathbf{P}(\boldsymbol{\sigma}) \mathbf{e}_k,
\end{equation}
where we make explicit that $\boldsymbol{\Theta}_h$ depends on the coefficients through the recursion \eqref{eq:mma} and $\mathbf{P}$ on the covariance parameters through the identification scheme.
The reduced-form response is nested as the special case $\mathbf{P} = \mathbf{I}_{N_1 N_2}$.

Given the separable covariance structure $\boldsymbol{\Sigma} = \boldsymbol{\Sigma}_2 \otimes \boldsymbol{\Sigma}_1$ maintained by the MLE, an identification carried out on the two covariance matrices separately (see e.g., \citealp{bucci_structural_2026, lara_structural_2026}) gives an impact matrix that respects the matrix structure.
Factoring $\boldsymbol{\Sigma}_1 = \mathbf{P}_1 \mathbf{P}_1^\top$ and $\boldsymbol{\Sigma}_2 = \mathbf{P}_2 \mathbf{P}_2^\top$ and setting $\mathbf{P} = \mathbf{P}_2 \otimes \mathbf{P}_1$ yields a valid impact matrix, since the mixed-product property gives $\mathbf{P}\mathbf{P}^\top = (\mathbf{P}_2 \mathbf{P}_2^\top) \otimes (\mathbf{P}_1 \mathbf{P}_1^\top) = \boldsymbol{\Sigma}$, and the impact vector inherits the Kronecker form of $\mathbf{e}_k$,
\begin{equation}\label{eq:kron_impact}
    \mathbf{P}\mathbf{e}_k = (\mathbf{P}_2 \mathbf{e}_s) \otimes (\mathbf{P}_1 \mathbf{e}_r).
\end{equation}
For instance, the Cholesky factor of $\boldsymbol{\Sigma}$ equals $\mathbf{P}_2 \otimes \mathbf{P}_1$ when $\mathbf{P}_1$ and $\mathbf{P}_2$ are the Cholesky factors of $\boldsymbol{\Sigma}_1$ and $\boldsymbol{\Sigma}_2$, and rotations applied to each covariance matrix separately preserve this form.
Since Proposition \ref{prop:rank} and Corollary \ref{cor:rank} are stated for an arbitrary Kronecker vector, they cover such separably identified responses by taking $\mathbf{v} = \mathbf{P}\mathbf{e}_k$ as in \eqref{eq:kron_impact}, and more generally the impact of any shock that is itself a Kronecker vector, such as the joint shock of Section \ref{subsec:ident}.
In contrast, under an unrestricted $\boldsymbol{\Sigma}$, or under an identification applied jointly to the full system through a non-Kronecker rotation, the impact vector $\mathbf{P}\mathbf{e}_k$ need not be a Kronecker vector, and the rank results no longer apply.
Assumption \ref{ass:wn} plays no role in Proposition \ref{prop:rank} itself, but it is what carries the rank results over to structural responses, which require Kronecker structure in the innovation covariance as well as in the coefficients.

\section{Inference for MAR Impulse Responses}\label{sec:inference}
This section develops inference for the impulse responses defined in Section \ref{subsec:irf}.
We first derive the joint asymptotic distribution of the coefficient and covariance estimators and apply the delta method, which yields closed-form standard errors (Section \ref{subsec:asymptotic}).
We then develop a bias-corrected bootstrap procedure, labeled ProBAB-MAR, whose asymptotic validity we establish (Section \ref{subsec:bootstrap}).
Throughout this section, we maintain Assumption \ref{ass:stationary} and impose the following two assumptions.

\begin{assumption}[Independence and fourth moments]\label{ass:iid-plus}
Assumption \ref{ass:wn} holds, and in addition the innovations $\{\mathbf{U}_t\}$ are independent and identically distributed with $\mathbb{E}\|\mathbf{U}_t\|_F^4 < \infty$.
\end{assumption}

\begin{assumption}[Non-degeneracy]\label{ass:nondegen}
The impulse response satisfies $\nabla g_{ik,h}(\boldsymbol{\zeta}) \neq \mathbf{0}$.
\end{assumption}
Assumption \ref{ass:iid-plus} adds two conditions to Assumption \ref{ass:wn}, independence and finite fourth moments.
Together with Assumption \ref{ass:stationary}, independence makes $\mathbf{y}_t$ stationary and ergodic and makes the score a martingale difference, as required for the asymptotic theory in Section \ref{subsec:asymptotic}.
It is also the property that the residual bootstrap in Section \ref{subsec:bootstrap} mimics, since the bootstrap innovations are drawn independently from the residuals (see \citealp{bruggemann_inference_2016} for the failure of such schemes under conditional heteroskedasticity).
Finite fourth moments are needed because the covariance block of the score is quadratic in $\mathbf{U}_t$, so that the asymptotic variance of $\widehat{\boldsymbol{\zeta}}$ is finite.

Assumption \ref{ass:nondegen} is not needed for the asymptotic normality of $\widehat{\boldsymbol{\zeta}}$ in Proposition \ref{prop:an}, but only for the resulting impulse response intervals, both from the delta method and from the bootstrap.
It is stated for a single response because the asymptotic variance $\boldsymbol{\Omega}$ of $\widehat{\boldsymbol{\zeta}}$ is singular, with null space spanned by the gradients of the $p+1$ normalizations.
Since impulse responses are invariant to the same rescalings, $\nabla g_{ik,h}(\boldsymbol{\zeta})$ can only lie in this null space if it vanishes.
Assumption \ref{ass:nondegen} therefore guarantees $\sigma_{ik,h}^2 = \nabla g_{ik,h}(\boldsymbol{\zeta})^\top \boldsymbol{\Omega} \nabla g_{ik,h}(\boldsymbol{\zeta}) > 0$, the condition under which percentile intervals retain their coverage in unrestricted VARs \citep{benkwitz_problems_2000, inoue_bootstrapping_2002}.
For a fixed linear combination of responses, such as the joint shock of Section \ref{subsec:ident}, the condition is imposed on the gradient of the combination.

\subsection{Asymptotic Inference}\label{subsec:asymptotic}
The identified impulse response $\theta_{ik,h} = g_{ik,h}(\boldsymbol{\zeta})$ depends on both the coefficients $\boldsymbol{\beta}$ and the covariance parameters $\boldsymbol{\sigma}$, so its inference requires the joint asymptotic distribution of their estimators.
We first derive this joint distribution and then apply the delta method to obtain closed-form standard errors at any fixed horizon.
Proofs are in Appendix \ref{app:asymptotic}.

Existing asymptotic results for matrix and tensor autoregressions concern the coefficient estimators only \citep{chen_autoregressive_2021, li_multilinear_2021}.
Proposition \ref{prop:an} extends these results to the joint distribution of the coefficient and covariance estimators of the MAR($p$), obtained by maximizing the Gaussian log-likelihood \eqref{eq:loglik} as described in Section \ref{subsec:MAR}.
Since Assumption \ref{ass:iid-plus} does not require Gaussian innovations, $\widehat{\boldsymbol{\zeta}}$ is a quasi-maximum likelihood estimator, and the asymptotic variance of the covariance estimators depends on the fourth moments of $\mathbf{U}_t$ rather than on $\boldsymbol{\Sigma}$ alone.

\begin{proposition}\label{prop:an}
    Let Assumptions \ref{ass:stationary} and \ref{ass:iid-plus} hold and let $\widehat{\boldsymbol{\zeta}}$ maximize \eqref{eq:loglik} subject to the normalization requirements.
    Then
    \begin{equation}\label{eq:an}
        \sqrt{T}\big(\widehat{\boldsymbol{\zeta}} - \boldsymbol{\zeta}\big)
        \xrightarrow{d} N\big(\mathbf{0}, \boldsymbol{\Omega}\big),
        \qquad
        \boldsymbol{\Omega} = \operatorname{blkdiag}\big(\boldsymbol{\Omega}_\beta, \boldsymbol{\Omega}_\sigma\big),
    \end{equation}
    with $\boldsymbol{\Omega}_\beta$ and $\boldsymbol{\Omega}_\sigma$ given explicitly in Appendix \ref{app:asymptotic}. 
\end{proposition}

The result for the  coefficient block is not new.
\citet{chen_autoregressive_2021} obtain it for the MAR($1$); for $p = 1$, the block $\boldsymbol{\Omega}_\beta$ is precisely that of Theorem 4 of \citet{chen_autoregressive_2021}.
\citet{li_multilinear_2021} obtain it at general lag order, as the two-way, single-term case of a tensor autoregression.
What Proposition \ref{prop:an} adds is the covariance block and the joint distribution, which is what an identified impulse response requires \citep{bruggemann_inference_2016}.
Equation \eqref{eq:an} holds under the fourth moments of Assumption \ref{ass:iid-plus}, whereas the maximum likelihood limit of \citet{li_multilinear_2021} is derived under Gaussian innovations.

Since the impulse responses depend on the coefficients only through the restricted VAR coefficients $\mathbf{C}$ in \eqref{eq:restrictedvar}, the Kronecker restriction on the coefficients enters impulse response inference only through the asymptotic distribution of $\widehat{\mathbf{C}}$.
With $\mathbf{J}_\beta = \partial \operatorname{vec}(\mathbf{C}) / \partial \boldsymbol{\beta}^\top$ the Jacobian of the mapping from the MAR coefficients to the restricted VAR coefficients, derived in Appendix \ref{app:asymptotic}, Proposition \ref{prop:an} and the delta method give
\begin{equation}\label{eq:OmegaC}
    \sqrt{T}\big(\operatorname{vec}(\widehat{\mathbf{C}}) - \operatorname{vec}(\mathbf{C})\big)
    \xrightarrow{d} N\big(\mathbf{0}, \boldsymbol{\Omega}_C\big),
    \qquad
    \boldsymbol{\Omega}_C = \mathbf{J}_\beta \boldsymbol{\Omega}_\beta \mathbf{J}_\beta^\top .
\end{equation}
Given $\boldsymbol{\Omega}_C$, the propagation part of the delta method coincides with that of an unrestricted VAR, with $\boldsymbol{\Omega}_C$ in place of the unrestricted asymptotic variance.

Two features distinguish $\boldsymbol{\Omega}_C$ from the asymptotic variance of an unrestricted VAR($p$).
First, it is singular, since $\mathbf{J}_\beta$ maps $\boldsymbol{\beta} \in \mathbb{R}^{p(N_1^2 + N_2^2)}$ into $\operatorname{vec}(\mathbf{C}) \in \mathbb{R}^{p(N_1 N_2)^2}$ and has rank at most $p(N_1^2 + N_2^2 - 1)$, which is smaller than $p(N_1 N_2)^2$ for $N_1, N_2 \geq 2$.
The Kronecker restriction thus appears as a rank deficiency in the coefficient space of the VAR($p$), where the corresponding asymptotic variance is nonsingular under Assumptions \ref{ass:stationary} and \ref{ass:iid-plus}.
Second, it does not depend on the normalizations imposed in Section \ref{subsec:MAR}, since $\mathbf{C}$ does not.
The $p$ scale directions between $\mathbf{A}_j$ and $\mathbf{B}_j$ along which the likelihood is flat lie in the null space of $\mathbf{J}_\beta$, so whatever they contribute to $\boldsymbol{\Omega}_\beta$ drops out of \eqref{eq:OmegaC}.

Given \eqref{eq:OmegaC} and the covariance block of Proposition \ref{prop:an}, inference on impulse responses is standard.
The response $\theta_{ik,h} = g_{ik,h}(\boldsymbol{\zeta})$ is differentiable in both blocks, through the recursion \eqref{eq:mma} in $\mathbf{C}$ and through the impact matrix in $\boldsymbol{\sigma}$, so the delta method gives, for any fixed $h \geq 0$,
\begin{equation}\label{eq:irf_an}
    \sqrt{T}\big(\widehat{\theta}_{ik,h} - \theta_{ik,h}\big)
    \xrightarrow{d} N\big(0, \sigma_{ik,h}^2\big),
\end{equation}
with $\sigma_{ik,h}^2$ the quantity assumed positive in Assumption \ref{ass:nondegen}.
Since $\boldsymbol{\Omega}$ is block diagonal, $\sigma_{ik,h}^2$ separates into a term in $\boldsymbol{\Omega}_C$, the uncertainty in how the shock propagates, and a term in $\boldsymbol{\Omega}_\sigma$, the uncertainty in its impact.
The two Jacobians involved are those of the unrestricted VAR \citep[Prop. 1]{lutkepohl_asymptotic_1990} and Appendix \ref{app:asymptotic} records the form they take under the separable covariance of the MAR.
The reduced-form case $\mathbf{P} = \mathbf{I}_{N_1 N_2}$ leaves only the first term, while at $h = 0$ only the second survives, since $\boldsymbol{\Theta}_0 = \mathbf{I}_{N_1 N_2}$ does not depend on $\mathbf{C}$.
Both are therefore required of any interval to be compared with the bootstrap procedures of Section \ref{subsec:bootstrap}, which account for the two sources jointly.
For a significance level $\alpha \in (0,1)$, the delta-method confidence interval for $\theta_{ik,h}$ is
\begin{equation*}
    \left[ \widehat{\theta}_{ik,h} - z_{1-\alpha/2}\frac{\widehat{\sigma}_{ik,h}}{\sqrt{T}}, \widehat{\theta}_{ik,h} + z_{1-\alpha/2}\frac{\widehat{\sigma}_{ik,h}}{\sqrt{T}} \right],
\end{equation*}
where $z_{1-\alpha/2}$ denotes the $(1-\alpha/2)$ quantile of the standard normal distribution.
Its asymptotic coverage is $1-\alpha$ by \eqref{eq:irf_an}.

\subsection{Bootstrap Inference}\label{subsec:bootstrap}

The delta-method intervals of Section \ref{subsec:asymptotic} can undercover in finite samples.
A well-known source of this distortion is the finite-sample bias of autoregressive coefficient estimators, which is of order $T^{-1}$ in unrestricted VARs but can be substantial when $T$ is small or the process is persistent \citep{pope_biases_1990}.
The bias propagates into the estimated impulse responses, typically understating their persistence, so that intervals built around them cover the true response less often than the nominal level.

For unrestricted VARs, \citet{kilian_small_1998} proposed the bootstrap-after-bootstrap (BAB), which uses a first-stage bootstrap to estimate and correct the bias before constructing the bootstrap distribution in a second stage.
Simulation evidence indicates that the resulting percentile intervals provide reliable coverage for VAR impulse responses in small samples \citep{kilian_structural_2017}, though the accuracy deteriorates in large systems \citep{kilian_accurate_2000} and near the boundary of the stationary region \citep{benkwitz_problems_2000, inoue_uniform_2020}.

Adapting the BAB to the MAR setting introduces a complication that does not arise in the unrestricted case.
The MAR estimator produces coefficient matrices with Kronecker structure ($\widehat{\mathbf{C}}_j = \widehat{\mathbf{B}}_j \otimes \widehat{\mathbf{A}}_j$).
However, there is no guarantee that the estimated bias $\widehat{\boldsymbol{\Psi}}_j$ shares this structure, so the bias-corrected matrix $\widehat{\mathbf{C}}_j - \widehat{\boldsymbol{\Psi}}_j$ need not be a Kronecker product.
The naive response is to accept this departure and bootstrap from the bias-corrected coefficients directly, which is the procedure we label BAB-MAR and relegate to Appendix \ref{app:bab}.
We instead project the bias-corrected coefficients back onto the space of Kronecker products, so that the bootstrap design and the estimated model coincide by construction.

\subsubsection{Projected Bootstrap after Bootstrap}\label{subsec:probab}

ProBAB-MAR retains the two-stage structure of the BAB, with one modification.
After bias correction, each lag-specific block is projected back onto the space of Kronecker products before being used as the bootstrap DGP.
This preserves the Kronecker structure throughout both stages of the procedure.

Specifically, define $\mathcal{P}(\cdot)$ as the nearest Kronecker product projection \citep{van_loan_approximation_1993, van_loan_kronecker_2000}, given by
\begin{equation}\label{eq:proj}
    \mathcal{P}(\mathbf{C}_j) = 
    \underset{\mathbf{V} \otimes \mathbf{U}}{\arg\min}
    \|\mathbf{C}_j - \mathbf{V} \otimes \mathbf{U}\|_F^2,
    \qquad \mathbf{U} \in \mathbb{R}^{N_1 \times N_1}, \; \mathbf{V} \in \mathbb{R}^{N_2 \times N_2},
\end{equation}
and let $\mathcal{P}(\mathbf{C}) = [\mathcal{P}(\mathbf{C}_1) \cdots \mathcal{P}(\mathbf{C}_p)]$ denote its blockwise application.

\begin{algorithm}[t]
\caption{ProBAB-MAR}\label{alg:probab}
\small
\begin{algorithmic}[1]
\Require Data $\{\mathbf{Y}_t\}_{t=-p+1}^{T}$, lag order $p$, replications $B_1, B_2$, significance level $\alpha$
\State Fit the MAR to obtain $\widehat{\mathbf{C}}_j = \widehat{\mathbf{B}}_j \otimes \widehat{\mathbf{A}}_j$ and $\widehat{\mathbf{u}}_t = \mathbf{y}_t - \sum_j \widehat{\mathbf{C}}_j \mathbf{y}_{t-j}$
\Statex \textit{Stage 1: bias-corrected DGP}
\For{$m = 1, \ldots, B_1$}
  \State Draw $\mathbf{u}^{*}_{t}$ from $\widehat{\mathbf{u}}_t$ and generate $\mathbf{y}^{*}_t$ from $\widehat{\mathbf{C}}$ and $\mathbf{u}^{*}_{t}$, with $\mathbf{y}^{*}_{-p+1} = \cdots = \mathbf{y}^{*}_{0} = \mathbf{0}$
  \State Re-fit the MAR to obtain $\widehat{\mathbf{C}}^{*(m)}$
\EndFor
\State Estimate the bias by $\widehat{\boldsymbol{\Psi}} \gets \bar{\mathbf{C}}^{*} - \widehat{\mathbf{C}}$, where $\bar{\mathbf{C}}^{*} = B_1^{-1}\sum_{m} \widehat{\mathbf{C}}^{*(m)}$
\State Set $\delta \gets \max\{d \in \mathcal{G} : \rho(\mathcal{P}(\widehat{\mathbf{C}} - d\,\widehat{\boldsymbol{\Psi}})) < 1\}$ and $\widetilde{\mathbf{C}} \gets \mathcal{P}(\widehat{\mathbf{C}} - \delta\widehat{\boldsymbol{\Psi}})$ \label{lin:probab-delta}
\Statex \textit{Stage 2: bootstrap distribution}
\For{$m = 1, \ldots, B_2$}
  \State Draw $\mathbf{u}^{*}_{t}$ from $\widehat{\mathbf{u}}_t$ and generate $\mathbf{y}^{*}_t$ from $\widetilde{\mathbf{C}}$ and $\mathbf{u}^{*}_{t}$, with $\mathbf{y}^{*}_{-p+1} = \cdots = \mathbf{y}^{*}_{0} = \mathbf{0}$
  \State Re-fit the MAR to obtain $\widehat{\mathbf{C}}^{*}$ and compute $\delta^{*}$ as in line \ref{lin:probab-delta} \label{lin:probab-refit}
\State Set $\widetilde{\mathbf{C}}^{*} \gets \mathcal{P}(\widehat{\mathbf{C}}^{*} - \delta^{*}\widehat{\boldsymbol{\Psi}})$, re-estimate $\widetilde{\boldsymbol{\sigma}}^{*}$ from the residuals $\mathbf{y}^{*}_t - \sum_j \widetilde{\mathbf{C}}^{*}_j \mathbf{y}^{*}_{t-j}$, and set $\widetilde{\mathbf{P}}^{*} = \mathbf{P}(\widetilde{\boldsymbol{\sigma}}^{*})$
\State Compute $\widetilde{\theta}^{*(m)}_{ik,h} = \mathbf{e}_i^\top \widetilde{\boldsymbol{\Theta}}^{*}_h \widetilde{\mathbf{P}}^{*} \mathbf{e}_k$, with $\widetilde{\boldsymbol{\Theta}}^{*}_h$ obtained from $\widetilde{\mathbf{C}}^{*}$ via \eqref{eq:mma}
\EndFor
\State \Return percentile interval from $\{\widetilde{\theta}^{*(m)}_{ik,h}\}_{m=1}^{B_2}$
\end{algorithmic}
\medskip
{\footnotesize \textit{Notes:} $\mathcal{G} = \{1, 0.99, \ldots, 0.01\}$, and $\delta = 0$ if no $d \in \mathcal{G}$ qualifies, in which case $\widetilde{\mathbf{C}} = \widehat{\mathbf{C}}$. $\rho(\mathbf{C})$ denotes the spectral radius of the companion matrix of $\mathbf{C}$, and $\mathcal{P}(\cdot)$ is the nearest Kronecker product projection \eqref{eq:proj}, applied to each lag coefficient.}
\end{algorithm}

\subsubsection*{Stage 1: Bias-corrected DGP}

Generate $B_1$ bootstrap samples $\{\mathbf{y}_t^{*}\}_{t=1}^T$ by simulating from the estimated model with $\widehat{\mathbf{C}}$ and the innovations drawn with replacement from the residuals $\{\widehat{\mathbf{u}}_t\}_{t=1}^T$.
Re-estimate the MAR on each bootstrap sample to obtain $\widehat{\mathbf{C}}^{*(m)}$, and estimate the bias as $\widehat{\boldsymbol{\Psi}} = \bar{\mathbf{C}}^{*} - \widehat{\mathbf{C}}$, where $\bar{\mathbf{C}}^{*} = B_1^{-1} \sum_{m=1}^{B_1} \widehat{\mathbf{C}}^{*(m)}$.
The projected bias-corrected coefficients are then
\begin{equation}\label{eq:bc}
    \widetilde{\mathbf{C}}_j = \mathcal{P}\!\left(\widehat{\mathbf{C}}_j - \delta\widehat{\boldsymbol{\Psi}}_j\right)
    = \widetilde{\mathbf{B}}_j \otimes \widetilde{\mathbf{A}}_j, \qquad j = 1, \ldots, p,
\end{equation}
where $\widehat{\boldsymbol{\Psi}}_j$ denotes the block of $\widehat{\boldsymbol{\Psi}}$ corresponding to lag $j$, and $\delta \in [0,1]$ shrinks the correction to keep the bootstrap DGP stationary.
As in the BAB, $\delta$ is the largest value on the grid $\{1, 0.99, \ldots, 0.01\}$ for which the matrix actually used to generate the bootstrap samples has all companion eigenvalues strictly inside the unit circle.
The order in which the projection and this stationarity adjustment are applied matters.
The projection minimizes Frobenius distance to the space of Kronecker products and places no restriction on the eigenvalues, so it may induce a companion root outside the unit circle even when the unprojected bias-corrected matrix has none.
Since it is the projected matrix that generates the bootstrap samples, it is the projected matrix whose eigenvalues must be checked, and $\delta$ in \eqref{eq:bc} is therefore chosen with the projection already applied.
If no grid value qualifies, the correction is not applied and $\widetilde{\mathbf{C}} = \widehat{\mathbf{C}}$.
Since the estimator does not impose stationarity, the resulting bootstrap DGP may itself be explosive; Section \ref{subsec:sim_persistence} reports how often this occurs.

\subsubsection*{Stage 2: Bootstrap distribution}

The second stage generates $B_2$ samples $\{\mathbf{y}_t^{*}\}_{t=1}^T$ recursively from the projected bias-corrected DGP $\widetilde{\mathbf{C}}$, with innovations $\mathbf{u}_t^{*}$ drawn with replacement from the residuals of the original fit, as in Stage 1, and initial values set to $\mathbf{y}_{-p+1}^{*} = \cdots = \mathbf{y}_{0}^{*} = \mathbf{0}$.
Re-estimating the MAR on each sample yields $\widehat{\mathbf{C}}^{*}$, which, as an estimator of the bootstrap DGP $\widetilde{\mathbf{C}}$, is subject to the same finite-sample bias as $\widehat{\mathbf{C}}$.
The correction is accordingly applied a second time within each replication, giving $\widetilde{\mathbf{C}}^{*} = \mathcal{P}(\widehat{\mathbf{C}}^{*} - \delta^{*}\widehat{\boldsymbol{\Psi}})$, with the projection again preceding the eigenvalue check that determines $\delta^{*}$.

Because each lag block $\widetilde{\mathbf{C}}_j^{*}$ again factors as $\widetilde{\mathbf{B}}_j^{*} \otimes \widetilde{\mathbf{A}}_j^{*}$, the bootstrap impulse responses $\widetilde{\boldsymbol{\Theta}}_h^{*}$ obtained through the recursion \eqref{eq:mma} are those of a MAR, and inherit the structure of Section \ref{subsec:irf}.
In each replication, the bootstrap response is $\widetilde{\theta}_{ik,h}^{*} = \mathbf{e}_i^\top \widetilde{\boldsymbol{\Theta}}_h^{*} \widetilde{\mathbf{P}}^{*} \mathbf{e}_k$, where $\widetilde{\boldsymbol{\Theta}}_h^{*}$ follows from $\widetilde{\mathbf{C}}^{*}$ through \eqref{eq:mma} and $\widetilde{\mathbf{P}}^{*} = \mathbf{P}(\widetilde{\boldsymbol{\sigma}}^{*})$ applies the identification scheme to the covariance re-estimated from the residuals of $\widetilde{\mathbf{C}}^{*}$, so that the bootstrap distribution reflects uncertainty in the impact as well as in the propagation of the shock.
The $(1-\alpha)$ percentile confidence interval for $\theta_{ik,h}$ in \eqref{eq:theta_def} is formed from the $\alpha/2$ and $1-\alpha/2$ quantiles of $\widetilde{\theta}_{ik,h}^{*}$ across the $B_2$ replications.

Finally, note that strictly, the bias of $\widehat{\mathbf{C}}^{*}$ relative to $\widetilde{\mathbf{C}}$ has to be estimated by a further bootstrap loop nested inside each replication.
As in the BAB, we instead reuse the Stage 1 estimate $\widehat{\boldsymbol{\Psi}}$ as a proxy, recomputing only the shrinkage factor $\delta^{*}$, so that the procedure requires $B_1 + B_2$ rather than $B_1 B_2$ re-estimations of the model.

\subsubsection{Asymptotic Validity}
\label{subsec:validity}

We establish that ProBAB-MAR yields intervals with asymptotically correct coverage.
First, we show that the residual bootstrap is consistent for the MAR, that is, resampling residuals and re-estimating reproduces the joint limiting distribution of Proposition \ref{prop:an}.
Second, we show that the three modifications ProBAB-MAR makes to that scheme, namely bias-correcting the bootstrap DGP, projecting it back onto the Kronecker class, and shrinking the correction to preserve stationarity, leave the first-order behavior unchanged.
We take these in turn.
Throughout, $P^{*}$ and $\mathbb{E}^{*}$ denote probability and expectation conditional on the sample $\mathbf{y}_{-p+1}, \ldots, \mathbf{y}_T$, and $\xrightarrow{d^{*}}$, $o_{p^{*}}(\cdot)$ and $O_{p^{*}}(\cdot)$ denote weak convergence and stochastic orders under $P^{*}$, each holding in probability.

\subsubsection*{Asymptotic validity of the residual bootstrap}

Consider first the simpler scheme in which bootstrap samples are generated from $\widehat{\mathbf{C}}$ itself, with innovations drawn with replacement from the residuals, and let $\widehat{\boldsymbol{\zeta}}^{*}$ denote the estimator recomputed on such a sample.

\begin{proposition}[Residual bootstrap consistency]\label{prop:boot}
Let Assumptions \ref{ass:stationary} and \ref{ass:iid-plus} hold, and let $\widehat{\boldsymbol{\zeta}}^{*}$ maximize the log-likelihood \eqref{eq:loglik} on the bootstrap sample, subject to the normalizations of Section \ref{subsec:MAR}. Then
\begin{equation}\label{eq:boot_zeta}
    \sqrt{T}\big(\widehat{\boldsymbol{\zeta}}^{*} - \widehat{\boldsymbol{\zeta}}\big)
    \xrightarrow{d^{*}} N\big(\mathbf{0}, \boldsymbol{\Omega}\big)
    \qquad \text{in probability.}
\end{equation}
\end{proposition}

The argument is the bootstrap-world counterpart of Proposition \ref{prop:an} and is given in Appendix \ref{app:bootstrap}.
Conditional on the sample, the innovations $\mathbf{u}_t^{*}$ are i.i.d.\ and the coefficients generating $\{\mathbf{y}_t^{*}\}$ are fixed, so the score decomposition of Appendix \ref{app:asymptotic} carries over with $\mathbb{E}^{*}$ in place of $\mathbb{E}$, and the rate $\widehat{\boldsymbol{\zeta}}^{*} - \widehat{\boldsymbol{\zeta}} = O_{p^{*}}(T^{-1/2})$ follows from the convexity argument of the proof of Theorem 4 of \citet{chen_autoregressive_2021} applied conditionally.
The bootstrap process is not, however, stationary and ergodic.
The resampling distribution and the coefficients both depend on $T$, and the initial block is set to zero rather than from the stationary distribution, so $\{\mathbf{y}_t^{*}\}$ is a triangular array.
Two steps of Appendix \ref{app:asymptotic} therefore require separate treatment.
First, the ergodic theorem is replaced by a law of large numbers for $T^{-1}\sum_t \mathbf{x}_{t}^{*}\mathbf{x}_{t}^{*\top}$, obtained from the moving average representation of the bootstrap process and the geometric decay of its coefficients, which holds uniformly on an event of probability approaching one.
Second, the martingale central limit theorem of \citet{billingsley_lindeberg-levy_1961}, which presumes stationarity, is replaced by its triangular-array counterpart \citep{brown_martingale_1971}, whose conditional variance and Lindeberg conditions are verified using the fourth moments of Assumption \ref{ass:iid-plus}.
The coefficient block of the score is a martingale difference because $\mathbf{W}_t^{*}$ is $\mathcal{F}_{t-1}^{*}$-measurable, as in the sampling world; the covariance block, which carries no such factor, is centered instead by the first-order condition $T^{-1}\sum_t \widehat{\mathbf{U}}_t \widehat{\boldsymbol{\Sigma}}_2^{-1}\widehat{\mathbf{U}}_t^\top = N_2 \widehat{\boldsymbol{\Sigma}}_1$ satisfied by the estimator, up to a demeaning residual of order $O_p(T^{-1})$ that is negligible after scaling.
The bootstrap Hessian converges to $\mathbf{H}$ and the sandwich $\mathbf{H}^{-1}\boldsymbol{\Lambda}\mathbf{H}^{-1}$ is recovered.
Proposition \ref{prop:boot} is stated as weak convergence because $\boldsymbol{\Omega}$ is singular, so the joint limiting distribution function need not be continuous.
This does not affect the impulse responses, whose limiting variance is positive under Assumption \ref{ass:nondegen}.

\subsubsection*{ProBAB-MAR modifications are asymptotically negligible}
ProBAB-MAR departs from the above described simplified scheme in three ways, none of which affects \eqref{eq:boot_zeta}.
Throughout, write
\begin{equation*}
    \widetilde{\mathbf{C}}(\delta) = \mathcal{P}\big(\widehat{\mathbf{C}} - \delta\widehat{\boldsymbol{\Psi}}\big),
    \qquad \delta \in [0,1],
\end{equation*}
for the candidate bootstrap DGP at shrinkage factor $\delta$, so that $\widetilde{\mathbf{C}}$ of \eqref{eq:bc} is $\widetilde{\mathbf{C}}(\delta)$ at the selected value and $\widetilde{\mathbf{C}}(0) = \widehat{\mathbf{C}}$.

First, Stage 2 generates from $\widetilde{\mathbf{C}}$ rather than from $\widehat{\mathbf{C}}$.
The estimated bias satisfies $\widehat{\boldsymbol{\Psi}} = O_p(T^{-1})$ \citep{pope_biases_1990}, while the estimator converges at rate $T^{-1/2}$.
The projection does not undo this gap.
The rearrangement of $\mathbf{C}_j$ is exactly rank one, so it carries a single leading singular value and is bounded away from zero whenever $\mathbf{A}_j$ and $\mathbf{B}_j$ are nonzero.
Since $\widehat{\mathbf{C}}$ lies on the Kronecker manifold, $\mathcal{P}(\widehat{\mathbf{C}}) = \widehat{\mathbf{C}}$, and therefore
\begin{equation}\label{eq:dgp_perturb}
    \widetilde{\mathbf{C}} - \widehat{\mathbf{C}}
    = \mathcal{P}\big(\widehat{\mathbf{C}} - \delta\widehat{\boldsymbol{\Psi}}\big) - \mathcal{P}\big(\widehat{\mathbf{C}}\big)
    = O_p(T^{-1}).
\end{equation}
The bootstrap DGP is thus perturbed by an amount an order of magnitude smaller than the sampling error it is meant to approximate, and conditioning on $\widetilde{\mathbf{C}}$ as the bootstrap population value leaves \eqref{eq:boot_zeta} intact.

Second, the same correction is applied within each bootstrap replication, with $\widehat{\boldsymbol{\Psi}}$ reused in place of the nested-bootstrap estimate that this nominally requires.
The two differ by $o_p(T^{-1/2})$ (see Appendix \ref{app:proofs}), which is negligible compared to the $O_p(T^{-1})$ correction itself, so replacing $\widehat{\mathbf{C}}^{*}$ by $\widetilde{\mathbf{C}}^{*}$ leaves \eqref{eq:boot_zeta} intact by the argument of \eqref{eq:dgp_perturb} applied conditionally.

Third, $\delta$ shrinks the correction whenever the bootstrap DGP would otherwise be nonstationary.
Write $\rho(\cdot)$ for the modulus of the largest companion eigenvalue.
Proposition \ref{prop:an} gives $\widehat{\mathbf{C}} \xrightarrow{p} \mathbf{C}$, and $\widetilde{\mathbf{C}}(1) - \widehat{\mathbf{C}} = O_p(T^{-1})$ by \eqref{eq:dgp_perturb}, so $\widetilde{\mathbf{C}}(1) \xrightarrow{p} \mathbf{C}$.
Eigenvalues are continuous in the coefficients, so $\rho(\widetilde{\mathbf{C}}(1)) \xrightarrow{p} \rho(\mathbf{C})$, and $\rho(\mathbf{C}) < 1$ under Assumption \ref{ass:stationary}.
Since $\delta = 1$ lies on the grid, the correction is left intact exactly when $\rho(\widetilde{\mathbf{C}}(1)) < 1$, so
\begin{equation*}
    P(\delta = 1) = P\big(\rho(\widetilde{\mathbf{C}}(1)) < 1\big)
    \; \ge \; P\Big(\big|\rho(\widetilde{\mathbf{C}}(1)) - \rho(\mathbf{C})\big| < 1 - \rho(\mathbf{C})\Big)
    \rightarrow 1 .
\end{equation*}
The adjustment therefore binds with probability approaching zero and leaves the limiting distribution unaffected.

\subsubsection*{Validity for impulse responses}
Since $\theta_{ik,h}=g_{ik,h}(\boldsymbol{\zeta})$ is continuously differentiable in both blocks, asymptotic validity for the impulse response now follows from \eqref{eq:boot_zeta} by the delta method, applied conditionally on the sample in the bootstrap world and unconditionally in the sampling world.
Let $\widetilde{\boldsymbol{\zeta}}^{*}$ collect $\operatorname{vec}(\widetilde{\mathbf{C}}^{*})$ and $\widetilde{\boldsymbol{\sigma}}^{*}$, and write $\widetilde{\theta}_{ik,h}^{*} = g_{ik,h}(\widetilde{\boldsymbol{\zeta}}^{*})$ for the response actually computed in Stage 2.
By the argument of the previous subsection, $\widetilde{\boldsymbol{\zeta}}^{*} - \widehat{\boldsymbol{\zeta}}^{*} = O_p(T^{-1})$, so the bias-corrected and projected draws obey \eqref{eq:boot_zeta} in place of $\widehat{\boldsymbol{\zeta}}^{*}$.
\begin{corollary}[Asymptotic validity of ProBAB-MAR]\label{thm:bootstrap}
Let Assumptions \ref{ass:stationary}--\ref{ass:nondegen} hold and fix $h \ge 0$. Then
\begin{equation}\label{eq:thm_kolmogorov}
    \sup_{x \in \mathbb{R}} \Big|
    P^{*} \Big(\sqrt{T} \big(g_{ik,h}(\widetilde{\boldsymbol{\zeta}}^{*})-g_{ik,h}(\widehat{\boldsymbol{\zeta}})\big)\le x\Big)
    -P \Big(\sqrt{T} \big(g_{ik,h}(\widehat{\boldsymbol{\zeta}})-g_{ik,h}(\boldsymbol{\zeta})\big)\le x\Big)\Big|
    \rightarrow 0
\end{equation}
in probability, and the percentile interval satisfies
\begin{equation}\label{eq:coverage}
    \lim_{T\to\infty}P\Big(\theta_{ik,h}\in
    \big[\widetilde{\theta}_{ik,h}^{*(\alpha/2)},\widetilde{\theta}_{ik,h}^{*(1-\alpha/2)}\big]\Big)=1-\alpha.
\end{equation}
\end{corollary}
The statistic in \eqref{eq:thm_kolmogorov} is centered at $\theta_{ik,h}(\widehat{\boldsymbol{\zeta}})$ rather than at the bootstrap population value $\theta_{ik,h}(\widetilde{\boldsymbol{\zeta}})$.
The two differ by $O_p(T^{-1})$, so the choice is immaterial after scaling by $\sqrt{T}$, and centering at the estimate delivers \eqref{eq:coverage} directly.
Both distribution functions in \eqref{eq:thm_kolmogorov} converge to $N(0,\sigma^{2}_{ik,h})$, which is continuous and strictly increasing because Assumption \ref{ass:nondegen} holds, so the convergence is uniform and the empirical quantiles of the bootstrap draws converge to those of the sampling distribution \citep[Lemma 21.2]{vaart_asymptotic_1998}.

Whether the corollary has content at $h = 0$ depends on the identification scheme.
The reduced-form response $\boldsymbol{\Theta}_0 = \mathbf{I}_{N_1N_2}$ is a known constant, so $\sigma^2_{ik,0} = 0$ and Assumption~\ref{ass:nondegen} excludes it.
Under a scheme in which the impact matrix is a smooth, nonconstant function of $\boldsymbol{\sigma}$, such as a Cholesky factorization or the generalized impulse response of \citet{chen_autoregressive_2021}, the $h = 0$ response inherits uncertainty from the covariance estimate and the assumption is restored, with the exception of the elements that the scheme sets to zero by construction, which remain constants.

\section{Monte Carlo Simulations}
\label{sec:simulations}

We assess the finite-sample coverage of the proposed procedures across three data-generating processes (DGPs): one calibrated to the empirical application, and two based on randomly generated Kronecker coefficients that vary the sample size and the persistence of the process.

\subsection{Simulation Design}
\label{subsec:sim_design}

We compare seven methods for constructing confidence intervals for the impulse responses, organized into two groups according to whether they impose the Kronecker structure of the MAR or treat the system as an unrestricted VAR.
The MAR group contains four intervals: (i) the delta-method intervals of Section \ref{subsec:asymptotic}, centered at the raw estimate $\widehat{\mathbf{C}}$; (ii) the same delta-method intervals recentered at the bias-corrected estimate $\widetilde{\mathbf{C}}$ of Stage 1; (iii) the BAB-MAR percentile intervals of Appendix \ref{app:bab}; and (iv) the ProBAB-MAR percentile intervals of Section \ref{subsec:probab}.
The VAR group comprises the three analogous intervals built from the unrestricted least-squares estimator: the delta-method intervals, their bias-corrected counterpart, and the BAB intervals of Section \ref{subsec:bootstrap}.
Comparing across the two groups isolates the role of the Kronecker restriction, while comparing within each group separates the contribution of bias correction from that of the bootstrap.

For each DGP we draw $R=1000$ independent samples of length $T$.
On every sample we estimate the MAR by maximum likelihood (MLE) and the VAR by ordinary least squares, both at the true lag order $p$, and construct the seven pointwise intervals at each horizon $h = 1, \ldots, H=15$.
The bootstrap procedures use $B_1=1000$ first-stage replications to estimate the bias and $B_2=2000$ second-stage replications to approximate the sampling distribution.

Since bias correction and the Kronecker projection act on the coefficients alone, we evaluate reduced-form responses $\boldsymbol{\Theta}_h\mathbf{e}_k$, which isolate the propagation of a shock from the uncertainty in its impact.
In each design we track a single response, indexing entries by their position $(r,s)$ in $\mathbf{Y}_t$ as in Section \ref{subsec:irf}, and evaluate the intervals on three criteria at every horizon, omitting $h = 0$ since $\boldsymbol{\Theta}_0 = \mathbf{I}_{N_1N_2}$ is known.
Bias is the average over the $R$ replications of the point estimate $\widehat{\theta}_{ik,h}$ minus its true value $\theta_{ik,h}$, and records the downward distortion that shifts impulse responses toward zero and that the bias-corrected procedures are designed to remove.
Empirical coverage is the fraction of replications in which the $95\%$ interval for $\theta_{ik,h}$ contains its true value.
Interval width, averaged across replications, measures efficiency.

\subsection{Empirically Calibrated DGP}
\label{subsec:sim_empirical}

The first DGP is calibrated to the empirical application of Section \ref{sec:application}, so that coverage is assessed under dynamics and error correlations that a practitioner could encounter.
The true coefficients are set to the MLE estimates of a MAR($2$) fitted to year-on-year HICP inflation for five euro area countries (Austria, Belgium, Germany, France, and the Netherlands) in three categories (energy, food, and core), giving $N_1 = 5$, $N_2 = 3$, and $p = 2$.
The innovations are drawn i.i.d.\ from a matrix normal distribution with row covariance $\widehat{\boldsymbol{\Sigma}}_1$ and column covariance $\widehat{\boldsymbol{\Sigma}}_2$ taken from the same fit, so that $\operatorname{Cov}(\operatorname{vec}(\mathbf{U}_t)) = \widehat{\boldsymbol{\Sigma}}_2 \otimes \widehat{\boldsymbol{\Sigma}}_1$ inherits the separable structure maintained by the MLE.
The sample length is set to $T = 200$, close to the empirical sample of $T = 205$, and each series is initialized at zero and simulated with a burn-in of $500$ observations that is subsequently discarded.
The estimated process is persistent, with the largest companion eigenvalue equal to $0.95$, and the fitted row covariance $\widehat{\boldsymbol{\Sigma}}_1$ implies strong contemporaneous correlation across countries, with pairwise correlations between $0.33$ and $0.66$.
The column covariance $\widehat{\boldsymbol{\Sigma}}_2$, by contrast, is close to diagonal, so that innovations are nearly uncorrelated across inflation categories but differ markedly in scale.
Both the MAR and the VAR are estimated at the true lag order $p = 2$.
We report the response of food inflation in the Netherlands, entry $(5,2)$, to a common unit innovation in energy inflation in all five countries, entries $(1,1)$ through $(5,1)$.

Figure \ref{fig:empirical_simulation} reports the results, with the bias of the point estimates, the empirical coverage of the nominal $95\%$ intervals, and the average interval width in the left, middle, and right panels, each across horizons $h = 1, \ldots, 15$.
The bias panel reveals the sharpest contrast between the two models in any of our designs.
With $450$ coefficients estimated from a sample of length $T = 200$, the unrestricted least-squares estimates carry a bias of $-0.03$ at worst (in the case of the VAR), and the first-stage correction reduces this to $-0.005$.
The MAR estimates, drawing on the Kronecker restriction that is correct by construction under this DGP, require only $66$ coefficients and are far less distorted, with a bias of $-0.02$ at worst before correction and close to zero after.

\begin{figure}[!t]
  \centering
  \includegraphics[width=\linewidth]{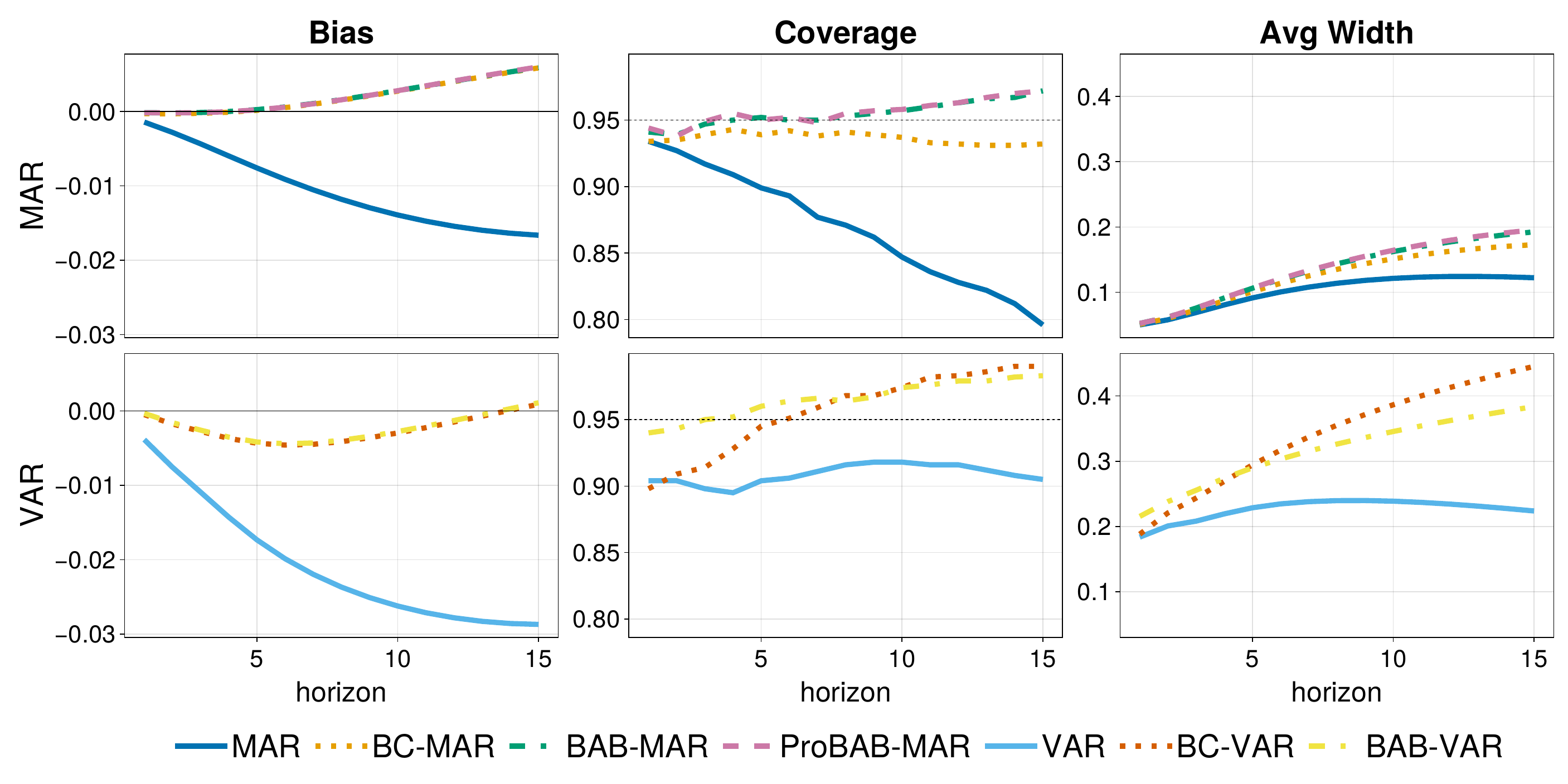}
  \caption{Bias, empirical coverage, and average width of nominal 95\% confidence intervals for the response of Dutch food inflation to a common energy shock under the empirically calibrated MAR(2) DGP ($N_1 = 5$, $N_2 = 3$, $T = 200$), with both models estimated at $p = 2$, across horizons $h = 1, \ldots, 15$.}
  \label{fig:empirical_simulation}
\end{figure}

The coverage panel shows the consequences.
BAB-MAR and ProBAB-MAR deliver near-nominal coverage at all horizons and are almost indistinguishable from one another.
The Kronecker projection that distinguishes them leaves the bootstrap distribution essentially unchanged even at this moderate sample size.
The delta-method intervals undercover, falling to $80\%$, and recentering them at the bias-corrected estimate $\widetilde{\mathbf{C}}$ recovers most of this gap.
The unrestricted BAB-VAR intervals also attain nominal coverage despite the residual bias left by the first-stage correction, and become conservative at longer horizons, where their coverage approaches $100\%$.
The difference between the two model classes therefore lies not in coverage but in the width required to attain it.

The width panel shows that the parsimony of the Kronecker restriction translates into efficiency.
The BAB-VAR intervals are, averaged across horizons, $2.57$ times as wide as their ProBAB-MAR counterparts, and their conservative coverage at longer horizons reflects this lack of precision rather than an advantage.
The MAR procedures therefore attain nominal coverage at a fraction of the width of the unrestricted benchmark.

\begin{figure}[!t]
  \centering
  \includegraphics[width=\linewidth]{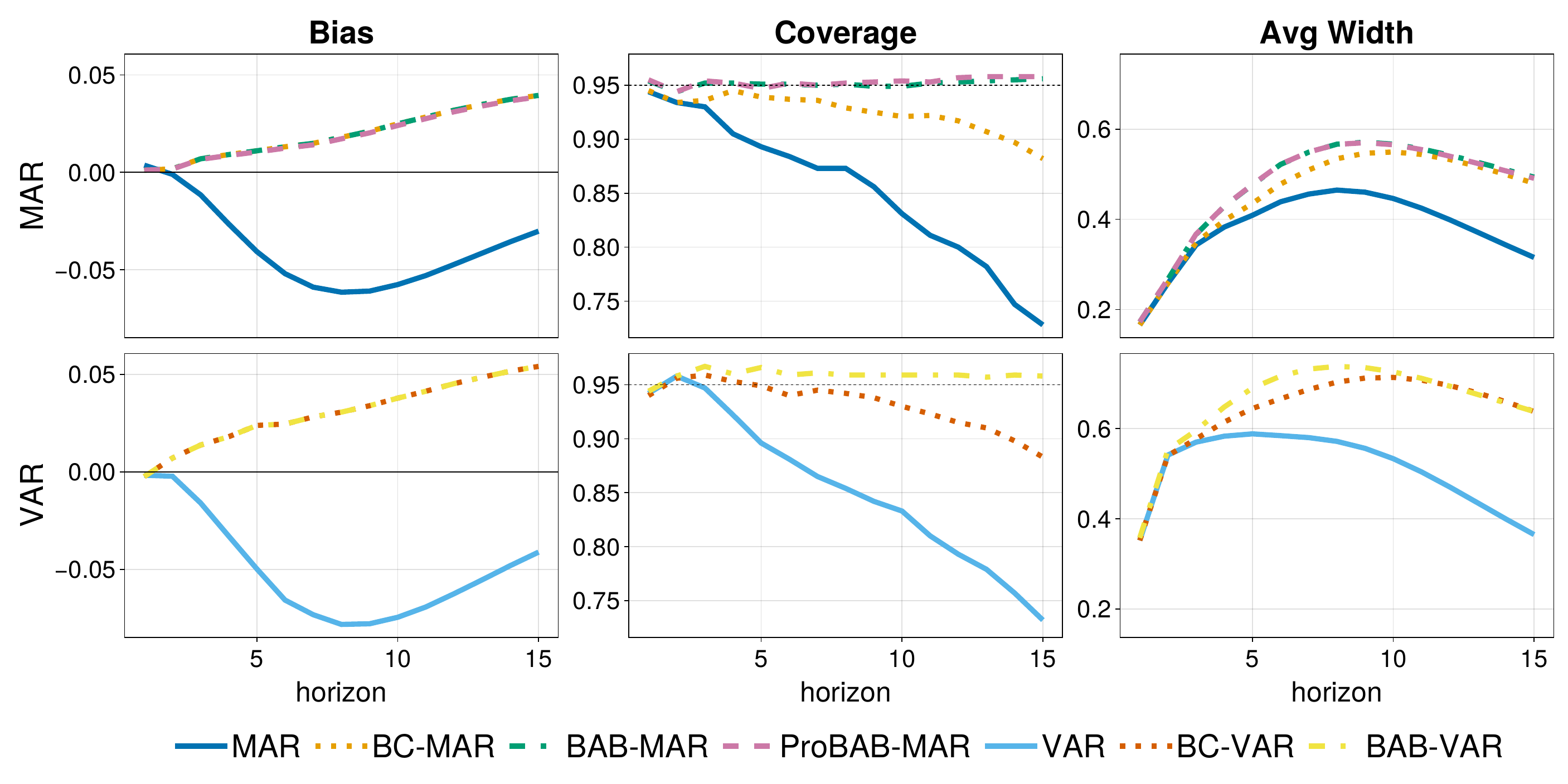}
  \caption{Bias, empirical coverage, and average width of nominal 95\% confidence intervals for impulse responses under the baseline MAR(1) DGP ($N_1 = 3$, $N_2 = 4$, $\rho_{\max} = 0.8$) with $T = 100$, across horizons $h = 1, \ldots, 15$.}
  \label{fig:baseline_obs100}
\end{figure}

\subsection{Coverage Across Sample Sizes}
\label{subsec:sim_samplesize}

The second design isolates the effect of the sample size.
The coefficients $\mathbf{C}_j = \mathbf{B}_j \otimes \mathbf{A}_j$ are randomly generated Kronecker products, rescaled so that the largest companion eigenvalue equals $\rho_{\max} = 0.8$, and the innovations are drawn i.i.d.\ from a matrix normal distribution with covariance $\operatorname{Cov}(\operatorname{vec}(\mathbf{U}_t)) = \sigma^2 \mathbf{I}_{N_1 N_2}$.
The covariance is thus trivially separable, $\boldsymbol{\Sigma} = \boldsymbol{\Sigma}_2 \otimes \boldsymbol{\Sigma}_1$, consistent with the structure maintained by the MLE.
We vary $T \in \{100, 200, 500\}$, holding the dimensions fixed at $N_1 = 3$, $N_2 = 4$, and $p = 1$.
We report the reduced-form response of entry $(2,4)$ to a unit shock in entry $(3,2)$.
Repeating the design for other shock and response entries leaves bias, coverage, and width essentially unchanged, so the results below are representative of the system as a whole.

Figure \ref{fig:baseline_obs100} reports the results for the smallest sample, $T = 100$. 
The bias panel shows the downward bias of Section \ref{subsec:bootstrap} at work in both models, with the raw estimates biased toward zero and the MAR estimates reaching $-0.06$ at the worst horizon.
The first-stage correction removes most of this distortion, leaving a small positive bias of $+0.01$ that rises to $+0.04$ at longer horizons, with the MAR procedures marginally more accurate than their VAR counterparts ($+0.03$ for BAB-MAR and ProBAB-MAR against $+0.05$ for BAB-VAR).

The coverage panel reveals that the delta-method intervals undercover severely, falling to $74\%$ at longer horizons for the MAR.
Recentering them at the bias-corrected point estimate recovers much of the gap but still leaves coverage at $88\%$.
Only the bootstrap-after-bootstrap procedures attain the nominal $95\%$ level at every horizon considered.
Accurate coverage, however, is only worth having if the intervals remain informative, and the width panel shows that the Kronecker restriction delivers exactly this efficiency.
BAB-MAR and ProBAB-MAR deliver intervals roughly two thirds as wide as BAB-VAR at every horizon, with average widths of $0.2$--$0.5$ against $0.3$--$0.7$ over the first five horizons, and the gap persists as the horizon grows.

The larger samples, reported in Figures \ref{fig:baseline_obs200} and \ref{fig:baseline_obs500} in Appendix \ref{app:additional_sims}, display similar results.
The bootstrap procedures remain at the nominal level throughout, while the coverage of the delta-method intervals rises from $74\%$ to $83\%$ at $T = 200$ and to $92\%$ at $T = 500$, with the bias of the point estimates and the interval widths shrinking accordingly.
Even at $T = 500$, a sample length rarely available in quarterly macroeconomic data, the delta-method intervals do not reach the nominal level.

\begin{figure}[!t]
  \centering
  \includegraphics[width=\linewidth]{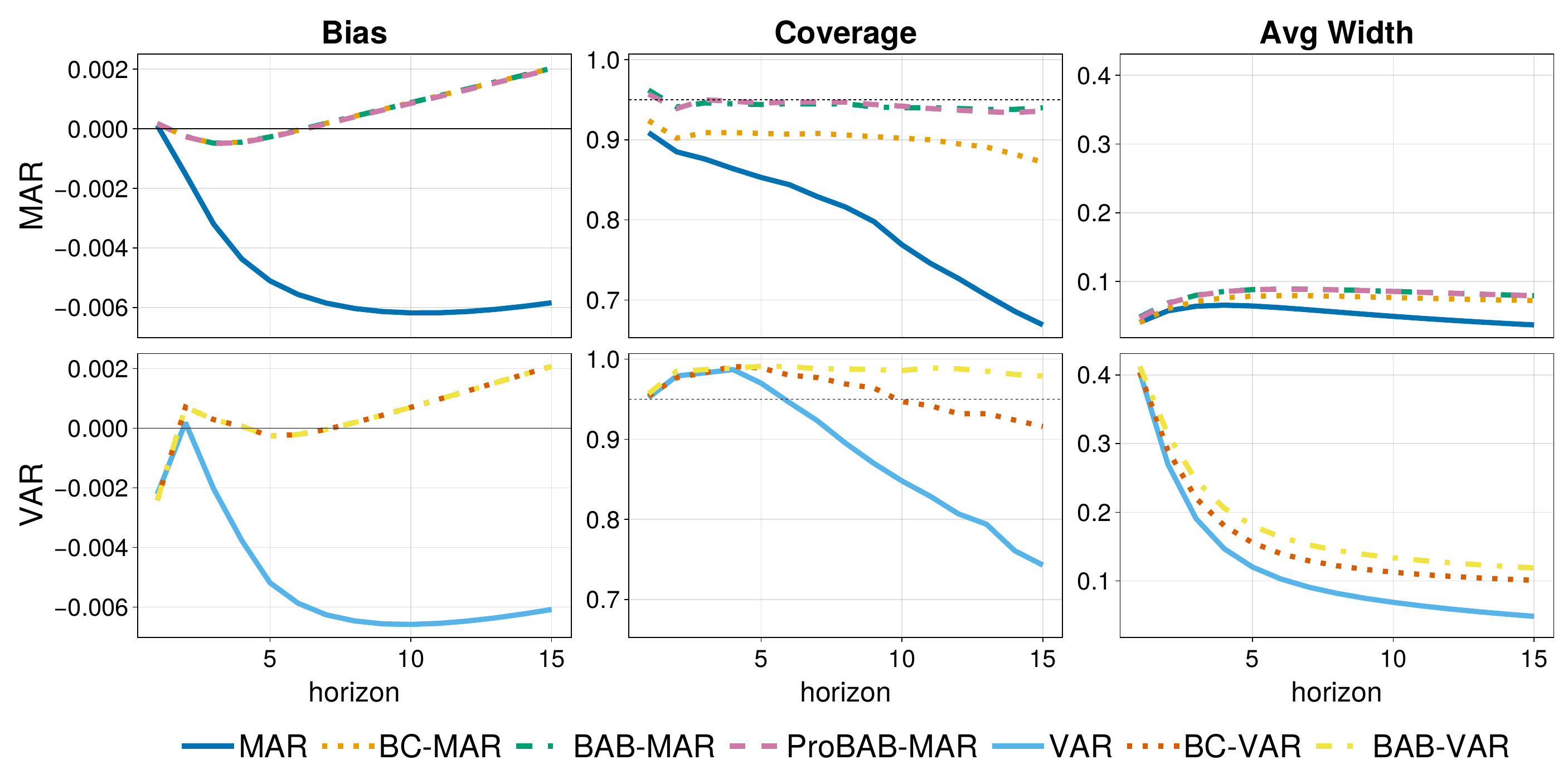}
  \caption{Bias, empirical coverage, and average width of nominal 95\%
  confidence intervals for impulse responses under the MAR(1) DGP
  ($N_1 = 3$, $N_2 = 4$, $T = 100$) with largest companion eigenvalue
  $\rho_{\max} = 0.95$, across horizons $h = 1, \ldots, 15$.}
  \label{fig:persistent_rho95}
\end{figure}
 
\subsection{Near-Unit-Root Behavior}
\label{subsec:sim_persistence}
 
The final design investigates sensitivity to the stability requirement of Assumption \ref{ass:stationary}.
The coefficients are randomly generated Kronecker products as in Section \ref{subsec:sim_samplesize}, rescaled so that the largest companion eigenvalue equals $\rho_{\max} \in \{0.95, 0.98, 0.99\}$ rather than the baseline value of $0.8$, holding $T = 100$, $N_1 = 3$, $N_2 = 4$, and $p = 1$ fixed.
As in Section \ref{subsec:sim_samplesize}, we track the response of entry $(2,4)$ of $\mathbf{Y}_t$ to a shock to entry $(3,2)$.
Every series is therefore stationary and the theory of Section \ref{subsec:bootstrap} formally applies.
Since Corollary \ref{thm:bootstrap} holds at each fixed $\rho_{\max} < 1$ under asymptotics that treat the parameter as fixed as $T$ grows, the guarantee is pointwise rather than uniform, and the quality of the approximation at a given $T$ deteriorates as the boundary is approached \citep{inoue_uniform_2020}.
Three forces compound near the boundary.
The $O(T^{-1})$ bias of the autoregressive estimates grows with persistence, the impulse responses themselves decay only at rate $\rho_{\max}^{h}$ so that estimation error accumulates over the horizons of interest, and the stationarity adjustment of the first bootstrap stage binds increasingly often, attenuating the bias correction precisely when the bias is largest.
The adjustment shrinks the correction for the MAR in 12.4\% of replications at $\rho_{\max} = 0.95$ and in 35.8\% at $\rho_{\max} = 0.99$, compared to 0.01\% under the baseline design, and in 2.7\% of replications at $\rho_{\max} = 0.99$ no correction can be applied at all, leaving the bootstrap to resample from the raw and occasionally explosive estimate.

Figure \ref{fig:persistent_rho95} reports results for $\rho_{\max} = 0.95$, the case closest to our empirical application.
The bias panel shows the anticipated deterioration.
The raw estimates of both models now reach a bias of about $-0.006$ at $h = 10$, which amounts to $35\%$ of the true response for the MAR and $38\%$ for the VAR.
The first stage correction removes nearly all of it over the first ten horizons but overshoots at longer horizons, leaving a positive bias of $0.002$ at $h = 15$, or $15\%$ of the true response there.

The coverage panel shows that persistence mainly affects the delta method.
Its coverage declines steadily with the horizon, to $67\%$ for the MAR and $74\%$ for the VAR at $h = 15$, and recentering at the bias-corrected estimate recovers only part of this gap, leaving $87\%$ and $91\%$.
BAB-MAR and ProBAB-MAR, by contrast, remain close to the nominal level at every horizon, between $93\%$ and $96\%$, while BAB-VAR is conservative, with coverage of $98$ to $99\%$ from $h = 2$ onward.
The width panel shows that the two model groups accumulate uncertainty differently.
The BAB-VAR intervals are widest at the first horizon, where the response is a single coefficient of an unrestricted system with 144 parameters, and narrow over the horizon, in part because the stationarity adjustment truncates the explosive draws that would otherwise dominate the longer horizons.
The MAR intervals start much narrower, reflecting the 24-parameter Kronecker structure, widen over the first five horizons as estimation error compounds through the matrix powers, and then level off.
They remain narrower than their VAR counterparts at every horizon, with a width ratio that rises from 0.11 at $h=1$ and levels off near 0.68 by $h=15$, so that the conservative coverage of BAB-VAR again comes at the cost of wider intervals.

Closer to the boundary, even the bootstrap procedures fall short.
At $\rho_{\max} = 0.99$, reported in Figure \ref{fig:persistent_rho99} in Appendix \ref{app:additional_sims}, BAB-MAR and ProBAB-MAR reach only $89\%$ at the worst horizons, against $94\%$ for BAB-VAR.
This mirrors the behavior of bootstrap impulse response inference in highly persistent unrestricted VARs, where the most accurate intervals, even at long horizons, come from bootstrapping a lag-augmented autoregression with a bias adjustment \citep{inoue_uniform_2020}, estimating the model with a redundant additional lag as \citet{toda_integrated_1995} proposed for Wald tests in possibly integrated systems.
Adapting lag augmentation to the Kronecker structure of the MAR is left for future work.
 
\section{Empirical Application: Euro Area Inflation}\label{sec:application}
 
We consider an application to euro area inflation, in which one dimension of the matrix indexes countries and the other inflation categories.
Section \ref{subsec:data} describes the data and the specification, Section \ref{subsec:ident} sets out the identification scheme, and Section \ref{subsec:results} reports the results.
 
\subsection{Data and Specification}\label{subsec:data}

\begin{table}[!t]
  \centering
  \begin{tabular}{lccccc@{\hspace{2.5em}}lccc}
    \toprule
    \multicolumn{6}{c}{Countries, $\widehat{\mathbf{S}}_1$} & \multicolumn{4}{c}{Categories, $\widehat{\mathbf{S}}_2$} \\
    \cmidrule(r){1-6} \cmidrule(l){7-10}
       & AT & BE & DE & FR & NL & & Energy & Food & Core \\
    \midrule
    AT & 1.00 &      &      &      &      & Energy & 1.00    &         &      \\
    BE & 0.56 & 1.00 &      &      &      & Food   & 0.11    & 1.00    &      \\
    DE & 0.55 & 0.57 & 1.00 &      &      & Core   & $-0.03$ & $-0.02$ & 1.00 \\
    FR & 0.58 & 0.66 & 0.59 & 1.00 &      &        &         &         &      \\
    NL & 0.33 & 0.43 & 0.47 & 0.40 & 1.00 &        &         &         &      \\
    \bottomrule
  \end{tabular}
  \caption{Contemporaneous structure of the MAR(2) innovations. The left panel reports the correlation matrix $\widehat{\mathbf{S}}_1$ implied by $\widehat{\boldsymbol{\Sigma}}_1$ and the right panel the correlation matrix $\widehat{\mathbf{S}}_2$ implied by $\widehat{\boldsymbol{\Sigma}}_2$.}
  \label{tab:cov}
\end{table}

We use monthly year-on-year HICP inflation rates obtained from Eurostat for $N_1 = 5$ euro area countries (Austria, Belgium, Germany, France, and the Netherlands) and $N_2 = 3$ categories (energy, food, and core), from December 2001 to December 2019.
Each month therefore delivers an $N_1 \times N_2$ matrix $\mathbf{Y}_t$, giving $217$ observations in total.
We end the sample in 2019 so that the estimates are not driven by the pandemic period.
In line with the model of Section \ref{subsec:MAR}, which contains no intercept, each series is demeaned prior to estimation.

We select the lag order over $p = 0, \ldots, 12$ by the Schwarz, Hannan--Quinn and Akaike criteria.
The former two agree within each model, selecting $p = 2$ for the MAR and $p = 1$ for the VAR, whereas the Akaike criterion selects $p = 12$ for both.
Since all candidate lags are evaluated on the same common sample of $205$ observations, the criteria are comparable across models, and both the Schwarz and Hannan--Quinn criteria prefer the MAR to the VAR.
We proceed with $p = 2$ lags for the MAR and $p = 1$ lag for the VAR for a fair comparison of the confidence intervals.

The MAR is left estimating $66$ coefficients against $225$ for the VAR.
The selected MAR(2) is persistent, with a largest companion eigenvalue of $0.95$.
Table \ref{tab:cov} reports the correlation matrices implied by the fitted covariances. 
This is the configuration under which the empirically calibrated DGP of Section \ref{subsec:sim_empirical} was simulated.

\subsection{Identification}\label{subsec:ident}

We identify the shocks recursively through the Cholesky factors of $\boldsymbol{\Sigma}_1$ and $\boldsymbol{\Sigma}_2$, with energy ordered first among the three categories, and, as in the block Cholesky impulse response of \citet{billio_bayesian_2023}, shock energy inflation in all five countries jointly.
We normalize the shock to raise energy inflation by one percentage point on impact in each country, so that under the separable covariance its impact vector is
\begin{equation*}
    \mathbf{v} = \frac{\boldsymbol{\Sigma}_2\mathbf{e}_1}{\sigma_{2,11}} \otimes \boldsymbol{\iota}_{N_1},
\end{equation*}
with $\mathbf{e}_1 \in \mathbb{R}^{N_2}$ selecting energy, $\sigma_{2,11}$ the leading diagonal entry of $\boldsymbol{\Sigma}_2$, and $\boldsymbol{\iota}_{N_1}$ the $N_1$-vector of ones, and the horizon-$h$ response is $\boldsymbol{\Theta}_h\mathbf{v}$.
Food and core thus respond on impact by the same amount in every country.
The impact vector depends on neither $\boldsymbol{\Sigma}_1$ nor the ordering of the countries, and on $\boldsymbol{\Sigma}_2$ only through its first column, so that the only restriction is that energy inflation does not respond within the month to food or core inflation.
This is the assumption that energy prices are predetermined at monthly frequency, which \citet{kilian_energy_2008} argues is a reasonable approximation for monthly, though not for annual, data, and which is standard in the literature on energy pass-through that orders energy first in a recursive scheme \citep{conflitti_oil_2019, borrallo_transmission_2026}.
Because energy is ordered first, the shock is the innovation to energy inflation itself, and it therefore combines the supply and demand shocks in energy markets that can affect the price level quite differently \citep{kilian_oil_2009}.
When energy prices rise because demand is strong, food and core prices may rise for the same reason, so the responses below average over these sources, and part of what we call pass-through may reflect common demand conditions rather than the transmission of energy costs.

\subsection{Results}\label{subsec:results}

\begin{figure}[htbp]
  \centering
  \includegraphics[width=\linewidth]{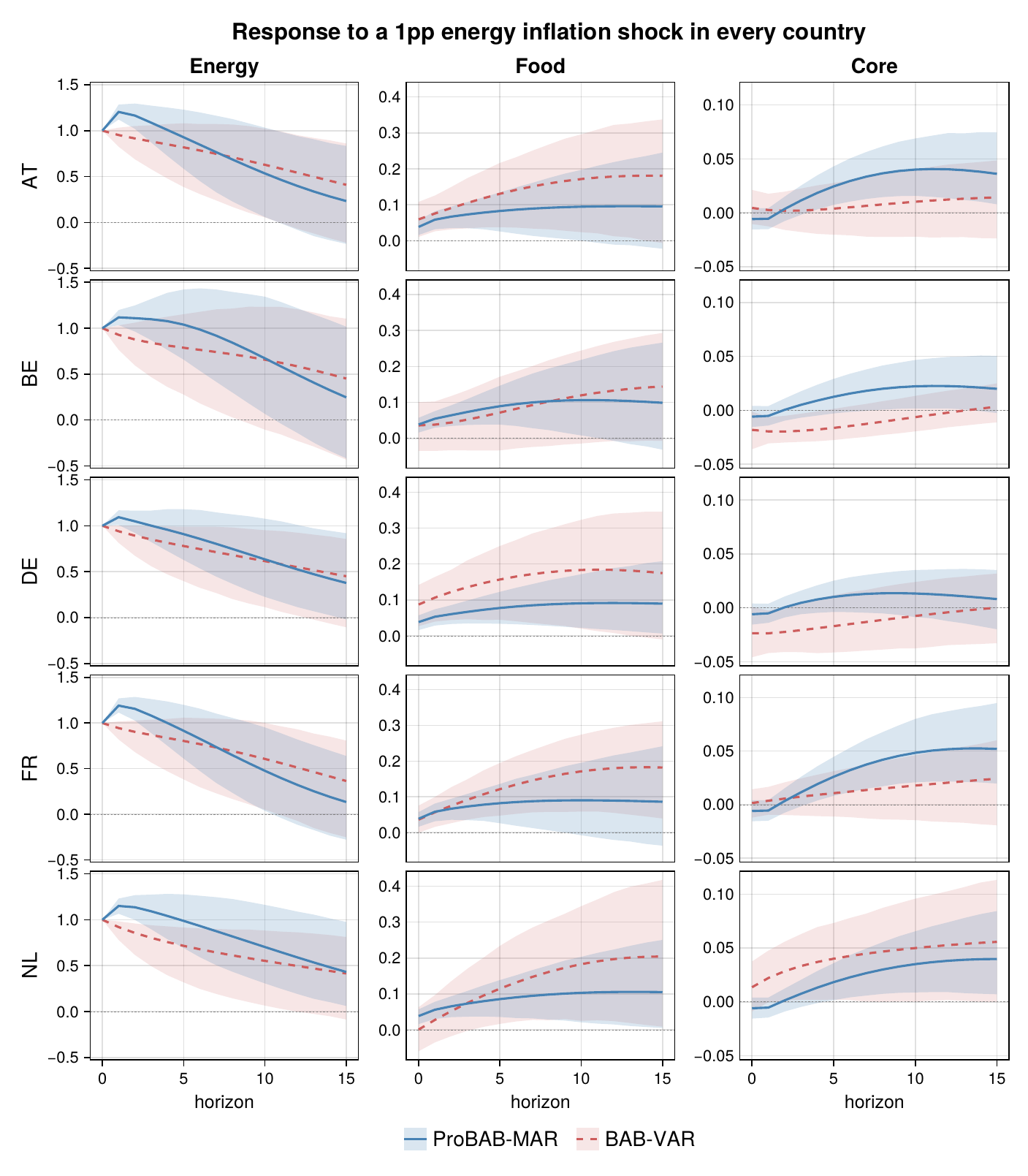}
  \caption{Responses to a shock raising energy inflation by one percentage point on impact in each of the five countries, with $95\%$ bootstrap confidence intervals. Rows index the responding country and columns the responding category, and panels in the same column share a scale. Solid lines with shaded bands are ProBAB-MAR; dashed lines with bands are BAB-VAR on the vectorized system.}
  \label{fig:application}
\end{figure}

Figure \ref{fig:application} reports the responses to a shock that raises energy inflation by one percentage point on impact in each of the five countries, for the three categories in each country, over horizons $h = 0, \ldots, 15$.
The solid blue line and band give the ProBAB-MAR point estimates and $95\%$ percentile intervals of Section \ref{subsec:probab}, computed with $B_1 = 1000$ first-stage and $B_2 = 2000$ second-stage replications, and the dashed red line and band give the BAB-VAR counterparts applied to the vectorized system.
Both are computed under the identification of Section \ref{subsec:ident}, with the MAR estimated at lag order $p = 2$ and the VAR at $p = 1$.

The point estimates from the two models track each other closely, while the MAR intervals are narrower, by a factor of roughly $1.4$ in mean and $1.3$ in median across responses and horizons $h=1$ to $h=15$. 
The gain is largest on impact and fades with the horizon, as in the persistent designs of Section \ref{subsec:sim_persistence}.
Estimating the VAR at the same lag order $p = 2$ widens this gap further, to $1.6$ in mean and $1.4$ in median, so the narrower MAR intervals do not stem from the difference in lag order (Appendix \ref{app:additional_application}).
The stationarity adjustment of Section \ref{subsec:probab} also binds far more often for the VAR.
It shrinks the bias correction in $99\%$ of the VAR bootstrap replications, against $23\%$ of the MAR replications, and shrinks it heavily, so that on average only $18\%$ of the correction is applied to the VAR, against $92\%$ to the MAR.

The Kronecker structure also shapes how the MAR responses can be read.
Since the impact on food and core is small, as shown in Table \ref{tab:cov}, the transmission of the shock into the other categories builds up only through the dynamics.
At $h = 1$ the matricized response is rank one by Proposition \ref{prop:rank}, so that every country responds with the same pattern across categories, scaled by a country-specific factor, and differences in the composition of the responses across countries emerge only from $h = 2$ onward.
Energy inflation itself is persistent, halving between $h = 10$ and $h = 14$ in every country, a persistence that the year-on-year measurement makes partly mechanical over the first twelve months, and remaining significant through $h = 9$ in Austria and Belgium, $h = 10$ in France, $h = 13$ in Germany, and to the end of the horizon in the Netherlands.
Pass-through into food, the extent to which the energy shock transmits into food inflation, is gradual but sizable, peaking at $0.09$ to $0.11$ percentage points after ten to thirteen months, whereas pass-through into core is smaller and slower, at most $0.05$ percentage points, significant from $h = 4$ or $h = 5$ onward in Austria, France and the Netherlands, only at $h = 10$ and $h = 11$ in Belgium, and at no horizon in Germany.
The composition of the responses thus differs across countries, with Germany showing no significant pass-through into core despite a persistent energy response, while in France and the Netherlands the core response peaks only at the end of the horizon.

These patterns are in line with existing evidence for the euro area.
\citet{borrallo_transmission_2026}, who also measure food inflation year on year, find that its response to fuel price shocks peaks twelve to fifteen months after the shock, consistent with a price-level adjustment spread over about a year, while \citet{conflitti_oil_2019}, working with month-on-month rates, find that the pass-through of oil prices into core inflation is zero within the month but small and persistent thereafter.
Since year-on-year rates return to zero once the price level stops rising, the core responses still rising at $h = 15$ in France and the Netherlands point to pass-through that continues beyond the first year rather than to a base effect.
This is the heterogeneity that the relaxation of the rank bound from $h = 2$ onward allows, and that a MAR(1) would rule out.

\section{Conclusion}
\label{sec:conclusion}

This paper develops impulse response inference for stable matrix autoregressive models.
We derive the joint asymptotic distribution of the coefficient and covariance estimators of the MAR($p$), which yields delta-method standard errors for identified impulse responses. In addition, we propose ProBAB-MAR, a bias-corrected bootstrap that projects the bias-corrected coefficients back onto the space of Kronecker products and attains asymptotically correct coverage.
We also show that the rank-one structure of MAR($1$) impulse responses relaxes gradually with the horizon for $p \geq 2$, so that the MAR($p$) admits increasingly heterogeneous responses across the two dimensions.

In our simulations, the delta-method intervals undercover in small samples, whereas ProBAB-MAR attains close to nominal coverage in all but the most persistent designs, with intervals substantially narrower than those of the unrestricted VAR bootstrap.
In the application to euro area inflation, a common energy shock is persistent and passes through gradually into food and, to a lesser extent and with marked differences across countries, into core inflation. 
Replication material for all analyses is available at \url{https://github.com/ivanuricardo/matrix_irf_reproduction}.

Several limitations point to directions for future work.
The efficiency gains rest on the Kronecker structure of the coefficients, and imposing it when it fails can bias impulse responses and distort coverage \citep{olea_primer_2025}.
A test of this restriction would therefore be a natural complement to our procedures.
The bootstrap after bootstrap is costly, since each of its $B_1 + B_2$ replications re-estimates the MAR by an iterative algorithm, and an analytical bias correction in the spirit of \citet{pope_biases_1990} would remove the first stage.
We have focused on stable processes, and our simulations show that the accuracy of all procedures deteriorates near the unit root.
A possible unit root would instead call for a lag-augmented autoregression, for which \citet{inoue_uniform_2020} establish the uniform asymptotic validity of bootstrap impulse-response inference, and adapting this to the MAR, as well as extending the framework to cointegrated matrix-valued systems, is left for future work.
Finally, the identification schemes we consider are recursive, and structural identification that exploits the matrix structure (see e.g., \citealp{bucci_structural_2026, lara_structural_2026}), as well as the extension of our results to tensor autoregressions \citep{li_multilinear_2021}, are natural next steps.

\section*{Acknowledgments}

The third author is financially supported by a grant from the Dutch Research Council (NWO), research programme Vidi (VI.Vidi.211.032).
Previous versions of this paper were presented at NESG 2026 and we gratefully acknowledge comments by the participants. 

\bibliographystyle{asa}
\bibliography{references}

\newpage

\appendix
\counterwithin{equation}{section}
\renewcommand{\theequation}{\thesection.\arabic{equation}}

\section{Proofs}\label{app:proofs}
For a generic matrix $\mathbf{X}$, we write $\mathbf{X}^\top$ for its transpose, $\|\mathbf{X}\|$ for its Frobenius norm, and $\operatorname{vec}(\mathbf{X})$ for its column-wise vectorization, with $\operatorname{unvec}_{N_1,N_2}(\cdot)$ denoting the inverse operation that reshapes an $N_1 N_2$-vector back into an $N_1 \times N_2$ matrix.
We use $\otimes$ for the Kronecker product, $\mathbf{I}_N$ for the $N \times N$ identity matrix, and $\mathbb{S}^N$ for the set of symmetric $N \times N$ matrices.
Furthermore, we use the following identities throughout.

\begin{lemma}[Kronecker identities]\label{lem:kron}
For $\mathbf{A}, \mathbf{B} \in \mathbb{R}^{N_1 \times N_2}$,
$\mathbf{C}, \boldsymbol{\Sigma}_1 \in \mathbb{R}^{N_1 \times N_1}$, and
$\mathbf{D}, \boldsymbol{\Sigma}_2 \in \mathbb{R}^{N_2 \times N_2}$,
\begin{enumerate}
    \item $(\mathbf{D} \otimes \mathbf{C})\operatorname{vec}(\mathbf{A}) = \operatorname{vec}(\mathbf{C}\mathbf{A}\mathbf{D}^\top)$;
    \item $\operatorname{tr}\big[(\mathbf{D} \otimes \mathbf{C})\operatorname{vec}(\mathbf{A})\operatorname{vec}(\mathbf{B})^\top\big] = \operatorname{tr}\big(\mathbf{B}^\top\mathbf{C}\mathbf{A}\mathbf{D}^\top\big)$;
    \item $\operatorname{tr}\big[(\mathbf{D} \otimes \mathbf{C})(\boldsymbol{\Sigma}_2 \otimes \boldsymbol{\Sigma}_1)\big] = \operatorname{tr}(\mathbf{D}\boldsymbol{\Sigma}_2)\operatorname{tr}(\mathbf{C}\boldsymbol{\Sigma}_1)$.
\end{enumerate}
\end{lemma}

\begin{proof}
Part 1 is the vectorization rule \citep[App.~A.12]{lutkepohl_new_2005}. Part 2 follows
from part 1 and $\operatorname{vec}(\mathbf{B})^\top\operatorname{vec}(\mathbf{X})
= \operatorname{tr}(\mathbf{B}^\top\mathbf{X})$, and part 3 from the mixed-product rule
together with $\operatorname{tr}(\mathbf{X}\otimes\mathbf{Y})
= \operatorname{tr}(\mathbf{X})\operatorname{tr}(\mathbf{Y})$
\citep[App.~A.11--A.12]{lutkepohl_new_2005}.
\end{proof}

\subsection{Rank of the  Matricized Impulse Responses}\label{app:rank}

\begin{proof}[Proof of Proposition \ref{prop:rank}]
Recall that $\boldsymbol{\Theta}_h$ are the MMA coefficients at horizon $h$ and $\mathbf{v} = \mathbf{v}_2 \otimes \mathbf{v}_1$ with $\mathbf{v}_1 \in \mathbb{R}^{N_1}$ and $\mathbf{v}_2 \in \mathbb{R}^{N_2}$.
We establish the expansion
    \begin{equation} \label{app:eq:kron_expansion}
        \boldsymbol{\Theta}_h \mathbf{v}
        = \sum_{\ell=1}^{c_h} \mathbf{b}^{(\ell)} \otimes \mathbf{a}^{(\ell)},
        \qquad\text{equivalently}\qquad
        \operatorname{unvec}_{N_1,N_2}(\boldsymbol{\Theta}_h \mathbf{v})
        = \sum_{\ell=1}^{c_h} \mathbf{a}^{(\ell)} \big(\mathbf{b}^{(\ell)}\big)^{\top},
    \end{equation}
by induction on $h$.
For $h = 0$ it holds with $c_0 = 1$, $\mathbf{a}^{(1)} = \mathbf{v}_1$ and $\mathbf{b}^{(1)} = \mathbf{v}_2$, since $\boldsymbol{\Theta}_0 \mathbf{v} = \mathbf{v} = \mathbf{v}_2 \otimes \mathbf{v}_1$.
Suppose \eqref{app:eq:kron_expansion} holds at horizons $0, \ldots, h-1$ for some $h \geq 1$, and write $\mathbf{a}_{h-j}^{(\ell)}$ and $\mathbf{b}_{h-j}^{(\ell)}$, $\ell = 1, \ldots, c_{h-j}$, for the factor vectors at horizon $h - j$.
Post-multiplying the recursion 
\begin{equation}
    \boldsymbol{\Theta}_0 = \mathbf{I}_{N_1 N_2}, \qquad
    \boldsymbol{\Theta}_h = \sum_{j=1}^{\min(p,h)} \mathbf{C}_j \boldsymbol{\Theta}_{h-j}, \qquad h \geq 1 \nonumber
\end{equation} 
by $\mathbf{v}$ and applying the mixed-product rule $(\mathbf{B} \otimes \mathbf{A})(\mathbf{b} \otimes \mathbf{a}) = (\mathbf{B}\mathbf{b}) \otimes (\mathbf{A}\mathbf{a})$ termwise, we have
\begin{equation*} 
    \boldsymbol{\Theta}_h \mathbf{v}
    = \sum_{j=1}^{\min(p,h)} (\mathbf{B}_j \otimes \mathbf{A}_j) \boldsymbol{\Theta}_{h-j} \mathbf{v}
    = \sum_{j=1}^{\min(p,h)} \sum_{\ell=1}^{c_{h-j}} \big(\mathbf{B}_j \mathbf{b}_{h-j}^{(\ell)}\big) \otimes \big(\mathbf{A}_j \mathbf{a}_{h-j}^{(\ell)}\big),
\end{equation*}
a sum of $\sum_{j=1}^{\min(p,h)} c_{h-j} = c_h$ Kronecker products, which is \eqref{app:eq:kron_expansion} at horizon $h$.
Since the  $\operatorname{unvec}_{N_1,N_2}(\mathbf{b} \otimes \mathbf{a}) = \mathbf{a}\mathbf{b}^\top$ operator is linear, applying it to \eqref{app:eq:kron_expansion} gives $\operatorname{unvec}_{N_1,N_2}(\boldsymbol{\Theta}_h \mathbf{v}) = \sum_{\ell=1}^{c_h} \mathbf{a}^{(\ell)} (\mathbf{b}^{(\ell)})^\top$, a sum of $c_h$ matrices of rank at most one.
Subadditivity of the rank, $\operatorname{rank}(\mathbf{X} + \mathbf{Y}) \leq \operatorname{rank}(\mathbf{X}) + \operatorname{rank}(\mathbf{Y})$, bounds the rank of the sum by $c_h$, and any $N_1 \times N_2$ matrix has rank at most $\min(N_1, N_2)$, so  the rank bound follows.
\end{proof}

\begin{proof}[Proof of Corollary \ref{cor:rank}]
Collect the factor vectors of \eqref{app:eq:kron_expansion} in $\boldsymbol{\mathcal{A}}_h = [\mathbf{a}^{(1)} \cdots \mathbf{a}^{(c_h)}] \in \mathbb{R}^{N_1 \times c_h}$ and $\boldsymbol{\mathcal{B}}_h = [\mathbf{b}^{(1)} \cdots \mathbf{b}^{(c_h)}] \in \mathbb{R}^{N_2 \times c_h}$, so that $\operatorname{unvec}_{N_1,N_2}(\boldsymbol{\Theta}_h \mathbf{v}) = \boldsymbol{\mathcal{A}}_h\boldsymbol{\mathcal{B}}_h^\top$.
Linear independence of $\{\mathbf{a}^{(\ell)}\}$ means that $\boldsymbol{\mathcal{A}}_h$ has $c_h$ linearly independent columns in $\mathbb{R}^{N_1}$, so $c_h \leq N_1$ and $\operatorname{rank}(\boldsymbol{\mathcal{A}}_h) = c_h$.
The same argument applied to $\boldsymbol{\mathcal{B}}_h$ gives $c_h \leq N_2$ and $\operatorname{rank}(\boldsymbol{\mathcal{B}}_h) = c_h$, hence $\min\{c_h, N_1, N_2\} = c_h$.
Because $\boldsymbol{\mathcal{A}}_h$ has full column rank, $\operatorname{rank}\big(\operatorname{unvec}_{N_1,N_2}(\boldsymbol{\Theta}_h \mathbf{v})\big) = \operatorname{rank}(\boldsymbol{\mathcal{B}}_h^\top) = c_h$, so \eqref{eq:rank_bound} holds with equality.
\end{proof}

\subsection{Asymptotic Inference}\label{app:asymptotic}

This subsection proves Proposition \ref{prop:an} and derives the Jacobian $\mathbf{J}_\beta$ of \eqref{eq:OmegaC}.
The coefficient block of the limit is known for the MAR($1$) \citep{chen_autoregressive_2021} and at general lag order \citep{li_multilinear_2021}, but three features of our setting require the expansion to be given in full.
First, the parameter vector includes $\boldsymbol{\sigma}$, whose score is quadratic rather than linear in $\mathbf{U}_t$.
Second, the MAR($p$) carries $p+1$ scale indeterminacies rather than one, so the correction that restores invertibility has rank $p+1$.
Finally, the cross-lag blocks $\mathbf{B}_i \otimes \mathbf{A}_j$, $i \neq j$, of the stacked Kronecker product must be discarded to pass from $(\widehat{\mathbf{A}}, \widehat{\mathbf{B}})$ to $\operatorname{vec}(\widehat{\mathbf{C}})$.
Section \ref{app:score} collects the score, the gradients of the normalizations and the constrained sandwich on which the argument rests, and shows why the coefficient and covariance blocks separate.
Section \ref{app:an} carries out the expansion, and Section \ref{app:jacobian} derives $\mathbf{J}_\beta$.

\subsubsection{Score and Constrained Sandwich}\label{app:score}

Recall $\boldsymbol{\zeta} = (\boldsymbol{\beta}^\top, \boldsymbol{\sigma}^\top)^\top \in \mathbb{R}^d$ from Section \ref{subsec:asymptotic} and let $\mathcal{F}_t = \sigma(\mathbf{U}_s : s \leq t)$.
The contribution of observation $t$ to the score of \eqref{eq:loglik} is $\mathbf{s}_t = (\mathbf{s}_{\beta,t}^\top, \mathbf{s}_{\sigma,t}^\top)^\top$, with a coefficient block linear in the innovations,
\begin{equation}\label{eq:score_beta}
    \mathbf{s}_{\beta,t} = \mathbf{W}_t\boldsymbol{\Sigma}^{-1}\mathbf{u}_t,
    \qquad
    \mathbf{W}_t =
    \begin{bmatrix}
        \mathbf{X}_t \mathbf{B}^\top \otimes \mathbf{I}_{N_1} \\
        \mathbf{I}_{N_2} \otimes \mathbf{X}_t^\top \mathbf{A}^\top
    \end{bmatrix}
    \in \mathbb{R}^{p(N_1^2 + N_2^2) \times N_1 N_2},
\end{equation}
and a covariance block quadratic in them,
\begin{equation}\label{eq:score_sigma}
    \mathbf{s}_{\sigma,t} = \frac{1}{2}
    \begin{bmatrix}
        \mathbf{D}_{N_1}^\top \operatorname{vec} \big(\boldsymbol{\Sigma}_1^{-1}
        (\mathbf{U}_t\boldsymbol{\Sigma}_2^{-1}\mathbf{U}_t^\top - N_2\boldsymbol{\Sigma}_1)
        \boldsymbol{\Sigma}_1^{-1}\big) \\[3pt]
        \mathbf{D}_{N_2}^\top \operatorname{vec} \big(\boldsymbol{\Sigma}_2^{-1}
        (\mathbf{U}_t^\top\boldsymbol{\Sigma}_1^{-1}\mathbf{U}_t - N_1\boldsymbol{\Sigma}_2)
        \boldsymbol{\Sigma}_2^{-1}\big)
    \end{bmatrix},
\end{equation}
where $\mathbf{D}_N \in \mathbb{R}^{N^2 \times N(N+1)/2}$ is the duplication matrix, defined by $\operatorname{vec}(\mathbf{S}) = \mathbf{D}_N\operatorname{vech}(\mathbf{S})$ for $\mathbf{S} \in \mathbb{S}^N$.
For fixed $\mathbf{M} \in \mathbb{R}^{N_2 \times N_2}$, Lemma \ref{lem:kron} with $\mathbb{E}(\mathbf{u}_t\mathbf{u}_t^\top) = \boldsymbol{\Sigma}_2 \otimes \boldsymbol{\Sigma}_1$ gives $\operatorname{tr}[\mathbf{F}\,\mathbb{E}(\mathbf{U}_t\mathbf{M}\mathbf{U}_t^\top)] = \operatorname{tr}(\mathbf{M}\boldsymbol{\Sigma}_2)\operatorname{tr}(\mathbf{F}\boldsymbol{\Sigma}_1)$ for every $\mathbf{F} \in \mathbb{R}^{N_1 \times N_1}$, so that
\begin{equation}\label{eqA:UMU}
    \mathbb{E}(\mathbf{U}_t\mathbf{M}\mathbf{U}_t^\top) = \operatorname{tr}(\mathbf{M}\boldsymbol{\Sigma}_2)\,\boldsymbol{\Sigma}_1,
\end{equation}
and likewise $\mathbb{E}(\mathbf{U}_t^\top\mathbf{M}\mathbf{U}_t) = \operatorname{tr}(\mathbf{M}\boldsymbol{\Sigma}_1)\,\boldsymbol{\Sigma}_2$ for $\mathbf{M} \in \mathbb{R}^{N_1 \times N_1}$.
In particular, $\mathbf{s}_{\sigma,t}$ has mean zero.
As $\mathbf{W}_t$ is $\mathcal{F}_{t-1}$-measurable and $\mathbf{U}_t$ is independent of $\mathcal{F}_{t-1}$, $\mathbf{s}_t$ is a martingale difference sequence.
The same measurability controls every average in which an innovation meets a regressor.
For any fixed conformable $\mathbf{M}$,
\begin{equation}\label{eqA:mds}
    \frac{1}{T}\sum_{t=1}^{T} \mathbf{U}_t\mathbf{M}\mathbf{X}_t^\top = O_p(T^{-1/2}),
\end{equation}
since each block $\mathbf{Y}_{t-j}$ of $\mathbf{X}_t$ is $\mathcal{F}_{t-1}$-measurable while
$\mathbb{E}(\mathbf{U}_t\mid\mathcal{F}_{t-1}) = \mathbf{0}$, so the summands form a
martingale difference sequence, square integrable under Assumptions
\ref{ass:stationary} and \ref{ass:iid-plus}.
We refer to \eqref{eqA:mds} as the predeterminedness of the regressors.

Since $\mathbf{B}_j \otimes \mathbf{A}_j = (c^{-1}\mathbf{B}_j) \otimes (c\mathbf{A}_j)$ and $\boldsymbol{\Sigma}_2 \otimes \boldsymbol{\Sigma}_1 = (c^{-1}\boldsymbol{\Sigma}_2) \otimes (c\boldsymbol{\Sigma}_1)$ for every $c > 0$, the likelihood is flat along $p+1$ directions, which span the null space of the expected Hessian $\mathbf{H}_0 = \mathbb{E}(-\partial\mathbf{s}_t/\partial\boldsymbol{\zeta}^\top)$.
The normalizations $\|\mathbf{A}_j\| = 1$, $j = 1, \ldots, p$, and $\|\boldsymbol{\Sigma}_1\| = 1$ remove these directions, and we collect their gradients with respect to $\boldsymbol{\zeta}$ in
\begin{equation}\label{eq:Gamma}
    \boldsymbol{\Gamma} = \operatorname{blkdiag}(\boldsymbol{\Gamma}_\beta, \boldsymbol{\Gamma}_\sigma) \in \mathbb{R}^{d \times (p+1)},
    \qquad
    \boldsymbol{\Gamma}_\beta =
    \begin{bmatrix}
        \operatorname{blkdiag}(\mathbf{a}_1, \ldots, \mathbf{a}_p) \\ \mathbf{0}
    \end{bmatrix},
    \qquad
    \boldsymbol{\Gamma}_\sigma =
    \begin{bmatrix}
        \mathbf{D}_{N_1}^\top\operatorname{vec}(\boldsymbol{\Sigma}_1) \\ \mathbf{0}
    \end{bmatrix},
\end{equation}
with $\mathbf{a}_j = \operatorname{vec}(\mathbf{A}_j)$.
The two row blocks of $\boldsymbol{\Gamma}_\beta$ correspond to $\operatorname{vec}(\mathbf{A})$ and $\operatorname{vec}(\mathbf{B}^\top)$, and those of $\boldsymbol{\Gamma}_\sigma$ to $\operatorname{vech}(\boldsymbol{\Sigma}_1)$ and $\operatorname{vech}(\boldsymbol{\Sigma}_2)$.
Each column of $\boldsymbol{\Gamma}$ has unit inner product with the tangent direction of its own normalization and is orthogonal to the other $p$.
Hence $\boldsymbol{\Gamma}\boldsymbol{\Gamma}^\top$ is positive definite on the null space of $\mathbf{H}_0$, and $\mathbf{H} = \mathbf{H}_0 + \boldsymbol{\Gamma}\boldsymbol{\Gamma}^\top$ is positive definite.
With $\boldsymbol{\Lambda} = \mathbb{E}(\mathbf{s}_t\mathbf{s}_t^\top)$, the asymptotic variance is the constrained sandwich
\begin{equation}\label{eq:Omega}
    \boldsymbol{\Omega} = \mathbf{H}^{-1}\boldsymbol{\Lambda}\mathbf{H}^{-1}.
\end{equation}

Both $\mathbf{H}$ and $\boldsymbol{\Lambda}$ are block diagonal between $\boldsymbol{\beta}$ and $\boldsymbol{\sigma}$, and we write $\mathbf{H}_\beta$, $\boldsymbol{\Lambda}_\sigma$, and so on, for their diagonal blocks.
For $\mathbf{H}$, the product $\boldsymbol{\Gamma}\boldsymbol{\Gamma}^\top$ is block diagonal by \eqref{eq:Gamma}.
The off-diagonal block of $\mathbf{H}_0$ vanishes because the derivative of $\mathbf{s}_{\beta,t}$ with respect to $\boldsymbol{\sigma}$ is linear in $\mathbf{u}_t$ with $\mathcal{F}_{t-1}$-measurable coefficients.
For $\boldsymbol{\Lambda}$, $\mathbf{s}_{\sigma,t}$ depends on $\mathbf{U}_t$ alone, so
\begin{equation*}
    \mathbb{E}\big(\mathbf{s}_{\beta,t}\mathbf{s}_{\sigma,t}^\top\big)
    = \mathbb{E}(\mathbf{W}_t)\boldsymbol{\Sigma}^{-1}\mathbb{E}\big(\mathbf{u}_t\mathbf{s}_{\sigma,t}^\top\big) = \mathbf{0},
\end{equation*}
since $\mathbf{W}_t$ is linear in $\mathbf{X}_t$ and the series are demeaned.
The third moments of the innovations, which enter through $\mathbb{E}(\mathbf{u}_t\mathbf{s}_{\sigma,t}^\top)$, therefore play no role, and no symmetry restriction on $\mathbf{U}_t$ is required.
It follows that $\boldsymbol{\Omega} = \operatorname{blkdiag}(\boldsymbol{\Omega}_\beta, \boldsymbol{\Omega}_\sigma)$.

\subsubsection{Joint Asymptotic Normality}\label{app:an}

Throughout, $\Delta_A = \widehat{\mathbf{A}} - \mathbf{A}$, $\Delta_B = \widehat{\mathbf{B}} - \mathbf{B}$ and $\Delta_i = \widehat{\boldsymbol{\Sigma}}_i - \boldsymbol{\Sigma}_i$ ($i = 1, 2$).
The expansion takes as given the $\sqrt{T}$ rate of the following lemma, which for the coefficient block of the MAR($1$) is part of Theorem 4 of \citet{chen_autoregressive_2021}.

\begin{lemma}[Rate]\label{lem:rate}
Let Assumptions \ref{ass:stationary} and \ref{ass:iid-plus} hold.
Moreover, let $\widehat{\boldsymbol{\zeta}}$ be a maximizer of \eqref{eq:loglik} over the set $\mathcal{Z}$ of parameters satisfying the normalizations of Section \ref{app:score}.
Then
\begin{equation}\label{eqA:rate}
    \widehat{\boldsymbol{\zeta}} - \boldsymbol{\zeta} = O_p(T^{-1/2}).
\end{equation}
\end{lemma}

\begin{proof}
The argument extends the proof of Theorem 4 of \citet{chen_autoregressive_2021}, which rests on the same decomposition of the likelihood.
Write $\mathbf{y}_t = \operatorname{vec}(\mathbf{Y}_t) = \mathbf{C}\mathbf{x}_t + \mathbf{u}_t$ and let
\begin{equation*}
    \widehat{\boldsymbol{\Phi}} = \frac{1}{T}\sum_{t=1}^{T}\mathbf{x}_t\mathbf{x}_t^\top,
    \qquad
    \breve{\mathbf{C}} = \Big(\frac{1}{T}\sum_{t=1}^{T}\mathbf{y}_t\mathbf{x}_t^\top\Big)\widehat{\boldsymbol{\Phi}}^{-1},
    \qquad
    \breve{\boldsymbol{\Sigma}} = \frac{1}{T}\sum_{t=1}^{T}(\mathbf{y}_t - \breve{\mathbf{C}}\mathbf{x}_t)(\mathbf{y}_t - \breve{\mathbf{C}}\mathbf{x}_t)^\top,
\end{equation*}
the last two being the least squares estimates of the unrestricted VAR($p$).
For every $\bar{\mathbf{C}}$,
\begin{equation*}
    \frac{1}{T}\sum_{t=1}^{T}(\mathbf{y}_t - \bar{\mathbf{C}}\mathbf{x}_t)(\mathbf{y}_t - \bar{\mathbf{C}}\mathbf{x}_t)^\top
    = \breve{\boldsymbol{\Sigma}} + (\bar{\mathbf{C}} - \breve{\mathbf{C}})\widehat{\boldsymbol{\Phi}}(\bar{\mathbf{C}} - \breve{\mathbf{C}})^\top.
\end{equation*}
Hence, up to a constant, the negative log likelihood of equation \eqref{eq:loglik} divided by $T$ is
\begin{equation}\label{eqA:F}
    F_T(\bar{\boldsymbol{\zeta}})
    = \frac{1}{2}\log|\bar{\boldsymbol{\Sigma}}|
    + \frac{1}{2}\operatorname{tr}\big(\bar{\boldsymbol{\Sigma}}^{-1}\big[\breve{\boldsymbol{\Sigma}} + (\bar{\mathbf{C}} - \breve{\mathbf{C}})\widehat{\boldsymbol{\Phi}}(\bar{\mathbf{C}} - \breve{\mathbf{C}})^\top\big]\big),
\end{equation}
where $\bar{\mathbf{C}}$ and $\bar{\boldsymbol{\Sigma}} = \bar{\boldsymbol{\Sigma}}_2 \otimes \bar{\boldsymbol{\Sigma}}_1$ are formed from $\bar{\boldsymbol{\zeta}}$, and $\widehat{\boldsymbol{\zeta}}$ minimizes $F_T$ over $\mathcal{Z}$.
Their limits are $\mathbf{C}$, $\boldsymbol{\Sigma}$ and $\boldsymbol{\Phi} = \mathbb{E}(\mathbf{x}_t\mathbf{x}_t^\top)$, $\widehat{\boldsymbol{\Phi}} \xrightarrow{p} \boldsymbol{\Phi}$ by the ergodic theorem and
\begin{equation*}
    \breve{\mathbf{C}} - \mathbf{C} = \Big(\frac{1}{T}\sum_{t=1}^{T}\mathbf{u}_t\mathbf{x}_t^\top\Big)\widehat{\boldsymbol{\Phi}}^{-1},
    \qquad
    \breve{\boldsymbol{\Sigma}} - \boldsymbol{\Sigma} = \frac{1}{T}\sum_{t=1}^{T}(\mathbf{u}_t\mathbf{u}_t^\top - \boldsymbol{\Sigma}) - (\breve{\mathbf{C}} - \mathbf{C})\widehat{\boldsymbol{\Phi}}(\breve{\mathbf{C}} - \mathbf{C})^\top
\end{equation*}
are both $O_p(T^{-1/2})$, the first by the argument behind \eqref{eqA:mds} and the second by Chebyshev's inequality under the fourth moments of Assumption \ref{ass:iid-plus}.
Replacing $\breve{\mathbf{C}}$, $\breve{\boldsymbol{\Sigma}}$ and $\widehat{\boldsymbol{\Phi}}$ in \eqref{eqA:F} by their limits gives
\begin{equation*}
    F(\bar{\boldsymbol{\zeta}})
    = \frac{1}{2}\log|\bar{\boldsymbol{\Sigma}}|
    + \frac{1}{2}\operatorname{tr}\big(\bar{\boldsymbol{\Sigma}}^{-1}\big[\boldsymbol{\Sigma} + (\bar{\mathbf{C}} - \mathbf{C})\boldsymbol{\Phi}(\bar{\mathbf{C}} - \mathbf{C})^\top\big]\big).
\end{equation*}

Write $\bar{\boldsymbol{\Omega}} = \bar{\boldsymbol{\Sigma}}^{-1}$ and, as in the proof of Theorem 4 of \citet{chen_autoregressive_2021}, $h(\bar{\boldsymbol{\Omega}}, \mathbf{S}) = -\log|\bar{\boldsymbol{\Omega}}| + \operatorname{tr}(\bar{\boldsymbol{\Omega}}\mathbf{S})$, so that
\begin{equation*}
    2F(\bar{\boldsymbol{\zeta}})
    = h(\bar{\boldsymbol{\Omega}}, \boldsymbol{\Sigma})
    + \operatorname{tr}\big[\bar{\boldsymbol{\Omega}}(\bar{\mathbf{C}} - \mathbf{C})\boldsymbol{\Phi}(\bar{\mathbf{C}} - \mathbf{C})^\top\big].
\end{equation*}
By their Proposition 9, $h(\cdot, \boldsymbol{\Sigma})$ has gradient $\boldsymbol{\Sigma} - \bar{\boldsymbol{\Omega}}^{-1}$ and positive definite Hessian $\bar{\boldsymbol{\Omega}}^{-1} \otimes \bar{\boldsymbol{\Omega}}^{-1}$, so it is strictly convex and uniquely minimized at $\bar{\boldsymbol{\Sigma}} = \boldsymbol{\Sigma}$.
As $\|\bar{\mathbf{A}}_j\| = \|\mathbf{A}_j\| = 1$, $\bar{\mathbf{C}}_j = \mathbf{C}_j$ forces $(\bar{\mathbf{A}}_j, \bar{\mathbf{B}}_j) = \pm(\mathbf{A}_j, \mathbf{B}_j)$, and the same argument with positive definite factors and $\|\bar{\boldsymbol{\Sigma}}_1\| = 1$ forces $\bar{\boldsymbol{\Sigma}}_i = \boldsymbol{\Sigma}_i$.
Hence $\boldsymbol{\zeta}$ is the unique minimizer of $F$ over $\mathcal{Z}$, which contains $\widehat{\boldsymbol{\zeta}}$ because a sign change leaves $F_T$ unchanged.

For consistency, let $\boldsymbol{\Omega} = \boldsymbol{\Sigma}^{-1}$ and $\breve{\boldsymbol{\Omega}} = \breve{\boldsymbol{\Sigma}}^{-1}$, which minimizes $h(\cdot, \breve{\boldsymbol{\Sigma}})$ and converges to $\boldsymbol{\Omega}$.
By \eqref{eqA:F},
\begin{equation*}
    2F_T(\bar{\boldsymbol{\zeta}})
    = h(\bar{\boldsymbol{\Omega}}, \breve{\boldsymbol{\Sigma}})
    + \operatorname{tr}\big[\bar{\boldsymbol{\Omega}}(\bar{\mathbf{C}} - \breve{\mathbf{C}})\widehat{\boldsymbol{\Phi}}(\bar{\mathbf{C}} - \breve{\mathbf{C}})^\top\big]
    \geq h(\bar{\boldsymbol{\Omega}}, \breve{\boldsymbol{\Sigma}}),
\end{equation*}
while $2F_T(\boldsymbol{\zeta}) = h(\breve{\boldsymbol{\Omega}}, \breve{\boldsymbol{\Sigma}}) + o_p(1)$.
Fix $c \in (0, \lambda_{\min}(\boldsymbol{\Omega})/2)$.
A second-order expansion of $h(\cdot, \breve{\boldsymbol{\Sigma}})$ about $\breve{\boldsymbol{\Omega}}$ gives, with probability approaching one, $h(\bar{\boldsymbol{\Omega}}, \breve{\boldsymbol{\Sigma}}) \geq h(\breve{\boldsymbol{\Omega}}, \breve{\boldsymbol{\Sigma}}) + \kappa c^2$ whenever $\|\bar{\boldsymbol{\Omega}} - \boldsymbol{\Omega}\| = c$, with $\kappa > 0$ depending only on $\boldsymbol{\Omega}$, and by convexity the same bound holds whenever $\|\bar{\boldsymbol{\Omega}} - \boldsymbol{\Omega}\| \geq c$.
Every such $\bar{\boldsymbol{\zeta}}$ thus has $F_T(\bar{\boldsymbol{\zeta}}) > F_T(\boldsymbol{\zeta})$, so $\|\widehat{\boldsymbol{\Omega}} - \boldsymbol{\Omega}\| < c$ with probability approaching one.
For the coefficients, let $\bar{\boldsymbol{\zeta}}$ have $\|\bar{\boldsymbol{\Omega}} - \boldsymbol{\Omega}\| < c$, and let $\bar{\boldsymbol{\zeta}}^\circ \in \mathcal{Z}$ have the same covariance factors and the true $(\mathbf{A}, \mathbf{B})$.
Then
\begin{equation*}
    2F_T(\bar{\boldsymbol{\zeta}}) - 2F_T(\bar{\boldsymbol{\zeta}}^\circ)
    \geq \lambda_{\min}(\bar{\boldsymbol{\Omega}})\lambda_{\min}(\widehat{\boldsymbol{\Phi}})\|\bar{\mathbf{C}} - \breve{\mathbf{C}}\|^2
    - \lambda_{\max}(\bar{\boldsymbol{\Omega}})\lambda_{\max}(\widehat{\boldsymbol{\Phi}})\|\mathbf{C} - \breve{\mathbf{C}}\|^2,
\end{equation*}
in which $\lambda_{\min}(\bar{\boldsymbol{\Omega}}) \geq \lambda_{\min}(\boldsymbol{\Omega})/2$ and the last term is $o_p(1)$ uniformly in such $\bar{\boldsymbol{\zeta}}$.
The difference is therefore positive, with probability approaching one so $\widehat{\mathbf{C}} \xrightarrow{p} \mathbf{C}$.
For the factors, $\bar{\mathbf{a}}_j = \operatorname{vec}(\bar{\mathbf{A}}_j)$, the rearrangement of \citet{chen_autoregressive_2021} maps $\bar{\mathbf{C}}_j$ to $\mathbf{M}_j = \bar{\mathbf{a}}_j\operatorname{vec}(\bar{\mathbf{B}}_j)^\top$, and $\|\bar{\mathbf{a}}_j\| = 1$ gives $\bar{\mathbf{a}}_j\bar{\mathbf{a}}_j^\top = \mathbf{M}_j\mathbf{M}_j^\top/\|\mathbf{M}_j\|^2$.
This gives $\widehat{\mathbf{a}}_j \xrightarrow{p} \mathbf{a}_j$ and then $\operatorname{vec}(\widehat{\mathbf{B}}_j) = \widehat{\mathbf{M}}_j^\top\widehat{\mathbf{a}}_j \xrightarrow{p} \operatorname{vec}(\mathbf{B}_j)$.
The covariance factors are handled identically and hence $\widehat{\boldsymbol{\zeta}} \xrightarrow{p} \boldsymbol{\zeta}$.

For the rate, rescaling $(\bar{\mathbf{A}}_j, \bar{\mathbf{B}}_j)$ or $(\bar{\boldsymbol{\Sigma}}_1, \bar{\boldsymbol{\Sigma}}_2)$ leaves $F_T$ unchanged.
Hence $\widehat{\boldsymbol{\zeta}}$ minimizes $F_T$ over the set of all parameters and $\nabla F_T(\widehat{\boldsymbol{\zeta}}) = \mathbf{0}$.
Every term of $\nabla F_T(\boldsymbol{\zeta})$ carries $\breve{\mathbf{C}} - \mathbf{C}$ or $\breve{\boldsymbol{\Sigma}} - \boldsymbol{\Sigma}$, so $\nabla F_T(\boldsymbol{\zeta}) = O_p(T^{-1/2})$.
The normalizations give $\mathbf{a}_j^\top(\widehat{\mathbf{a}}_j - \mathbf{a}_j) = -\frac{1}{2}\|\widehat{\mathbf{a}}_j - \mathbf{a}_j\|^2$ and $\operatorname{tr}(\boldsymbol{\Sigma}_1\Delta_1) = -\frac{1}{2}\|\Delta_1\|^2$.
Adding $\boldsymbol{\Gamma}\boldsymbol{\Gamma}^\top(\widehat{\boldsymbol{\zeta}} - \boldsymbol{\zeta})$ to both sides,
\begin{equation*}
    (\mathbf{H}_{0} + \boldsymbol{\Gamma}\boldsymbol{\Gamma}^\top)(\widehat{\boldsymbol{\zeta}} - \boldsymbol{\zeta})
    = O_p(T^{-1/2}) + O_p\big(\|\widehat{\boldsymbol{\zeta}} - \boldsymbol{\zeta}\|^2\big).
\end{equation*}
Since $\mathbf{H}_{0} + \boldsymbol{\Gamma}\boldsymbol{\Gamma}^\top \xrightarrow{p} \mathbf{H}$, which is positive definite by Section \ref{app:score}, and $\widehat{\boldsymbol{\zeta}} - \boldsymbol{\zeta} = o_p(1)$, we obtain $\|\widehat{\boldsymbol{\zeta}} - \boldsymbol{\zeta}\|\,(1 + o_p(1)) = O_p(T^{-1/2})$.
\end{proof}

Since the sample second moments of $\mathbf{X}_t$ are $O_p(1)$ by Assumption \ref{ass:stationary} and the ergodic theorem, \eqref{eqA:rate} makes any term carrying two deviations $O_p(T^{-1})$, and we drop such terms without comment.
In particular, $\widehat{\boldsymbol{\Sigma}}_j^{-1} = \boldsymbol{\Sigma}_j^{-1} - \boldsymbol{\Sigma}_j^{-1}\Delta_j\boldsymbol{\Sigma}_j^{-1} + O_p(T^{-1})$.
Together with the predeterminedness \eqref{eqA:mds}, these are the only devices the expansion uses.

\begin{proof}[Proof of Proposition \ref{prop:an}]
Differentiating the likelihood in \eqref{eq:loglik} with respect to $(\mathbf{A}, \mathbf{B}, \boldsymbol{\Sigma}_1, \boldsymbol{\Sigma}_2)$ and evaluating at the estimates gives four first-order conditions,
\begin{equation}\label{eqA:foc}
\begin{aligned}
    \sum_{t=1}^{T} \widehat{\boldsymbol{\Sigma}}_1^{-1}\big(\widehat{\mathbf{A}}\mathbf{X}_t\widehat{\mathbf{B}}^\top - \mathbf{Y}_t\big)\widehat{\boldsymbol{\Sigma}}_2^{-1}\widehat{\mathbf{B}}\mathbf{X}_t^\top &= \mathbf{0}, \\
    \sum_{t=1}^{T} \mathbf{X}_t^\top\widehat{\mathbf{A}}^\top\widehat{\boldsymbol{\Sigma}}_1^{-1}\big(\widehat{\mathbf{A}}\mathbf{X}_t\widehat{\mathbf{B}}^\top - \mathbf{Y}_t\big)\widehat{\boldsymbol{\Sigma}}_2^{-1} &= \mathbf{0}, \\
    -T N_2 \widehat{\boldsymbol{\Sigma}}_1^{-1}
      + \sum_{t=1}^{T}\widehat{\boldsymbol{\Sigma}}_1^{-1}\widehat{\mathbf{U}}_t\widehat{\boldsymbol{\Sigma}}_2^{-1}\widehat{\mathbf{U}}_t^\top\widehat{\boldsymbol{\Sigma}}_1^{-1} &= \mathbf{0}, \\
    -T N_1 \widehat{\boldsymbol{\Sigma}}_2^{-1}
      + \sum_{t=1}^{T}\widehat{\boldsymbol{\Sigma}}_2^{-1}\widehat{\mathbf{U}}_t^\top\widehat{\boldsymbol{\Sigma}}_1^{-1}\widehat{\mathbf{U}}_t\widehat{\boldsymbol{\Sigma}}_2^{-1} &= \mathbf{0},
\end{aligned}
\end{equation}
where $\widehat{\mathbf{U}}_t = \mathbf{Y}_t - \widehat{\mathbf{A}}\mathbf{X}_t\widehat{\mathbf{B}}^\top$ and the second line is the transpose of the gradient with respect to $\mathbf{B}$, taken so that it aligns with the first.

\paragraph{The coefficient conditions.}
Start with the first-order condition for $\mathbf{A}$.
Substituting $\mathbf{Y}_t = \mathbf{A}\mathbf{X}_t\mathbf{B}^\top + \mathbf{U}_t$,
$\widehat{\mathbf{A}} = \mathbf{A} + \Delta_A$ and $\widehat{\mathbf{B}} = \mathbf{B} + \Delta_B$,
the fitted residual is
\begin{equation}\label{eqA:resid}
    \widehat{\mathbf{A}}\mathbf{X}_t\widehat{\mathbf{B}}^\top - \mathbf{Y}_t
    = \Delta_A\mathbf{X}_t\mathbf{B}^\top
    + \mathbf{A}\mathbf{X}_t\Delta_B^\top
    + \Delta_A\mathbf{X}_t\Delta_B^\top
    - \mathbf{U}_t,
\end{equation}
the level term $\mathbf{A}\mathbf{X}_t\mathbf{B}^\top$ having dropped out.
Dividing by $T$ and moving the innovation to the right, the first-order condition for $\mathbf{A}$ is the exact rearrangement, with the estimates $\widehat{\boldsymbol{\Sigma}}_j$ still in place,
\begin{align*}
    \frac{1}{T}\sum_t \widehat{\boldsymbol{\Sigma}}_1^{-1}
    \big(\Delta_A\mathbf{X}_t\mathbf{B}^\top
       + \mathbf{A}\mathbf{X}_t\Delta_B^\top
       + \Delta_A\mathbf{X}_t\Delta_B^\top\big)
    \widehat{\boldsymbol{\Sigma}}_2^{-1}(\mathbf{B} + \Delta_B)\mathbf{X}_t^\top \\
    = \frac{1}{T}\sum_t \widehat{\boldsymbol{\Sigma}}_1^{-1}\mathbf{U}_t
    \widehat{\boldsymbol{\Sigma}}_2^{-1}(\mathbf{B} + \Delta_B)\mathbf{X}_t^\top.
\end{align*}
Expanding the factor $(\mathbf{B} + \Delta_B)$ and the two $\widehat{\boldsymbol{\Sigma}}_j^{-1}$, every term on the left carries at least one deviation, and all but the two in which $\Delta_A$ or $\Delta_B$ appears alone carry two.
On the right, \eqref{eqA:mds} makes the average $O_p(T^{-1/2})$, so every term there that carries a deviation is again of the order of a product of two.
What remains is
\begin{equation*}
    \frac{1}{T}\sum_t\Big(\boldsymbol{\Sigma}_1^{-1}\Delta_A\mathbf{X}_t\mathbf{B}^\top\boldsymbol{\Sigma}_2^{-1}\mathbf{B}\mathbf{X}_t^\top
    + \boldsymbol{\Sigma}_1^{-1}\mathbf{A}\mathbf{X}_t\Delta_B^\top\boldsymbol{\Sigma}_2^{-1}\mathbf{B}\mathbf{X}_t^\top\Big)
    = \frac{1}{T}\sum_t \boldsymbol{\Sigma}_1^{-1}\mathbf{U}_t\boldsymbol{\Sigma}_2^{-1}\mathbf{B}\mathbf{X}_t^\top + O_p(T^{-1}).
\end{equation*}
Vectorizing with Lemma \ref{lem:kron}.1 and using the symmetry of $\boldsymbol{\Sigma}_2$,
\begin{small}
\begin{align*}
    \frac{1}{T}\sum_t\Big[&\big(\mathbf{X}_t\mathbf{B}^\top\boldsymbol{\Sigma}_2^{-1}\mathbf{B}\mathbf{X}_t^\top\otimes\boldsymbol{\Sigma}_1^{-1}\big)\operatorname{vec}(\Delta_A)
    + \big(\mathbf{X}_t\mathbf{B}^\top\boldsymbol{\Sigma}_2^{-1}\otimes\boldsymbol{\Sigma}_1^{-1}\mathbf{A}\mathbf{X}_t\big)\operatorname{vec}(\Delta_B^\top)\Big] \\
    &= \frac{1}{T}\sum_t\big(\mathbf{X}_t\mathbf{B}^\top\boldsymbol{\Sigma}_2^{-1}\otimes\boldsymbol{\Sigma}_1^{-1}\big)\operatorname{vec}(\mathbf{U}_t) + O_p(T^{-1}).
\end{align*}
\end{small}
Using \eqref{eqA:resid} again in the first-order condition for $\mathbf{B}$ and proceeding
identically,
\begin{small}
\begin{align*}
    &\frac{1}{T} \sum_t \left( (\boldsymbol{\Sigma}_{2}^{-1} \mathbf{B} \mathbf{X}_t^{\top} \otimes \mathbf{X}_t^{\top} \mathbf{A}^{\top} \boldsymbol{\Sigma}_{1}^{-1}) \operatorname{vec}(\Delta_A) + (\boldsymbol{\Sigma}_{2}^{-1} \otimes \mathbf{X}_t^{\top} \mathbf{A}^{\top} \boldsymbol{\Sigma}_{1}^{-1} \mathbf{A} \mathbf{X}_t) \operatorname{vec}(\Delta_B^{\top}) \right) \\
    &= \frac{1}{T} \sum_t (\boldsymbol{\Sigma}_{2}^{-1} \otimes \mathbf{X}_t^{\top} \mathbf{A}^{\top} \boldsymbol{\Sigma}_{1}^{-1}) \operatorname{vec}(\mathbf{U}_{t}) + O_{p}(T^{-1}).
\end{align*}
\end{small}
With $\boldsymbol{\Sigma} = \boldsymbol{\Sigma}_2 \otimes \boldsymbol{\Sigma}_1$ and $\mathbf{W}_t$ as in \eqref{eq:score_beta}, the two first-order conditions stack as
\begin{equation}\label{eqA:coefsys}
    \frac{1}{T}\sum_{t=1}^{T}\mathbf{W}_t\boldsymbol{\Sigma}^{-1}\mathbf{W}_t^\top
    \begin{bmatrix}
        \operatorname{vec}(\Delta_{A}) \\
        \operatorname{vec}(\Delta_{B}^{\top}) \\
    \end{bmatrix}
    = \frac{1}{T}\sum_{t=1}^{T}\mathbf{W}_t\boldsymbol{\Sigma}^{-1}\operatorname{vec}(\mathbf{U}_t)
      + O_p(T^{-1}).
\end{equation}

\paragraph{The covariance conditions.}
We next turn to the first-order conditions associated with the covariances.
Write $\mathbf{S}_1 = T^{-1}\sum_t \mathbf{U}_t\boldsymbol{\Sigma}_2^{-1}\mathbf{U}_t^\top$.
By \eqref{eqA:UMU}, $\mathbb{E}(\mathbf{S}_1) = N_{2} \boldsymbol{\Sigma}_{1}$.
The estimated residual is $\widehat{\mathbf{U}}_t = \mathbf{U}_t - \boldsymbol{\Upsilon}_t$ with $\boldsymbol{\Upsilon}_t = \Delta_A\mathbf{X}_t\mathbf{B}^\top + \mathbf{A}\mathbf{X}_t\Delta_B^\top + O_p(T^{-1})$, which is the only place the lag order enters.

Divide  the first-order condition for $\boldsymbol{\Sigma}_1$ by $T/2$ and expand the four hatted objects in turn, replacing any sample average that already multiplies a deviation by its probability limit, which is $O_p(T^{-1})$,
\begin{small}
\begin{align*}
    \frac{1}{T}\sum_t\widehat{\boldsymbol{\Sigma}}_1^{-1}\widehat{\mathbf{U}}_t\widehat{\boldsymbol{\Sigma}}_2^{-1}\widehat{\mathbf{U}}_t^\top\widehat{\boldsymbol{\Sigma}}_1^{-1}
    &= \boldsymbol{\Sigma}_1^{-1}\mathbf{S}_1\boldsymbol{\Sigma}_1^{-1}
     - 2N_2\boldsymbol{\Sigma}_1^{-1}\Delta_1\boldsymbol{\Sigma}_1^{-1}
     - \boldsymbol{\Sigma}_1^{-1}\Big[\frac{1}{T}\sum_t\mathbf{U}_t\boldsymbol{\Sigma}_2^{-1}\Delta_2\boldsymbol{\Sigma}_2^{-1}\mathbf{U}_t^\top\Big]\boldsymbol{\Sigma}_1^{-1} \\
    &\quad - \boldsymbol{\Sigma}_1^{-1}\Big[\frac{1}{T}\sum_t\big(\boldsymbol{\Upsilon}_t\boldsymbol{\Sigma}_2^{-1}\mathbf{U}_t^\top + \mathbf{U}_t\boldsymbol{\Sigma}_2^{-1}\boldsymbol{\Upsilon}_t^\top\big)\Big]\boldsymbol{\Sigma}_1^{-1}
    + O_p(T^{-1}),
\end{align*}
\end{small}
the factor $2N_2$ arising because the two occurrences of $\widehat{\boldsymbol{\Sigma}}_1^{-1}$ each contribute
\begin{equation*}
    \boldsymbol{\Sigma}_1^{-1}\Delta_1\boldsymbol{\Sigma}_1^{-1}(N_2\boldsymbol{\Sigma}_1)\boldsymbol{\Sigma}_1^{-1}
    = N_2 \boldsymbol{\Sigma}_1^{-1}\Delta_1\boldsymbol{\Sigma}_1^{-1}.
\end{equation*}
Applying \eqref{eqA:UMU} with the symmetric $\mathbf{M} = \boldsymbol{\Sigma}_2^{-1}\Delta_2\boldsymbol{\Sigma}_2^{-1}$ reduces the $\Delta_2$ term to $\operatorname{tr}(\boldsymbol{\Sigma}_2^{-1}\Delta_2)\boldsymbol{\Sigma}_1 + O_p(T^{-1})$.
For the residual term, the randomness in $t$ is carried by $\mathbf{X}_t$, so
\begin{equation*}
    \frac{1}{T}\sum_t \boldsymbol{\Upsilon}_t\boldsymbol{\Sigma}_2^{-1}\mathbf{U}_t^\top
    = \Delta_A\Big[\frac{1}{T}\sum_t\mathbf{X}_t\mathbf{B}^\top\boldsymbol{\Sigma}_2^{-1}\mathbf{U}_t^\top\Big]
    + \mathbf{A}\Big[\frac{1}{T}\sum_t\mathbf{X}_t\Delta_B^\top\boldsymbol{\Sigma}_2^{-1}\mathbf{U}_t^\top\Big]
    + O_p(T^{-1}),
\end{equation*}
and each bracket is $O_p(T^{-1/2})$ by \eqref{eqA:mds}.
Multiplied by $\Delta_A$ or $\Delta_B$ the whole term is $O_p(T^{-1})$ and drops.
Had $\mathbb{E}(\mathbf{X}_t\mathbf{U}_t^\top)$ been nonzero the brackets would be $O_p(1)$ and the term would survive at order $O_p(T^{-1/2})$.
This is the exact counterpart of the discard made in Step 1, and it is why estimating $\mathbf{A}$ and $\mathbf{B}$ has no first-order effect on the covariance factors.

Collecting and pre- and post-multiplying by $\boldsymbol{\Sigma}_1$, which rescales a linear equation without changing its solution,
\begin{equation}\label{eqA:cov1}
    N_2\Delta_1 + \operatorname{tr}\big(\boldsymbol{\Sigma}_2^{-1}\Delta_2\big)\boldsymbol{\Sigma}_1
    = \frac{1}{T}\sum_t \mathbf{G}_{1t} + O_p(T^{-1}),
    \qquad
    \mathbf{G}_{1t} = \mathbf{U}_t\boldsymbol{\Sigma}_2^{-1}\mathbf{U}_t^\top - N_2\boldsymbol{\Sigma}_1 .
\end{equation}
The  first-order condition for $\boldsymbol{\Sigma}_2$ is the mirror image with the two dimensions exchanged,
\begin{equation}\label{eqA:cov2}
    N_1\Delta_2 + \operatorname{tr}\big(\boldsymbol{\Sigma}_1^{-1}\Delta_1\big)\boldsymbol{\Sigma}_2
    = \frac{1}{T}\sum_t \mathbf{G}_{2t} + O_p(T^{-1}),
    \qquad
    \mathbf{G}_{2t} = \mathbf{U}_t^\top\boldsymbol{\Sigma}_1^{-1}\mathbf{U}_t - N_1\boldsymbol{\Sigma}_2,
\end{equation}
and $\mathbf{G}_{1t}$, $\mathbf{G}_{2t}$ are symmetric with mean zero by \eqref{eqA:UMU}.
The two equations are coupled only through the scalars $\operatorname{tr}(\boldsymbol{\Sigma}_2^{-1}\Delta_2)$ and $\operatorname{tr}(\boldsymbol{\Sigma}_1^{-1}\Delta_1)$.

Undoing the rescaling by $\boldsymbol{\Sigma}_1$ and $\boldsymbol{\Sigma}_2$ returns \eqref{eqA:cov1}--\eqref{eqA:cov2} to the form of the raw score \eqref{eq:score_sigma}: with $\mathbf{c}_i = \mathbf{D}_{N_i}^\top\operatorname{vec}(\boldsymbol{\Sigma}_i^{-1})$,
\begin{equation}\label{eqA:covraw}
\begin{gathered}
    \mathbf{H}_{0,\sigma}
    \begin{bmatrix}
        \operatorname{vech}(\Delta_{1}) \\
        \operatorname{vech}(\Delta_{2})
    \end{bmatrix}
    = \frac{1}{T}\sum_{t=1}^{T}\mathbf{s}_{\sigma,t} + O_p(T^{-1}), \\[4pt]
    \mathbf{H}_{0,\sigma}
    = \frac12
    \begin{bmatrix}
        N_2 \mathbf{D}_{N_1}^\top(\boldsymbol{\Sigma}_1^{-1}\otimes\boldsymbol{\Sigma}_1^{-1})\mathbf{D}_{N_1}
        & \mathbf{c}_1\mathbf{c}_2^\top \\[3pt]
        \mathbf{c}_2\mathbf{c}_1^\top
        & N_1 \mathbf{D}_{N_2}^\top(\boldsymbol{\Sigma}_2^{-1}\otimes\boldsymbol{\Sigma}_2^{-1})\mathbf{D}_{N_2}
    \end{bmatrix}.
\end{gathered}
\end{equation}
Each equation is pre- and post-multiplied by $\boldsymbol{\Sigma}_i^{-1}$ and then vectorized with $\tfrac{1}{2}\mathbf{D}_{N_i}^\top$.
Both steps can be undone on symmetric arguments, so the two forms have the same solution.

\paragraph{The normalizations.}
The $p+1$ normalizations of Section \ref{app:score}, $\|\mathbf{A}_j\| = 1$ for
$j = 1, \ldots, p$ and $\|\boldsymbol{\Sigma}_1\| = 1$, are imposed on the estimates as
well as on the true parameters.
From $\|\widehat{\mathbf{A}}_j\| = \|\mathbf{A}_j\| = 1$,
\begin{equation*}
    \|\widehat{\mathbf{a}}_j - \mathbf{a}_j\|^2 = 2 - 2\widehat{\mathbf{a}}_j^\top\mathbf{a}_j
    \qquad\Longrightarrow\qquad
    \mathbf{a}_j^\top(\widehat{\mathbf{a}}_j - \mathbf{a}_j) = -\frac{1}{2}\|\widehat{\mathbf{a}}_j - \mathbf{a}_j\|^2 = O_p(T^{-1}),
\end{equation*}
and identically, from $\|\widehat{\boldsymbol{\Sigma}}_1\| = \|\boldsymbol{\Sigma}_1\| = 1$, $\langle\boldsymbol{\Sigma}_1, \Delta_1\rangle = -\frac{1}{2}\|\Delta_1\|^2 = O_p(T^{-1})$, where $\langle\mathbf{X},\mathbf{Y}\rangle = \operatorname{tr}(\mathbf{X}\mathbf{Y})$ for symmetric arguments.
Collecting the $p+1$ gradients in the matrix $\boldsymbol{\Gamma}$ of \eqref{eq:Gamma}, whose covariance row is $\mathbf{D}_{N_1}^\top\operatorname{vec}(\boldsymbol{\Sigma}_1)$ because the trace counts off-diagonal entries twice, these read
\begin{equation*}
    \boldsymbol{\Gamma}^\top (\widehat{\boldsymbol{\zeta}} - \boldsymbol{\zeta}) = O_p(T^{-1})
    \qquad\text{componentwise.}
\end{equation*}
Completing the system by premultiplying the stacked system by $\begin{bmatrix}\mathbf{I} & \boldsymbol{\Gamma} \end{bmatrix}$ adds $\boldsymbol{\Gamma}\boldsymbol{\Gamma}^\top(\widehat{\boldsymbol{\zeta}} - \boldsymbol{\zeta})$ to the left-hand side at a cost of $O_p(T^{-1})$, which is the rank-$(p+1)$ correction $\mathbf{H} = \mathbf{H}_0 +
\boldsymbol{\Gamma}\boldsymbol{\Gamma}^\top$ of Section \ref{app:score}.

\paragraph{The stacked system and the limit.}
Combining \eqref{eqA:coefsys} and \eqref{eqA:covraw}, replacing the coefficient sample average by its limit and adding the correction of Step 3,
\begin{equation}\label{eqA:stacked}
    \mathbf{H} (\widehat{\boldsymbol{\zeta}} - \boldsymbol{\zeta})
    = \frac{1}{T}\sum_{t=1}^{T}\mathbf{s}_t + O_p(T^{-1}),
\end{equation}
with $\mathbf{s}_t = (\mathbf{s}_{\beta,t}^\top, \mathbf{s}_{\sigma,t}^\top)^\top$ the score of \eqref{eq:score_beta}--\eqref{eq:score_sigma}.
The system is block diagonal because the two off-diagonal blocks were discarded in Steps 1 and 2, both as consequences of $\mathbb{E}(\mathbf{X}_t\mathbf{U}_t^\top) = \mathbf{0}$.

By construction $\{\mathbf{s}_t, \mathcal{F}_t\}$ is a stationary ergodic martingale difference sequence.
$\mathbf{W}_t$ is $\mathcal{F}_{t-1}$-measurable with $\mathbb{E}(\mathbf{u}_t\mid\mathcal{F}_{t-1}) = \mathbf{0}$, and $\mathbb{E}(\mathbf{G}_{it}) = \mathbf{0}$ by \eqref{eqA:UMU}.
Its second moment $\boldsymbol{\Lambda} = \mathbb{E}(\mathbf{s}_t\mathbf{s}_t^\top)$ is finite under Assumption \ref{ass:iid-plus}, the coefficient block needing only $\mathbb{E}\|\mathbf{X}_t\|^2 \mathbb{E}\|\mathbf{U}_t\|^2$ and the covariance block, being quadratic in $\mathbf{U}_t$, the fourth moment.
Applying the martingale central limit theorem of \citet{billingsley_lindeberg-levy_1961}, $T^{-1/2}\sum_t\mathbf{s}_t \xrightarrow{d} N(\mathbf{0}, \boldsymbol{\Lambda})$, so that \eqref{eqA:stacked} gives \eqref{eq:an} with $\boldsymbol{\Omega} = \mathbf{H}^{-1}\boldsymbol{\Lambda}\mathbf{H}^{-1}$, which is \eqref{eq:Omega}.

In the coefficient block, the matrix on the left of \eqref{eqA:coefsys} converges to $\boldsymbol{\Xi} = \mathbb{E}(\mathbf{W}_t\boldsymbol{\Sigma}^{-1}\mathbf{W}_t^\top)$, so $\mathbf{H}_{0,\beta} = \boldsymbol{\Xi}$, and $\mathbb{E}(\mathbf{u}_t\mathbf{u}_t^\top \mid \mathcal{F}_{t-1}) = \boldsymbol{\Sigma}$ gives $\boldsymbol{\Lambda}_\beta = \boldsymbol{\Xi}$ as well.
Hence,
\begin{align*}
\boldsymbol{\Omega}_\beta = \mathbf{H}_\beta^{-1}\boldsymbol{\Xi}\mathbf{H}_\beta^{-1} \qquad \text{with} \qquad \mathbf{H}_\beta = \boldsymbol{\Xi} + \boldsymbol{\Gamma}_\beta\boldsymbol{\Gamma}_\beta^\top
\end{align*}

\end{proof}

\subsubsection{From \texorpdfstring{$(\widehat{\mathbf{A}}, \widehat{\mathbf{B}})$}{(A,B)} to \texorpdfstring{$\operatorname{vec}(\widehat{\mathbf{C}})$}{vec(C)}}\label{app:jacobian}

Proposition \ref{prop:an} delivers the joint 
distribution for $\widehat{\boldsymbol{\zeta}}$, but impulse responses depend on the coefficients only through $\mathbf{C} = [\mathbf{C}_1 \cdots \mathbf{C}_p]$.
The passage from one to the other is a composition of four linear maps, and identifies the Jacobian $\mathbf{J}_\beta$ of \eqref{eq:OmegaC}.

First, the commutation matrix $\mathbf{K} = \operatorname{blkdiag}(\mathbf{I}_{N_1^2p}, \mathbf{K}^*)$, with $\mathbf{K}^*$ defined by $\operatorname{vec}(\mathbf{B}) = \mathbf{K}^*\operatorname{vec}(\mathbf{B}^\top)$, converts the second block of $\boldsymbol{\beta}$ from $\operatorname{vec}(\mathbf{B}^\top)$ to $\operatorname{vec}(\mathbf{B})$.
Second, expanding $\operatorname{vec}(\widehat{\mathbf{B}})\otimes\operatorname{vec}(\widehat{\mathbf{A}})$ about $\operatorname{vec}(\mathbf{B})\otimes\operatorname{vec}(\mathbf{A})$ and using $\mathbf{x}\otimes\mathbf{y} = (\mathbf{I}\otimes\mathbf{y})\mathbf{x}$,
\begin{equation*}
    \operatorname{vec}(\widehat{\mathbf{B}})\otimes\operatorname{vec}(\widehat{\mathbf{A}})
    - \operatorname{vec}(\mathbf{B})\otimes\operatorname{vec}(\mathbf{A})
    = \mathbf{V}
    \begin{bmatrix}
        \operatorname{vec}(\Delta_A) \\ \operatorname{vec}(\Delta_B)
    \end{bmatrix}
    + \operatorname{vec}(\Delta_B)\otimes\operatorname{vec}(\Delta_A),
\end{equation*}
where $\mathbf{V} = \begin{bmatrix} \operatorname{vec}(\mathbf{B})\otimes\mathbf{I}_{N_1^2p} & \mathbf{I}_{N_2^2p}\otimes\operatorname{vec}(\mathbf{A})\end{bmatrix}$ and the last term is $O_p(T^{-1})$.
Third, the permutation $\boldsymbol{\Pi}$ defined by
\begin{equation*}
    \boldsymbol{\Pi}^\top\operatorname{vec}(\mathbf{B}\otimes\mathbf{A})
    = \operatorname{vec}(\mathbf{B})\otimes\operatorname{vec}(\mathbf{A})
\end{equation*}
returns the vectorized Kronecker product itself \citep[Sec. 3]{van_loan_approximation_1993}.
Together these give
\begin{equation}\label{eq:an_kron_p}
    \sqrt{T}\big(\operatorname{vec}(\widehat{\mathbf{B}}\otimes\widehat{\mathbf{A}})
    - \operatorname{vec}(\mathbf{B}\otimes\mathbf{A})\big)
    \xrightarrow{d} N\big(\mathbf{0},
    \boldsymbol{\Pi}\mathbf{V}\mathbf{K} \boldsymbol{\Omega}_\beta
    \mathbf{K}^\top\mathbf{V}^\top\boldsymbol{\Pi}^\top\big).
\end{equation}

The fourth map has no counterpart at $p = 1$.
The full Kronecker product $\mathbf{B}\otimes\mathbf{A}\in\mathbb{R}^{N_1N_2\times N_1N_2p^2}$ contains cross-lag blocks $\mathbf{B}_i\otimes\mathbf{A}_j$ for $i\neq j$, which the MAR($p$) does not admit, so only the diagonal blocks $\mathbf{C}_j = \mathbf{B}_j\otimes\mathbf{A}_j$ are to be retained.
Define the selection matrix $\mathbf{Q}\in\mathbb{R}^{N_1N_2p\times N_1N_2p^2}$ satisfying
\begin{equation}\label{eq:Q_def}
    \mathbf{Q}\operatorname{vec}(\mathbf{X}_t) =
    \begin{bmatrix}
        \operatorname{vec}(\mathbf{Y}_{t-1}) \\ \vdots \\ \operatorname{vec}(\mathbf{Y}_{t-p})
    \end{bmatrix}
    = \mathbf{x}_t \in \mathbb{R}^{N_1N_2p},
\end{equation}
which keeps the rows of $\operatorname{vec}(\mathbf{X}_t)$ belonging to the diagonal blocks of the block-diagonal regressor and discards the zeros.
The columns of $\mathbf{Q}^\top$ select exactly the columns of $\mathbf{B}\otimes\mathbf{A}$ belonging to those blocks, so $(\mathbf{B}\otimes\mathbf{A})\mathbf{Q}^\top = \mathbf{C}$ and $\operatorname{vec}(\mathbf{C}) = (\mathbf{Q}\otimes\mathbf{I}_{N_1N_2})\operatorname{vec}(\mathbf{B}\otimes\mathbf{A})$.
Applying this to \eqref{eq:an_kron_p},
\begin{equation}\label{eq:an_C_app}
    \sqrt{T}\big(\operatorname{vec}(\widehat{\mathbf{C}}) - \operatorname{vec}(\mathbf{C})\big)
    \xrightarrow{d} N\big(\mathbf{0}, \boldsymbol{\Omega}_C\big),
    \qquad
    \boldsymbol{\Omega}_C = \mathbf{J}_\beta \boldsymbol{\Omega}_\beta \mathbf{J}_\beta^\top,
\end{equation}
which is \eqref{eq:OmegaC} with the Jacobian identified as
\begin{equation*}
    \mathbf{J}_\beta
    = \frac{\partial\operatorname{vec}(\mathbf{C})}{\partial\boldsymbol{\beta}^\top}
    = (\mathbf{Q}\otimes\mathbf{I}_{N_1N_2}) \boldsymbol{\Pi} \mathbf{V} \mathbf{K}
    \in \mathbb{R}^{(N_1N_2)^2p \times p(N_1^2+N_2^2)}.
\end{equation*}
Since $\mathbf{J}_\beta$ maps a space of dimension $p(N_1^2+N_2^2)$ into one of dimension $p(N_1N_2)^2$, $\boldsymbol{\Omega}_C$ is singular, with rank at most $p(N_1^2+N_2^2-1)$ once the $p$ flat directions of $\boldsymbol{\Omega}_\beta$ are accounted for.
When $p = 1$ we have $\mathbf{Q} = \mathbf{I}$ and \eqref{eq:an_C_app} reduces to the Kronecker limit of \citet{chen_autoregressive_2021}.

\subsection{Bootstrap Inference}\label{app:bootstrap}

\subsubsection{Bootstrap Setup and Preliminary Lemmas}\label{app:lemmas}

We write
\begin{equation}\label{eqB:Semp}
    \widehat{\boldsymbol{\Sigma}}_{emp} = \frac{1}{T}\sum_{t=1}^{T}\widehat{\mathbf{u}}_t\widehat{\mathbf{u}}_t^\top,
    \qquad
    \widehat{\mathbf{S}}_T = \widehat{\boldsymbol{\Sigma}}_{emp} - \widehat{\boldsymbol{\Sigma}}_2 \otimes \widehat{\boldsymbol{\Sigma}}_1.
\end{equation}
The bootstrap sample is generated recursively as
\begin{equation}\label{eqB:dgp}
    \mathbf{y}^{*}_{t} = \sum_{j=1}^{p}\widehat{\mathbf{C}}_j \mathbf{y}^{*}_{t-j} + \mathbf{u}^{*}_{t},
    \qquad t = 1,\ldots,T,
    \qquad \mathbf{y}^{*}_{-p+1} = \cdots = \mathbf{y}^{*}_{0} = \mathbf{0},
\end{equation}
with $\mathbf{u}^{*}_{1},\ldots,\mathbf{u}^{*}_{T}$ drawn independently from the empirical distribution of $\{\widehat{\mathbf{u}}_{t}\}_{t=1}^{T}$.
The MMA coefficients \eqref{eq:mma} propagate a shock to $\mathbf{y}_t$.
The corresponding coefficients for the regressor stack are obtained by collecting $p$ of them,
\begin{equation}\label{eqB:Rdef}
    \mathbf{D}_{j} =
    \begin{bmatrix}
        \boldsymbol{\Theta}_{j-1}^{\top} & \cdots & \boldsymbol{\Theta}_{j-p}^{\top}
    \end{bmatrix}^{\top}
    \in \mathbb{R}^{pN \times N},
    \qquad N = N_1 N_2,
\end{equation}
with the convention $\boldsymbol{\Theta}_{h} = \mathbf{0}$ for $h < 0$, so that $\mathbf{x}_{t} = \sum_{j \geq 1} \mathbf{D}_{j}\mathbf{u}_{t-j}$ and $\|\mathbf{D}_{j}\|^{2} = \sum_{i=1}^{p}\|\boldsymbol{\Theta}_{j-i}\|^{2}$.
We write $\widehat{\boldsymbol{\Theta}}_{j}$ and $\widehat{\mathbf{D}}_{j}$ for the same objects formed from $\widehat{\mathbf{C}}$.
We write $\mathbf{X}^{*}_t$, $\mathbf{x}^{*}_t$, $\mathbf{W}^{*}_t$, $\mathbf{s}^{*}_t$, $\mathbf{H}^{*}$, $\boldsymbol{\Gamma}^{*}$ for the bootstrap counterparts of the objects of Appendix \ref{app:asymptotic}, all evaluated at $\widehat{\boldsymbol{\zeta}}$, and set $\mathcal{F}^{*}_{t} = \sigma(\mathbf{u}^{*}_{1},\ldots,\mathbf{u}^{*}_{t})$.
Likewise, $\Delta^{*}_{A} = \widehat{\mathbf{A}}^{*} - \widehat{\mathbf{A}}$, $\Delta^{*}_{B} = \widehat{\mathbf{B}}^{*} - \widehat{\mathbf{B}}$ and $\Delta^{*}_{i} = \widehat{\boldsymbol{\Sigma}}^{*}_{i} - \widehat{\boldsymbol{\Sigma}}_{i}$ ($i = 1, 2$) denote the bootstrap deviations, and $\boldsymbol{\Upsilon}^{*}_{t} = \Delta^{*}_{A}\mathbf{X}^{*}_{t}\widehat{\mathbf{B}}^{\top} + \widehat{\mathbf{A}}\mathbf{X}^{*}_{t}\Delta^{*\top}_{B} + \Delta^{*}_{A}\mathbf{X}^{*}_{t}\Delta^{*\top}_{B}$ the counterpart of $\boldsymbol{\Upsilon}_{t}$, so that $\widehat{\mathbf{U}}^{*}_{t} = \mathbf{U}^{*}_{t} - \boldsymbol{\Upsilon}^{*}_{t}$.
For a bootstrap quantity $\mathbf{Z}^{*}_{T}$ we write $\mathbf{Z}^{*}_{T} \xrightarrow{p^{*}}_{p} \mathbf{z}$ when $P^{*}(\|\mathbf{Z}^{*}_{T}-\mathbf{z}\|>\epsilon) \xrightarrow{p} 0$ for every $\epsilon>0$, and $\mathbf{Z}^{*}_{T} = o_{p^{*}}(1)$ for the same statement with $\mathbf{z}=\mathbf{0}$.
We use without further comment the conditional Markov and Chebyshev inequalities, which turn a statement about $\mathbb{E}^{*}$ and $\operatorname{Var}^{*}$ into one about $\xrightarrow{p^{*}}_{p}$; the conditional Slutsky lemma and continuous mapping theorem; the martingale central limit theorem for triangular arrays of \citet{brown_martingale_1971}, which replaces \citet{billingsley_lindeberg-levy_1961} because $\{\mathbf{y}^{*}_t\}$ is not stationary; and P\'olya's theorem in probability.

\begin{lemma}[Geometric decay]\label{lem:decay}
Let Assumption \ref{ass:stationary} hold and let $\mathbf{F}$, $\widehat{\mathbf{F}}$ be the companion matrices of $\mathbf{C}$, $\widehat{\mathbf{C}}$.
For every $\bar{\rho} \in (\rho(\mathbf{F}), 1)$ there is a finite $K$ such that, with probability approaching one,
\begin{equation}\label{eqB:decay}
    \max\big\{\|\boldsymbol{\Theta}_{j}\|,\; \|\widehat{\boldsymbol{\Theta}}_{j}\|,\; \|\mathbf{D}_{j}\|,\; \|\widehat{\mathbf{D}}_{j}\|\big\}
    \leq K\bar{\rho}^{\,j},
    \qquad j \geq 0 .
\end{equation}
\end{lemma}

\begin{proof}
Set $\varepsilon = \tfrac{1}{2}(\bar{\rho} - \rho(\mathbf{F}))$.
By the proof of Lemma 5.6.10 in \citet{horn_matrix_2013} there is a nonsingular
$\mathbf{S} \in \mathbb{C}^{pN \times pN}$ with
$\|\mathbf{S}^{-1}\mathbf{F}\mathbf{S}\|_{1} \leq \rho(\mathbf{F}) + \varepsilon$.
Hence, with $\kappa = \|\mathbf{S}\|_{1}\|\mathbf{S}^{-1}\|_{1}$, every $\mathbf{M}$ with
$\kappa\|\mathbf{M} - \mathbf{F}\|_{1} \leq \varepsilon$ satisfies
\begin{align*}
    \big\|\mathbf{S}^{-1}\mathbf{M}\mathbf{S}\big\|_{1} &= \|\mathbf{S}^{-1} \mathbf{M} \mathbf{S} - \mathbf{S}^{-1} \mathbf{F} \mathbf{S} + \mathbf{S}^{-1} \mathbf{F} \mathbf{S} \|_1 \\
							&\leq \| \mathbf{S}^{-1} (\mathbf{M} - \mathbf{F}) \mathbf{S} \|_1 + \| \mathbf{S}^{-1} \mathbf{F} \mathbf{S} \|_1 \\
							&\leq \| \mathbf{S} \|_1 \| \mathbf{S}^{-1} \|_1 \| \mathbf{M} - \mathbf{F} \|_1 + \| \mathbf{S}^{-1} \mathbf{F} \mathbf{S} \|_1 \\
    &= \big\|\mathbf{S}^{-1}\mathbf{F}\mathbf{S}\big\|_{1} + \kappa\big\|\mathbf{M} - \mathbf{F}\big\|_{1} \\
    &\leq \rho(\mathbf{F}) + 2\varepsilon = \bar{\rho}.
\end{align*}
Taking powers of $\mathbf{M}$ yields
\begin{align*}
    \big\|\mathbf{M}^{j}\big\|_{1} = \big\|\mathbf{S}\big(\mathbf{S}^{-1}\mathbf{M}\mathbf{S}\big)^{j}\mathbf{S}^{-1}\big\|_{1}
    \leq \kappa\bar{\rho}^{\,j}.
\end{align*}
Proposition \ref{prop:an} gives $\widehat{\mathbf{C}} \xrightarrow{p} \mathbf{C}$, so
$P(\kappa\|\widehat{\mathbf{F}} - \mathbf{F}\|_{1} \leq \varepsilon) \to 1$, and \eqref{eqB:decay}
follows from $\boldsymbol{\Theta}_{j} = \mathbf{J}\mathbf{F}^{j}\mathbf{J}^{\top}$ with
$\mathbf{J} = [\mathbf{I}_{N}\; \mathbf{0} \cdots \mathbf{0}]$, together with
$\|\widehat{\mathbf{D}}_{j}\|^{2} = \sum_{i=1}^{p}\|\widehat{\boldsymbol{\Theta}}_{j-i}\|^{2}$ and
$\|\cdot\| \leq \sqrt{pN}\|\cdot\|_{1}$.
\end{proof}

\begin{lemma}[Bootstrap innovations]\label{lem:bootmom}
Let Assumptions \ref{ass:stationary} and \ref{ass:iid-plus} hold.
Then
\begin{equation*}
    \mathbb{E}^{*}(\mathbf{u}^{*}_{t}) = \mathbf{0},
    \qquad
    \mathbb{E}^{*}\big(\mathbf{u}^{*}_{t}\mathbf{u}^{*\top}_{t}\big) = \widehat{\boldsymbol{\Sigma}}_{emp} \xrightarrow{p} \boldsymbol{\Sigma},
    \qquad
    \mu^{*}_{4} := \mathbb{E}^{*}\big\|\mathbf{u}^{*}_{t}\big\|^{4} = O_{p}(1).
\end{equation*}
\end{lemma}

\begin{proof}
Since $P^{*}$ is the empirical distribution of $\{\widehat{\mathbf{u}}_{t}\}_{t=1}^{T}$, the three statements read
\begin{equation*}
    \frac{1}{T}\sum_{t=1}^{T}\widehat{\mathbf{u}}_{t} = \mathbf{0},
    \qquad
    \frac{1}{T}\sum_{t=1}^{T}\widehat{\mathbf{u}}_{t}\widehat{\mathbf{u}}^{\top}_{t} \xrightarrow{p} \boldsymbol{\Sigma},
    \qquad
    \frac{1}{T}\sum_{t=1}^{T}\big\|\widehat{\mathbf{u}}_{t}\big\|^{4} = O_{p}(1),
\end{equation*}
the first holding by the demeaning convention.
By Lemma \ref{lem:decay} and Minkowski's inequality, $(\mathbb{E}\|\mathbf{y}_{t}\|^{4})^{1/4} \leq \sum_{j\geq0}\|\boldsymbol{\Theta}_{j}\|(\mathbb{E}\|\mathbf{u}_{t}\|^{4})^{1/4} < \infty$ under Assumption \ref{ass:iid-plus}, so $\mathbb{E}\|\mathbf{x}_{t}\|^{4} < \infty$ and $\{(\mathbf{u}_{t}, \mathbf{x}_{t})\}$ is stationary and ergodic.
Given that $\widehat{\mathbf{C}} \xrightarrow{p} \mathbf{C}$ by Proposition \ref{prop:an}, the second and third statements follow from \eqref{eqA:mds}.
\end{proof}

\begin{lemma}[Orthogonality]\label{lem:orth}
For all $\mathbf{F}_1 \in \mathbb{S}^{N_1}$ and $\mathbf{F}_2 \in \mathbb{S}^{N_2}$,
\begin{equation*}
    \operatorname{tr}\big[(\widehat{\boldsymbol{\Sigma}}_2^{-1} \otimes \mathbf{F}_1)\widehat{\mathbf{S}}_{T}\big] = 0,
    \qquad
    \operatorname{tr}\big[(\mathbf{F}_2 \otimes \widehat{\boldsymbol{\Sigma}}_1^{-1})\widehat{\mathbf{S}}_{T}\big] = 0 .
\end{equation*}
\end{lemma}

\begin{proof}
Fix $\mathbf{F}_1 \in \mathbb{S}^{N_1}$. Then
\begin{align*}
    \operatorname{tr}\big[(\widehat{\boldsymbol{\Sigma}}_2^{-1} \otimes \mathbf{F}_1)\widehat{\boldsymbol{\Sigma}}_{emp}\big]
    &\overset{\text{L\ref{lem:kron}.2}}{=}
    \frac{1}{T}\sum_{t=1}^{T}\operatorname{tr}\big(\mathbf{F}_1\widehat{\mathbf{U}}_t\widehat{\boldsymbol{\Sigma}}_2^{-1}\widehat{\mathbf{U}}_t^\top\big) \\
    &= N_2\operatorname{tr}\big(\mathbf{F}_1\widehat{\boldsymbol{\Sigma}}_1\big) \\
    &\overset{\text{L\ref{lem:kron}.3}}{=}
    \operatorname{tr}\big[(\widehat{\boldsymbol{\Sigma}}_2^{-1} \otimes \mathbf{F}_1)(\widehat{\boldsymbol{\Sigma}}_2 \otimes \widehat{\boldsymbol{\Sigma}}_1)\big],
\end{align*}
where the equality from the first to the second line follows from the third condition of \eqref{eqA:foc}.
Subtracting and using linearity of the trace gives the first claim.
The second is identical, fixing $\mathbf{F}_2 \in \mathbb{S}^{N_2}$ and using $\operatorname{tr}[(\mathbf{F}_2 \otimes \widehat{\boldsymbol{\Sigma}}_1^{-1})\widehat{\mathbf{u}}_t\widehat{\mathbf{u}}_t^\top] = \operatorname{tr}(\mathbf{F}_2\widehat{\mathbf{U}}_t^\top\widehat{\boldsymbol{\Sigma}}_1^{-1}\widehat{\mathbf{U}}_t)$.
\end{proof}

\begin{lemma}[Bootstrap predeterminedness]\label{lem:bootmds}
Let Assumptions \ref{ass:stationary} and \ref{ass:iid-plus} hold.
Then,
\begin{equation}\label{eqB:xmom}
    \mathbb{E}^{*}\big\|\mathbf{x}^{*}_{t}\big\|^{2} \leq K\big\|\widehat{\boldsymbol{\Sigma}}_{emp}\big\| = O_{p}(1),
\end{equation}
and, for every $\mathbf{M}$ with $\|\mathbf{M}\| = O_{p}(1)$,
\begin{equation*}
    \mathbb{E}^{*}\Big\|\frac{1}{T}\sum_{t=1}^{T}\mathbf{U}^{*}_{t}\mathbf{M}\mathbf{X}^{*\top}_{t}\Big\|^{2}
    \leq \frac{K \|\mathbf{M}\|^{2}\big\|\widehat{\boldsymbol{\Sigma}}_{emp}\big\|^{2}}{T},
    \qquad\text{so}\qquad
    \frac{1}{T}\sum_{t=1}^{T}\mathbf{U}^{*}_{t}\mathbf{M}\mathbf{X}^{*\top}_{t} = O_{p^{*}}\big(T^{-1/2}\big).
\end{equation*}
\end{lemma}

\begin{proof}
    We have
\begin{equation*}
    \mathbb{E}^{*}\big\|\mathbf{x}^{*}_{t}\big\|^{2}
    =
    \sum_{j=1}^{t-1}\operatorname{tr}\big(\widehat{\mathbf{D}}_{j}\widehat{\boldsymbol{\Sigma}}_{emp}\widehat{\mathbf{D}}^{\top}_{j}\big)
    \leq \big\|\widehat{\boldsymbol{\Sigma}}_{emp}\big\|\sum_{j\geq1}\big\|\widehat{\mathbf{D}}_{j}\big\|^{2}
    \overset{\text{L\ref{lem:decay}}}{\leq} K\big\|\widehat{\boldsymbol{\Sigma}}_{emp}\big\| ,
\end{equation*}
which is $O_{p}(1)$ by Lemma \ref{lem:bootmom}.
With $\|\mathbf{U}^{*}_{t}\| = \|\mathbf{u}^{*}_{t}\|$ and $\|\mathbf{X}^{*}_{t}\| = \|\mathbf{x}^{*}_{t}\|$ by \eqref{eq:Q_def},
\begin{align*}
    \mathbb{E}^{*}\Big\|\frac{1}{T}\sum_{t=1}^{T}\mathbf{U}^{*}_{t}\mathbf{M}\mathbf{X}^{*\top}_{t}\Big\|^{2} 
    &= \frac{1}{T^{2}}\sum_{t=1}^{T}\mathbb{E}^{*}\big\|\mathbf{U}^{*}_{t}\mathbf{M}\mathbf{X}^{*\top}_{t}\big\|^{2} \\
    &\leq \frac{\|\mathbf{M}\|^{2}}{T^{2}}\sum_{t=1}^{T}
    \operatorname{tr}\big(\widehat{\boldsymbol{\Sigma}}_{emp}\big)\,\mathbb{E}^{*}\big\|\mathbf{x}^{*}_{t}\big\|^{2} \\
    &\overset{\eqref{eqB:xmom}}{\leq}
    \frac{K \|\mathbf{M}\|^{2}\big\|\widehat{\boldsymbol{\Sigma}}_{emp}\big\|^{2}}{T}.
\end{align*}
\end{proof}

\begin{lemma}[Bootstrap LLN]\label{lem:bootlln}
Let Assumptions \ref{ass:stationary} and \ref{ass:iid-plus} hold.
Then
\begin{equation*}
    \frac{1}{T}\sum_{t=1}^{T}\mathbf{x}^{*}_{t}\mathbf{x}^{*\top}_{t}
    \xrightarrow{p^{*}}_{p}
    \boldsymbol{\Phi} := \mathbb{E}\big(\mathbf{x}_{t}\mathbf{x}_{t}^{\top}\big),
    \qquad\text{and}\qquad
    \mathbf{H}^{*} \xrightarrow{p^{*}}_{p} \mathbf{H} .
\end{equation*}
\end{lemma}

\begin{proof}
Since $\mathbf{y}^{*}_{-p+1} = \cdots = \mathbf{y}^{*}_{0} = \mathbf{0}$, iterating \eqref{eqB:dgp} gives the truncated MMA form $\mathbf{y}_{t}^{*} = \sum_{j=0}^{t-1}\widehat{\boldsymbol{\Theta}}_{j}\mathbf{u}_{t-j}^{*}$, and hence, stacking $p$ lags as in \eqref{eqB:Rdef},
\begin{equation}\label{eqB:xstar}
    \mathbf{x}_{t}^{*} = \sum_{j=1}^{t-1}\widehat{\mathbf{D}}_{j}\mathbf{u}_{t-j}^{*},
    \qquad t = 1, \ldots, T.
\end{equation}
Insert \eqref{eqB:xstar} into $\frac{1}{T}\sum_{t=1}^{T}\mathbf{x}^{*}_{t}\mathbf{x}^{*\top}_{t}$ and collect terms by the difference $h = j_{1} - j_{2}$ of the two lag indices to obtain
    \begin{align*}
	\frac{1}{T}\sum_{t=1}^{T}\mathbf{x}^{*}_{t}\mathbf{x}^{*\top}_{t} &= \frac{1}{T} \sum_{t = 1}^{T} \sum_{j_{1}=1}^{t-1} \sum_{j_{2}=1}^{t-1} \widehat{\mathbf{D}}_{j_{1}} \mathbf{u}_{t-j_{1}}^{*} \mathbf{u}_{t-j_{2}}^{*\top} \widehat{\mathbf{D}}_{j_{2}}^{\top} \\
									  &= \frac{1}{T} \sum_{t = 1}^{T} \sum_{h=1}^{t-2} \sum_{s=1}^{t-1-h} \widehat{\mathbf{D}}_{s+h} \mathbf{u}_{t-(s+h)}^{*} \mathbf{u}_{t-s}^{*\top} \widehat{\mathbf{D}}_{s}^{\top} \\
									  &+ \frac{1}{T} \sum_{t = 1}^{T} \sum_{s=1}^{t-1} \widehat{\mathbf{D}}_{s} \mathbf{u}_{t-s}^{*} \mathbf{u}_{t-s}^{*\top} \widehat{\mathbf{D}}_{s}^{\top} \\
									  &+ \frac{1}{T} \sum_{t = 1}^{T} \sum_{h=1}^{t-2} \sum_{s=1}^{t-1-h} \widehat{\mathbf{D}}_{s} \mathbf{u}_{t-s}^{*} \mathbf{u}_{t-(s+h)}^{*\top} \widehat{\mathbf{D}}_{s+h}^{\top} \\
									  &= \mathbf{M}_{1}^{*} + \mathbf{M}_{2}^{*} + \mathbf{M}_{3}^{*}.
    \end{align*}
    Exchanging the order of summation and writing
    \begin{align*}
        \mathbf{V}_{T}^{*}(s) = \frac{1}{T} \sum_{t=s+1}^{T} \mathbf{u}_{t-s}^{*} \mathbf{u}_{t-s}^{*\top}, \qquad \mathbf{R}_{T}^{*}(h,s) = \frac{1}{T} \sum_{t=s+h+1}^{T}\mathbf{u}_{t-s-h}^{*} \mathbf{u}_{t-s}^{*\top},
    \end{align*}
    then the three terms read
    \begin{align*}
	\mathbf{M}_{2}^{*} = \sum_{s=1}^{T-1} \widehat{\mathbf{D}}_{s} \mathbf{V}_{T}^{*}(s) \widehat{\mathbf{D}}_{s}^{\top}, \qquad \mathbf{M}_{1}^{*} = \sum_{h=1}^{T-2} \sum_{s=1}^{T-1-h} \widehat{\mathbf{D}}_{s+h} \mathbf{R}_{T}^{*}(h,s) \widehat{\mathbf{D}}_{s}^{\top}, \qquad \mathbf{M}_{3}^{*} = \mathbf{M}_{1}^{*\top}.
    \end{align*}
    In what follows, we show (a) $\mathbf{M}^{*}_{2} \xrightarrow{p^{*}}_{p} \boldsymbol{\Phi}$ and (b) $\mathbf{M}^{*}_{1} = \mathbf{M}^{*\top}_{3} \xrightarrow{p^{*}}_{p} \mathbf{0}$.
    For part (a), the $\mathbf{u}^{*}_{t}$ are i.i.d.\ under $P^{*}$, so $\mathbb{E}^{*}(\mathbf{V}^{*}_{T}(s)) = \tfrac{T-s}{T}\widehat{\boldsymbol{\Sigma}}_{emp}$ and, the summands being independent, for every $0 \leq s < T$
    \begin{equation}\label{eqB:vdev}
	\begin{aligned}
	    \mathbb{E}^{*}\big\|\mathbf{V}^{*}_{T}(s) - \mathbb{E}^{*}(\mathbf{V}^{*}_{T}(s))\big\|^{2}
	    &= \frac{1}{T^{2}} \mathbb{E}^{*} \left\| \sum_{t=s+1}^{T} \mathbf{u}_{t-s}^{*} \mathbf{u}_{t-s}^{*\top}
	    - (T-s) \widehat{\boldsymbol{\Sigma}}_{emp} \right\|^2 \\
	    &= \frac{1}{T^{2}} \mathbb{E}^{*} \left\| \sum_{r=1}^{T-s} \mathbf{u}_{r}^{*} \mathbf{u}_{r}^{*\top}
	    - \sum_{r=1}^{T-s} \widehat{\boldsymbol{\Sigma}}_{emp}\right\|^2 \\
	    &= \frac{1}{T^{2}}\sum_{r=1}^{T-s}
	    \mathbb{E}^{*}\big\|\mathbf{u}^{*}_{r}\mathbf{u}^{*\top}_{r} - \widehat{\boldsymbol{\Sigma}}_{emp}\big\|^{2} \\
	    &\leq \frac{\mu^{*}_{4}}{T}
	    = O_{p}\big(T^{-1}\big).
	\end{aligned}
    \end{equation}
    by Lemma \ref{lem:bootmom}. 
    Now write $\mathbf{V}_{T}^{*}(s)$ as
    \begin{align*}
	\mathbf{V}_{T}^{*}(s) &= \underbrace{\widehat{\boldsymbol{\Sigma}}_{emp} - \frac{s}{T} \widehat{\boldsymbol{\Sigma}}_{emp}}_{\mathbb{E}^{*}(\mathbf{V}_{T}^{*}(s))} + \left(\mathbf{V}_{T}^{*}(s) - \mathbb{E}^{*}(\mathbf{V}_{T}^{*}(s)\right),
    \end{align*}
    sandwich between $\widehat{\mathbf{D}}_{s}$ and $\widehat{\mathbf{D}}_{s}^{\top}$ and sum over $s$ to obtain
    \begin{align*}
	\mathbf{M}_{2}^{*} = \sum_{s=1}^{T-1} \widehat{\mathbf{D}}_{s} \widehat{\boldsymbol{\Sigma}}_{emp} \widehat{\mathbf{D}}_{s}^{\top} - \underbrace{\frac{1}{T}\sum_{s=1}^{T-1} s \widehat{\mathbf{D}}_{s} \widehat{\boldsymbol{\Sigma}}_{emp} \widehat{\mathbf{D}}_{s}^{\top}}_{(i)} + \underbrace{\sum_{s=1}^{T-1} \widehat{\mathbf{D}}_{s} \left(\mathbf{V}_{T}^{*}(s) - \mathbb{E}^{*}(\mathbf{V}_{T}^{*}(s)\right) \widehat{\mathbf{D}}_{s}^{\top}}_{(ii)}.
    \end{align*}
    For $(i)$, we may take the norm and obtain
    \begin{align*}
	\left\| \frac{1}{T}\sum_{s=1}^{T-1} s \widehat{\mathbf{D}}_{s} \widehat{\boldsymbol{\Sigma}}_{emp} \widehat{\mathbf{D}}_{s}^{\top} \right\| \leq \frac{\| \widehat{\boldsymbol{\Sigma}}_{emp} \|}{T} \sum_{s=1}^{T-1} s \|\widehat{\mathbf{D}}_{s} \|^2 \leq \frac{K^{2} \| \widehat{\boldsymbol{\Sigma}}_{emp}\|}{T} \sum_{s \geq 1} s \bar{\rho}^{2s} = O_{p}(T^{-1}).
    \end{align*}
    For $(ii)$, we have
    \begin{align*}
	\left\|\sum_{s=1}^{T-1} \widehat{\mathbf{D}}_{s} \left(\mathbf{V}_{T}^{*}(s) - \mathbb{E}^{*}(\mathbf{V}_{T}^{*}(s)\right) \widehat{\mathbf{D}}_{s}^{\top} \right\| \leq \sum_{s=1}^{T-1} \|\widehat{\mathbf{D}}_{s} \|^2 \mathbb{E}^{*} \| \mathbf{V}_{T}^{*}(s) - \mathbb{E}^{*} (\mathbf{V}_{T}^{*}(s)) \| = O_{p}(T^{-1/2}).
    \end{align*}
    The conditional Markov inequality then yields that $(ii)$ is $o_{p^{*}}(1)$, hence
    \begin{equation*}
        \mathbf{M}^{*}_{2}
        = \sum_{s=1}^{T-1}\widehat{\mathbf{D}}_{s}\widehat{\boldsymbol{\Sigma}}_{emp}\widehat{\mathbf{D}}^{\top}_{s}
        + o_{p^{*}}(1).
    \end{equation*}
    For the leading term, $\mathbf{x}_{t} = \sum_{j\geq1}\mathbf{D}_{j}\mathbf{u}_{t-j}$ gives $\boldsymbol{\Phi} = \sum_{j\geq1}\mathbf{D}_{j}\boldsymbol{\Sigma}\mathbf{D}^{\top}_{j}$, so for fixed $m$,
    \begin{small}
    \begin{align*}
        \Big\|\sum_{s=1}^{T-1}\widehat{\mathbf{D}}_{s}\widehat{\boldsymbol{\Sigma}}_{emp}\widehat{\mathbf{D}}^{\top}_{s} - \boldsymbol{\Phi}\Big\|
        &\leq \sum_{s=1}^{m}\big\|\widehat{\mathbf{D}}_{s}\widehat{\boldsymbol{\Sigma}}_{emp}\widehat{\mathbf{D}}^{\top}_{s} - \mathbf{D}_{s}\boldsymbol{\Sigma}\mathbf{D}^{\top}_{s}\big\|
        + \sum_{s=m+1}^{T-1}\big\|\widehat{\mathbf{D}}_{s}\big\|^{2}\big\|\widehat{\boldsymbol{\Sigma}}_{emp}\big\|
        + \sum_{s>m}\big\|\mathbf{D}_{s}\big\|^{2}\big\|\boldsymbol{\Sigma}\big\| \\
        &\overset{\text{L\ref{lem:decay}}}{\leq}
        \sum_{s=1}^{m}\big\|\widehat{\mathbf{D}}_{s}\widehat{\boldsymbol{\Sigma}}_{emp}\widehat{\mathbf{D}}^{\top}_{s} - \mathbf{D}_{s}\boldsymbol{\Sigma}\mathbf{D}^{\top}_{s}\big\|
        + \frac{K^{2}\bar{\rho}^{\,2m}}{1-\bar{\rho}^{2}}\Big(\big\|\widehat{\boldsymbol{\Sigma}}_{emp}\big\| + \big\|\boldsymbol{\Sigma}\big\|\Big),
    \end{align*}
    \end{small}
    Each $\widehat{\boldsymbol{\Theta}}_{j}$ is a polynomial in $\widehat{\mathbf{C}}$ through \eqref{eq:mma}, so the first term is $o_{p}(1)$ for fixed $m$ by Proposition \ref{prop:an} and Lemma \ref{lem:bootmom}, while the second is $O_{p}(1)\bar{\rho}^{\,2m}$ uniformly in $T$ and vanishes as $m \to \infty$, which completes the proof of (a).

    For part (b), $h \geq 1$ and $\mathbb{E}^{*}(\mathbf{u}_{t}^{*}) = \mathbf{0}$, so $\mathbb{E}^{*}(\mathbf{R}^{*}_{T}(h,s)) = \mathbf{0}$ and independence leaves only the diagonal terms of $\mathbb{E}^{*}\|\mathbf{R}^{*}_{T}(h,s)\|^{2}$, where $\sup_{h,s}\mathbb{E}^{*}\|\mathbf{R}^{*}_{T}(h,s)\| \leq (\mu^{*}_{4}/T)^{1/2}$. Therefore
    \begin{equation*}
        \mathbb{E}^{*}\big\|\mathbf{M}^{*}_{1}\big\|
        \leq \sum_{h\geq1}\sum_{s\geq1}\big\|\widehat{\mathbf{D}}_{s+h}\big\|\big\|\widehat{\mathbf{D}}_{s}\big\|
        \sup_{h,s}\mathbb{E}^{*}\big\|\mathbf{R}^{*}_{T}(h,s)\big\|
        \overset{\text{L\ref{lem:decay}}}{\leq}
        c\,\big(\mu^{*}_{4}/T\big)^{1/2} = O_{p}\big(T^{-1/2}\big),
    \end{equation*}
    where $c$ is the constant from Lemma \ref{lem:decay} and $\mathbf{M}^{*}_{3} = \mathbf{M}^{*\top}_{1}$, which completes the proof of (b) and, with (a), of the first claim.

    For the second claim, $\mathbf{Y}^{*}_{t} = \widehat{\mathbf{A}}\mathbf{X}^{*}_{t}\widehat{\mathbf{B}}^{\top} + \mathbf{U}^{*}_{t}$ holds exactly, so every term of the expansion of \eqref{eqA:foc} is a sample average of a bilinear form in $(\mathbf{x}^{*}_{t}, \mathbf{u}^{*}_{t})$ of which there are three kinds,
    \begin{equation}\label{eqB:Hterms}
        \underbrace{\frac{1}{T}\sum_{t=1}^{T}\mathbf{W}^{*}_{t}\widehat{\boldsymbol{\Sigma}}^{-1}\mathbf{W}^{*\top}_{t}}_{\text{linear in }T^{-1}\sum_{t=1}^{T}\mathbf{x}^{*}_{t}\mathbf{x}^{*\top}_{t}},
        \qquad
        \underbrace{\frac{1}{T}\sum_{t=1}^{T}\mathbf{U}^{*}_{t}\mathbf{M}\mathbf{U}^{*\top}_{t}}_{\text{linear in }\mathbf{V}^{*}_{T}(0)},
        \qquad
        \underbrace{\frac{1}{T}\sum_{t=1}^{T}\mathbf{P}^{*}_{t}\mathbf{u}^{*}_{t}}_{\operatorname{vec}(\mathbf{N}^{*}_{T})},
    \end{equation}
The first is the $\boldsymbol{\beta}$ block of $\mathbf{H}^{*}$, the second the $\boldsymbol{\sigma}$ block, and the third the two off-diagonal blocks, which collect the terms in which $\mathbf{u}^{*}_{t}$ enters linearly and are the bootstrap counterparts of the averages that \eqref{eqA:mds} makes $O_{p}(T^{-1/2})$, leaving $\mathbf{H}$ block diagonal.
Because the bootstrap residual $\mathbf{Y}^{*}_{t} - \widehat{\mathbf{A}}^{*}\mathbf{X}^{*}_{t}\widehat{\mathbf{B}}^{*\top}$ equals $\mathbf{U}^{*}_{t}$ minus a term linear in $\mathbf{X}^{*}_{t}$, and each condition in \eqref{eqA:foc} multiplies it by $\mathbf{X}^{*\top}_{t}$ or by its own transpose, the coefficient $\mathbf{P}^{*}_{t}$ on $\mathbf{u}^{*}_{t}$ contains $\mathbf{X}^{*}_{t}$ exactly once and is therefore linear in $\mathbf{x}^{*}_{t}$, with coefficients continuous in $\widehat{\boldsymbol{\zeta}}$.
Each block of $\mathbf{H}^{*}$ is a linear function, with coefficients continuous in $\widehat{\boldsymbol{\zeta}}$, of one of the three averages in \eqref{eqB:Hterms}.
Part (a) and $\widehat{\boldsymbol{\zeta}} \xrightarrow{p} \boldsymbol{\zeta}$ give the first, and it remains to show (c) $\mathbf{V}^{*}_{T}(0) \xrightarrow{p^{*}}_{p} \boldsymbol{\Sigma}$, which by \eqref{eqB:UMUlim} below delivers the second, and (d) $\mathbf{N}^{*}_{T} = o_{p^{*}}(1)$, which kills the third.

    For part (c), $s = 0$ in \eqref{eqB:vdev} gives $\mathbb{E}^{*}(\mathbf{V}^{*}_{T}(0)) = \widehat{\boldsymbol{\Sigma}}_{emp}$ and
    \begin{equation*}
        \mathbb{E}^{*}\big\|\mathbf{V}^{*}_{T}(0) - \widehat{\boldsymbol{\Sigma}}_{emp}\big\|^{2}
        \leq \frac{\mu^{*}_{4}}{T} = O_{p}\big(T^{-1}\big),
    \end{equation*}
    so $\widehat{\boldsymbol{\Sigma}}_{emp} \xrightarrow{p} \boldsymbol{\Sigma}$ of Lemma \ref{lem:bootmom} gives $\mathbf{V}^{*}_{T}(0) \xrightarrow{p^{*}}_{p} \boldsymbol{\Sigma}$.
    Consequently, for every $\mathbf{F}_{1} \in \mathbb{S}^{N_{1}}$,
    \begin{equation}\label{eqB:UMUlim}
    \begin{aligned}
        \operatorname{tr}\Big[\mathbf{F}_{1}\frac{1}{T}\sum_{t=1}^{T}
            \mathbf{U}^{*}_{t}\mathbf{M}\mathbf{U}^{*\top}_{t}\Big]
        &\overset{\text{L\ref{lem:kron}.2}}{=}
            \operatorname{tr}\big[(\mathbf{M} \otimes \mathbf{F}_{1})
            \mathbf{V}^{*}_{T}(0)\big]
            \xrightarrow{p^{*}}_{p}
            \operatorname{tr}\big[(\mathbf{M} \otimes \mathbf{F}_{1})
            (\boldsymbol{\Sigma}_{2} \otimes \boldsymbol{\Sigma}_{1})\big] \\
        &=
            \operatorname{tr}\big(\mathbf{M}\boldsymbol{\Sigma}_{2}\big)
            \operatorname{tr}\big(\mathbf{F}_{1}\boldsymbol{\Sigma}_{1}\big),
    \end{aligned}
    \end{equation}
    which is \eqref{eqA:UMU}, and the mirror image is identical.
    Separability is available only in the limit.
    Replacing the average in \eqref{eqB:UMUlim} by $\mathbb{E}^{*}(\mathbf{U}^{*}_{t}\mathbf{M}\mathbf{U}^{*\top}_{t})$ through \eqref{eqB:UMUstar} and subtracting the same two steps at $(\widehat{\boldsymbol{\Sigma}}_{2}, \widehat{\boldsymbol{\Sigma}}_{1})$,
    \begin{equation*}
        \operatorname{tr}\Big[\mathbf{F}_{1}\Big(
        \mathbb{E}^{*}\big(\mathbf{U}^{*}_{t}\mathbf{M}\mathbf{U}^{*\top}_{t}\big)
        - \operatorname{tr}\big(\mathbf{M}\widehat{\boldsymbol{\Sigma}}_{2}\big)\widehat{\boldsymbol{\Sigma}}_{1}\Big)\Big]
        \overset{\eqref{eqB:Semp}}{=} \operatorname{tr}\big[(\mathbf{M} \otimes \mathbf{F}_{1})\widehat{\mathbf{S}}_{T}\big],
        \qquad \mathbf{F}_{1} \in \mathbb{S}^{N_1},
    \end{equation*}
which vanishes at $\mathbf{M} = \widehat{\boldsymbol{\Sigma}}^{-1}_{2}$ by Lemma \ref{lem:orth} but not at $\mathbf{M} = \widehat{\boldsymbol{\Sigma}}^{-1}_{2}\Delta^{*}_{2}\widehat{\boldsymbol{\Sigma}}^{-1}_{2}$, the matrix that multiplies the bootstrap deviation in the expansion of the covariance conditions, where it is $\widehat{\mathbf{S}}_{T} \xrightarrow{p} \mathbf{0}$ that restores \eqref{eqA:UMU}.

    For part (d), $\mathbf{P}^{*}_{t}$ is $\mathcal{F}^{*}_{t-1}$-measurable by \eqref{eqB:xstar} while $\mathbf{u}^{*}_{t}$ is centered and independent of $\mathcal{F}^{*}_{t-1}$, so the summands of $\mathbf{N}^{*}_{T}$ are centered and uncorrelated under $P^{*}$ and
    \begin{equation*}
        \mathbb{E}^{*}\big\|\mathbf{N}^{*}_{T}\big\|^{2}
        = \frac{1}{T^{2}}\sum_{t=1}^{T}\mathbb{E}^{*}\operatorname{tr}\big(\mathbf{P}^{*\top}_{t}\mathbf{P}^{*}_{t}\widehat{\boldsymbol{\Sigma}}_{emp}\big)
        \leq \frac{c\big\|\widehat{\boldsymbol{\Sigma}}_{emp}\big\|}{T^{2}}
        \sum_{t=1}^{T}\mathbb{E}^{*}\big\|\mathbf{x}^{*}_{t}\big\|^{2}
        = O_{p}\big(T^{-1}\big),
    \end{equation*}
    the last equality from Lemma \ref{lem:bootmds}.
    Collecting, $\widehat{\boldsymbol{\zeta}} \xrightarrow{p} \boldsymbol{\zeta}$ with part (a) and parts (c) and (d) give
    \begin{equation}\label{eqB:Hlim}
        \frac{1}{T}\sum_{t=1}^{T}\mathbf{W}^{*}_{t}\widehat{\boldsymbol{\Sigma}}^{-1}\mathbf{W}^{*\top}_{t}
        \xrightarrow{p^{*}}_{p}
        \mathbb{E}\big(\mathbf{W}_{t}\boldsymbol{\Sigma}^{-1}\mathbf{W}^{\top}_{t}\big) = \boldsymbol{\Xi},
        \qquad
        \frac{1}{T}\sum_{t=1}^{T}\mathbf{U}^{*}_{t}\mathbf{M}\mathbf{U}^{*\top}_{t}
        \xrightarrow{p^{*}}_{p}
        \operatorname{tr}\big(\mathbf{M}\boldsymbol{\Sigma}_{2}\big)\boldsymbol{\Sigma}_{1},
    \end{equation}
    so the three kinds of \eqref{eqB:Hterms} converge to $\boldsymbol{\Xi}$, to $\mathbf{H}_{0,\sigma}$ of \eqref{eqA:covraw} and to $\mathbf{0}$.
    The limits being constant and with $\boldsymbol{\Gamma}^{*}\boldsymbol{\Gamma}^{*\top} \xrightarrow{p} \boldsymbol{\Gamma}\boldsymbol{\Gamma}^{\top}$ the continuous mapping theorem gives $\mathbf{H}^{*} \xrightarrow{p^{*}}_{p} \mathbf{H}_{0} + \boldsymbol{\Gamma}\boldsymbol{\Gamma}^{\top} = \mathbf{H}$.
\end{proof}

\begin{lemma}[Projection Inequality]\label{lem:proj}
For $\mathbf{M} = [\mathbf{M}_1 \cdots \mathbf{M}_p]$ and $\mathbf{M}_0 = [\mathbf{M}_{0,1} \cdots \mathbf{M}_{0,p}]$ with each $\mathbf{M}_{0,j}$ a Kronecker product of any two real valued matrices, $\|\mathcal{P}(\mathbf{M}) - \mathbf{M}_0\| \leq 2\|\mathbf{M} - \mathbf{M}_0\|$.
\end{lemma}

\begin{proof}
Since $\mathbf{M}_{0,j}$ is feasible for the minimization defining $\mathcal{P}_j$,
\begin{equation}\label{eq:feas}
    \big\|\mathcal{P}_j(\mathbf{M}_j) - \mathbf{M}_j\big\|
    = \min_{\mathbf{B}_{j} \otimes \mathbf{A}_{j}}\big\|\mathbf{M}_j - \mathbf{B}_{j} \otimes \mathbf{A}_{j}\big\|
    \leq \big\|\mathbf{M}_j - \mathbf{M}_{0,j}\big\|.
\end{equation}
Then
\begin{align*}
    \big\|\mathcal{P}(\mathbf{M}) - \mathbf{M}_0\big\|^2
    &= \sum_{j=1}^{p}\big\|\mathcal{P}_j(\mathbf{M}_j) - \mathbf{M}_{0,j}\big\|^2 \\
    &\leq \sum_{j=1}^{p}\Big(\big\|\mathcal{P}_j(\mathbf{M}_j) - \mathbf{M}_j\big\|
      + \big\|\mathbf{M}_j - \mathbf{M}_{0,j}\big\|\Big)^2 \\
    &\overset{\eqref{eq:feas}}{\leq} \sum_{j=1}^{p}\Big(2\big\|\mathbf{M}_j - \mathbf{M}_{0,j}\big\|\Big)^2
    = 4\big\|\mathbf{M} - \mathbf{M}_0\big\|^2 .
\end{align*}
\end{proof}

\begin{lemma}[Bias term]\label{lem:bias}
Suppose that $\{\sqrt{T}(\widehat{\mathbf{C}}^{*} - \widehat{\mathbf{C}})\}_{T}$ is uniformly integrable under $P^{*}$, in probability.
Then,
\begin{equation*}
    \big\|\widehat{\boldsymbol{\Psi}}\big\| = o_{p}\big(T^{-1/2}\big),
    \qquad
    \max\Big\{\big\|\widetilde{\mathbf{C}} - \widehat{\mathbf{C}}\big\| 
    \big\|\widetilde{\mathbf{C}}^{*} - \widehat{\mathbf{C}}^{*}\big\|\Big\}
    \leq 2\big\|\widehat{\boldsymbol{\Psi}}\big\| .
\end{equation*}
\end{lemma}

\begin{proof}
To prove $\big\|\widehat{\boldsymbol{\Psi}}\big\| = o_{p}\big(T^{-1/2}\big)$, $\widehat{\boldsymbol{\Psi}} = \mathbb{E}^{*}(\widehat{\mathbf{C}}^{*}) - \widehat{\mathbf{C}}$, where $\widehat{\mathbf{C}}^{*}$ is estimated on first-stage bootstrap samples generated from $\widehat{\mathbf{C}}$.
Proposition \ref{prop:boot} gives $\sqrt{T}(\widehat{\boldsymbol{\beta}}^{*} - \widehat{\boldsymbol{\beta}}) \xrightarrow{d^{*}} N(\mathbf{0}, \boldsymbol{\Omega}_{\beta})$ in probability, and since
\begin{equation*}
    \widehat{\mathbf{C}}^{*}_{j} - \widehat{\mathbf{C}}_{j}
    = (\widehat{\mathbf{B}}^{*}_{j} - \widehat{\mathbf{B}}_{j}) \otimes \widehat{\mathbf{A}}_{j}
    + \widehat{\mathbf{B}}_{j} \otimes (\widehat{\mathbf{A}}^{*}_{j} - \widehat{\mathbf{A}}_{j})
    + (\widehat{\mathbf{B}}^{*}_{j} - \widehat{\mathbf{B}}_{j}) \otimes (\widehat{\mathbf{A}}^{*}_{j} - \widehat{\mathbf{A}}_{j}),
\end{equation*}
with the last term $O_{p^{*}}(T^{-1})$, the mapping that leads to \eqref{eq:an_C_app}, with $\mathbf{J}_{\beta}$ evaluated at $\widehat{\boldsymbol{\beta}} \xrightarrow{p} \boldsymbol{\beta}$, and the conditional Slutsky lemma give
\begin{equation*}
    \sqrt{T}\operatorname{vec}\big(\widehat{\mathbf{C}}^{*} - \widehat{\mathbf{C}}\big)
    \xrightarrow{d^{*}} N\big(\mathbf{0}, \boldsymbol{\Omega}_C\big)
    \qquad \text{in probability}.
\end{equation*}
So by uniform integrability and Theorem 25.12 in \cite{billingsley_probability_1995},
\begin{equation*}
    \sqrt{T}\,\widehat{\boldsymbol{\Psi}}
    =
    \mathbb{E}^{*}\big[\sqrt{T}\big(\widehat{\mathbf{C}}^{*} - \widehat{\mathbf{C}}\big)\big]
    \xrightarrow{p} \mathbf{0} .
\end{equation*}
To prove $\max\Big\{\big\|\widetilde{\mathbf{C}} - \widehat{\mathbf{C}}\big\| \big\|\widetilde{\mathbf{C}}^{*} - \widehat{\mathbf{C}}^{*}\big\|\Big\} \leq 2\big\|\widehat{\boldsymbol{\Psi}}\big\|$, $\widehat{\mathbf{C}}_j = \widehat{\mathbf{B}}_j \otimes \widehat{\mathbf{A}}_j$ for each $j$, so Lemma \ref{lem:proj} applies with
$\mathbf{M} = \widehat{\mathbf{C}} - \delta\widehat{\boldsymbol{\Psi}}$ and
$\mathbf{M}_0 = \widehat{\mathbf{C}}$:
\begin{equation*}
    \big\|\widetilde{\mathbf{C}} - \widehat{\mathbf{C}}\big\|
    =
    \big\|\mathcal{P}\big(\widehat{\mathbf{C}} - \delta\widehat{\boldsymbol{\Psi}}\big) - \widehat{\mathbf{C}}\big\|
    \overset{\text{L\ref{lem:proj}}}{\leq}
    2\big\|\big(\widehat{\mathbf{C}} - \delta\widehat{\boldsymbol{\Psi}}\big) - \widehat{\mathbf{C}}\big\|
    = 2\delta\big\|\widehat{\boldsymbol{\Psi}}\big\|
    \leq 2\big\|\widehat{\boldsymbol{\Psi}}\big\| ,
\end{equation*}
using $\delta \le 1$ and $\|\widetilde{\mathbf{C}} - \widehat{\mathbf{C}}\| = 0$ when
$\widetilde{\mathbf{C}} = \widehat{\mathbf{C}}$. Since
$\widehat{\mathbf{C}}^{*}_j = \widehat{\mathbf{B}}^{*}_j \otimes \widehat{\mathbf{A}}^{*}_j$ by
construction, the same display with $(\widehat{\mathbf{C}}^{*}, \delta^{*})$ in place of
$(\widehat{\mathbf{C}}, \delta)$ bounds $\|\widetilde{\mathbf{C}}^{*} - \widehat{\mathbf{C}}^{*}\|$.
\end{proof}

\subsubsection{Asymptotic validity of ProBAB-MAR}
\label{app:validity}

Fix $i, k$ and $h \geq 0$, write $g = g_{ik,h}$, and let
\begin{equation*}
    J^{*}_{T}(x)
    = P^{*}\Big(\sqrt{T}\big(g(\widetilde{\boldsymbol{\zeta}}^{*}) - g(\widehat{\boldsymbol{\zeta}})\big) \leq x\Big),
    \qquad
    J_{T}(x)
    = P\Big(\sqrt{T}\big(g(\widehat{\boldsymbol{\zeta}}) - g(\boldsymbol{\zeta})\big) \leq x\Big),
\end{equation*}
with $F$ the distribution function of $N(0, \sigma^{2}_{ik,h})$.
Adding and subtracting $F$ inside \eqref{eq:thm_kolmogorov},
\begin{equation*}
    \sup_{x \in \mathbb{R}}\big|J^{*}_{T}(x) - J_{T}(x)\big|
    \leq
    \sup_{x \in \mathbb{R}}\big|J_{T}(x) - F(x)\big|
    + \sup_{x \in \mathbb{R}}\big|J^{*}_{T}(x) - F(x)\big| .
\end{equation*}
The first term is then $o(1)$ by \eqref{eq:irf_an} and P\'olya's theorem, and the second is $o_{p}(1)$ by P\'olya's theorem in probability once $J^{*}_{T}(x) \xrightarrow{p} F(x)$ at every fixed $x$, so that \eqref{eq:thm_kolmogorov} reduces to
\begin{equation}\label{eqB:target}
    \sqrt{T}\big(g(\widetilde{\boldsymbol{\zeta}}^{*}) - g(\widehat{\boldsymbol{\zeta}})\big)
    \xrightarrow{d^{*}} N\big(0, \sigma^{2}_{ik,h}\big)
    \qquad\text{in probability.}
\end{equation}

Lemma \ref{lem:bias} gives $\max\{\|\widetilde{\mathbf{C}} - \widehat{\mathbf{C}}\|, \|\widetilde{\mathbf{C}}^{*} - \widehat{\mathbf{C}}^{*}\|\} \leq 2\|\widehat{\boldsymbol{\Psi}}\| = o_{p}(T^{-1/2})$, so \eqref{eqB:dgp} may be read with $\widetilde{\mathbf{C}}$ in place of $\widehat{\mathbf{C}}$, which enters below only through its consistency and Lemma \ref{lem:decay}, and $\widetilde{\boldsymbol{\zeta}}^{*}$ may be replaced by $\widehat{\boldsymbol{\zeta}}^{*}$ in \eqref{eqB:target}.
Since $g$ is continuously differentiable in a neighborhood of $\boldsymbol{\zeta}$ and $\widehat{\boldsymbol{\zeta}} \xrightarrow{p} \boldsymbol{\zeta}$ by Proposition \ref{prop:an}, the conditional delta method takes \eqref{eq:boot_zeta} to \eqref{eqB:target} with $\sigma^{2}_{ik,h} = \nabla g(\boldsymbol{\zeta})^{\top}\boldsymbol{\Omega}\nabla g(\boldsymbol{\zeta})$, and it suffices to prove Proposition \ref{prop:boot}.

\begin{proof}[Proof of Proposition \ref{prop:boot}]
By construction $\mathbf{Y}^{*}_{t} = \widehat{\mathbf{A}}\mathbf{X}^{*}_{t}\widehat{\mathbf{B}}^{\top} + \mathbf{U}^{*}_{t}$, so $\widehat{\boldsymbol{\zeta}}$ plays in the bootstrap world the role of $\boldsymbol{\zeta}$ and the expansion of Appendix \ref{app:asymptotic} applies with two substitutions.
In place of \eqref{eqA:mds}, we may use Lemma \ref{lem:bootmds}.
Discarding products of deviations presupposes the bootstrap counterpart of \eqref{eqA:rate}, $\widehat{\boldsymbol{\zeta}}^{*} - \widehat{\boldsymbol{\zeta}} = O_{p^{*}}(T^{-1/2})$.
Repeating on $\mathbf{Y}^{*}_{t}$ the expansion of \eqref{eqA:foc} that produced \eqref{eqA:stacked} gives
\begin{equation}\label{eqB:stacked}
    \mathbf{H}^{*}\big(\widehat{\boldsymbol{\zeta}}^{*} - \widehat{\boldsymbol{\zeta}}\big)
    = \frac{1}{T}\sum_{t=1}^{T}\mathbf{s}^{*}_{t} + \mathbf{R}^{*}_{T},
\end{equation}
so that, by the nonsingularity of $\mathbf{H}$ and the conditional Slutsky lemma, \eqref{eq:boot_zeta} follows from
\begin{enumerate}
    \item[(a)] $\mathbb{E}^{*}\big(\mathbf{s}^{*}_{t}\mid\mathcal{F}^{*}_{t-1}\big) = \mathbb{E}^{*}\big(\mathbf{s}^{*}_{t}\big) = \mathbf{0}$;
    \item[(b)] $T^{-1}\sum_{t=1}^{T}\mathbf{x}^{*}_{t}\mathbf{x}^{*\top}_{t} \xrightarrow{p^{*}}_{p} \boldsymbol{\Phi}$ and $\mathbf{H}^{*} \xrightarrow{p^{*}}_{p} \mathbf{H}$;
    \item[(c)] $\sqrt{T}\,\mathbf{R}^{*}_{T} = o_{p^{*}}(1)$;
    \item[(d)] $T^{-1/2}\sum_{t=1}^{T}\mathbf{s}^{*}_{t} \xrightarrow{d^{*}} N(\mathbf{0}, \boldsymbol{\Lambda})$ in probability,
\end{enumerate}
since (a) makes $\{\mathbf{s}^{*}_{t}\}$ a martingale difference array under $P^{*}$, on which (d) rests, and (b), (c) and $\boldsymbol{\Omega} = \mathbf{H}^{-1}\boldsymbol{\Lambda}\mathbf{H}^{-1}$ complete \eqref{eq:boot_zeta}.
Condition (b) is Lemma \ref{lem:bootlln}, so it remains to verify (a), (c) and (d).

For (a), $\mathbf{W}^{*}_{t}$ is $\mathcal{F}^{*}_{t-1}$-measurable by \eqref{eqB:xstar} and $\mathbb{E}^{*}(\mathbf{u}^{*}_{t}) = \mathbf{0}$ by Lemma \ref{lem:bootmom}, so \eqref{eq:score_beta} gives $\mathbb{E}^{*}(\mathbf{s}^{*}_{\beta,t}\mid\mathcal{F}^{*}_{t-1}) = \mathbf{W}^{*}_{t}\widehat{\boldsymbol{\Sigma}}^{-1}\mathbb{E}^{*}(\mathbf{u}^{*}_{t}) = \mathbf{0}$.
The covariance block carries no such factor, but $\mathbf{s}^{*}_{\sigma,t}$ is a function of $\mathbf{u}^{*}_{t}$ alone, which is independent of $\mathcal{F}^{*}_{t-1}$ under $P^{*}$, so its conditional and unconditional expectations coincide.
For any conformable $\mathbf{M}$,
\begin{equation}\label{eqB:UMUstar}
    \mathbb{E}^{*}\big(\mathbf{U}^{*}_{t}\mathbf{M}\mathbf{U}^{*\top}_{t}\big)
    = \frac{1}{T}\sum_{t=1}^{T}\widehat{\mathbf{U}}_{t}\mathbf{M}\widehat{\mathbf{U}}^{\top}_{t},
\end{equation}
and identically for $\mathbf{U}^{*\top}_{t}\mathbf{M}\mathbf{U}^{*}_{t}$.
At $\mathbf{M} = \widehat{\boldsymbol{\Sigma}}_2^{-1}$ and $\mathbf{M} = \widehat{\boldsymbol{\Sigma}}_1^{-1}$ the right-hand sides are $N_2\widehat{\boldsymbol{\Sigma}}_1$ and $N_1\widehat{\boldsymbol{\Sigma}}_2$ by the third and fourth conditions of \eqref{eqA:foc}, which hold with equality because the normalization $\|\boldsymbol{\Sigma}_1\| = 1$ lies along a direction in which \eqref{eq:loglik} is flat, so that its multiplier vanishes.
Evaluating \eqref{eq:score_sigma} at $\widehat{\boldsymbol{\zeta}}$ therefore gives $\mathbb{E}^{*}(\mathbf{s}^{*}_{\sigma,t}) = \mathbf{0}$, and the two blocks together give (a).

For (c), since $\widehat{\boldsymbol{\Sigma}}_{i} \xrightarrow{p} \boldsymbol{\Sigma}_{i}$ by Proposition \ref{prop:an},
\begin{equation}\label{eqB:sigexp}
    \widehat{\boldsymbol{\Sigma}}^{*-1}_{i}
    = \widehat{\boldsymbol{\Sigma}}^{-1}_{i}
    - \widehat{\boldsymbol{\Sigma}}^{-1}_{i}\Delta^{*}_{i}\widehat{\boldsymbol{\Sigma}}^{-1}_{i}
    + O_{p^{*}}\big(T^{-1}\big),
    \qquad
    \frac{1}{T}\sum_{t=1}^{T}\big\|\mathbf{x}^{*}_{t}\big\|^{2}
    = \operatorname{tr}\Big(\frac{1}{T}\sum_{t=1}^{T}\mathbf{x}^{*}_{t}\mathbf{x}^{*\top}_{t}\Big)
    \xrightarrow{p^{*}}_{p} \operatorname{tr}(\boldsymbol{\Phi}),
\end{equation}
the second by Lemma \ref{lem:bootlln}.
The terms discarded from \eqref{eqA:foc} to \eqref{eqB:stacked} are finite in number and of two kinds.
Those carrying two deviations and no $\mathbf{u}^{*}_{t}$, namely the last four terms of the coefficient expansion and the $\boldsymbol{\Upsilon}^{*}_{t}\widehat{\boldsymbol{\Sigma}}^{-1}_{2}\boldsymbol{\Upsilon}^{*\top}_{t}$ term of the covariance expansion, follow, for $\Delta^{*}, \Delta^{*\prime} \in \{\Delta^{*}_{A}, \Delta^{*}_{B}\}$,
\begin{equation*}
\begin{split}
    \Big\|\frac{1}{T}\sum_{t=1}^{T}\mathbf{M}_{1}\Delta^{*}\mathbf{X}^{*}_{t}\Delta^{*\prime\top}\mathbf{M}_{2}\mathbf{X}^{*\top}_{t}\Big\|
    &\leq K\big\|\Delta^{*}\big\|\big\|\Delta^{*\prime}\big\|\,
    \frac{1}{T}\sum_{t=1}^{T}\big\|\mathbf{x}^{*}_{t}\big\|^{2} \\
    &= O_{p^{*}}\big(T^{-1/2}\big)\,O_{p^{*}}\big(T^{-1/2}\big)\,O_{p^{*}}(1)
    = O_{p^{*}}\big(T^{-1}\big),
\end{split}
\end{equation*}
the last step by the rate $\widehat{\boldsymbol{\zeta}}^{*} - \widehat{\boldsymbol{\zeta}} = O_{p^{*}}(T^{-1/2})$ and \eqref{eqB:sigexp}.
Those carrying one deviation and one $\mathbf{u}^{*}_{t}$, namely the $\Delta^{*}_{B}$ term on the right of the coefficient system and the two $\boldsymbol{\Upsilon}^{*}_{t}\widehat{\boldsymbol{\Sigma}}^{-1}_{2}\mathbf{U}^{*\top}_{t}$ terms, factor as
\begin{equation*}
    \Big\|\Delta^{*}\,\frac{1}{T}\sum_{t=1}^{T}\mathbf{U}^{*}_{t}\mathbf{M}\mathbf{X}^{*\top}_{t}\Big\|
    \leq \big\|\Delta^{*}\big\|\Big\|\frac{1}{T}\sum_{t=1}^{T}\mathbf{U}^{*}_{t}\mathbf{M}\mathbf{X}^{*\top}_{t}\Big\|
    \overset{\text{L\ref{lem:bootmds}}}{=} O_{p^{*}}\big(T^{-1/2}\big)O_{p^{*}}\big(T^{-1/2}\big)
    = O_{p^{*}}\big(T^{-1}\big),
\end{equation*}
which is the counterpart of the discards of Steps 1 and 2 of Appendix \ref{app:asymptotic}, and is what leaves \eqref{eqB:stacked} block diagonal.
Those generated by $\widehat{\boldsymbol{\Sigma}}^{*-1}_{j}$ are $O_{p^{*}}(T^{-1})$.
Summing the two kinds, $\mathbf{R}^{*}_{T} = O_{p^{*}}(T^{-1})$ and $\sqrt{T}\,\mathbf{R}^{*}_{T} = o_{p^{*}}(1)$, which is (c).

For (d), we show $\dfrac1T\sum_{t=1}^{T}\mathbb{E}^{*}\big(\mathbf{s}^{*}_{t}\mathbf{s}^{*\top}_{t}\big) \xrightarrow{p^{*}}_{p} \boldsymbol{\Lambda}$.
$\mathbf{W}^{*}_{t}$ is $\mathcal{F}^{*}_{t-1}$-measurable by \eqref{eqB:xstar} while $\mathbf{u}^{*}_{t}$ is independent of $\mathcal{F}^{*}_{t-1}$, so
\begin{equation*}
\begin{aligned}
    \frac{1}{T}\sum_{t=1}^{T}\mathbb{E}^{*}\big(\mathbf{s}^{*}_{\beta,t}\mathbf{s}^{*\top}_{\beta,t}\mid\mathcal{F}^{*}_{t-1}\big)
    &\overset{\text{L\ref{lem:bootmom}}}{=}
    \frac{1}{T}\sum_{t=1}^{T}\mathbf{W}^{*}_{t}\widehat{\boldsymbol{\Sigma}}^{-1}
    \widehat{\boldsymbol{\Sigma}}_{emp}\widehat{\boldsymbol{\Sigma}}^{-1}\mathbf{W}^{*\top}_{t} \\
    &\overset{\eqref{eqB:Semp}}{=}
    \frac{1}{T}\sum_{t=1}^{T}\mathbf{W}^{*}_{t}\widehat{\boldsymbol{\Sigma}}^{-1}\mathbf{W}^{*\top}_{t}
    + \frac{1}{T}\sum_{t=1}^{T}\mathbf{W}^{*}_{t}\widehat{\boldsymbol{\Sigma}}^{-1}
      \widehat{\mathbf{S}}_{T}\widehat{\boldsymbol{\Sigma}}^{-1}\mathbf{W}^{*\top}_{t},
\end{aligned}
\end{equation*}
of which the first term is $\boldsymbol{\Xi} = \boldsymbol{\Lambda}_{\beta}$ in the limit by \eqref{eqB:Hlim} and the second obeys
\begin{equation*}
    \Big\|\frac{1}{T}\sum_{t=1}^{T}\mathbf{W}^{*}_{t}\widehat{\boldsymbol{\Sigma}}^{-1}
    \widehat{\mathbf{S}}_{T}\widehat{\boldsymbol{\Sigma}}^{-1}\mathbf{W}^{*\top}_{t}\Big\|
    \leq K\big\|\widehat{\mathbf{S}}_{T}\big\|\frac{1}{T}\sum_{t=1}^{T}\big\|\mathbf{x}^{*}_{t}\big\|^{2}
    \overset{\eqref{eqB:sigexp}}{=} o_{p}(1)O_{p^{*}}(1),
\end{equation*}
because $\widehat{\mathbf{S}}_{T} \xrightarrow{p} \mathbf{0}$ by Lemma \ref{lem:bootmom} and Proposition \ref{prop:an}.
The covariance block carries no $\mathcal{F}^{*}_{t-1}$-measurable factor and is a function of $\mathbf{u}^{*}_{t}$ alone,
\begin{equation}\label{eqB:varsigma}
    \frac{1}{T}\sum_{t=1}^{T}\mathbb{E}^{*}\big(\mathbf{s}^{*}_{\sigma,t}\mathbf{s}^{*\top}_{\sigma,t}\mid\mathcal{F}^{*}_{t-1}\big)
    \overset{\text{\eqref{eqB:UMUstar}}}{=}
    \frac{1}{T}\sum_{t=1}^{T}\widehat{\mathbf{s}}_{\sigma,t}\widehat{\mathbf{s}}^{\top}_{\sigma,t}
    \xrightarrow{p}
    \mathbb{E}\big(\mathbf{s}_{\sigma,t}\mathbf{s}^{\top}_{\sigma,t}\big) = \boldsymbol{\Lambda}_{\sigma}.
\end{equation}
For the cross block, $\mathbf{W}^{*}_{t}$ is $\mathcal{F}^{*}_{t-1}$-measurable while $\mathbf{u}^{*}_{t}\mathbf{s}^{*\top}_{\sigma,t}$ is a function of $\mathbf{u}^{*}_{t}$ alone, so as in \eqref{eqB:varsigma}, $\mathbb{E}^{*}(\mathbf{s}^{*}_{\beta,t}\mathbf{s}^{*\top}_{\sigma,t}\mid\mathcal{F}^{*}_{t-1}) = \mathbf{W}^{*}_{t}\widehat{\boldsymbol{\Sigma}}^{-1}\widehat{\mathbf{m}}$ with $\widehat{\mathbf{m}} = T^{-1}\sum_{s}\widehat{\mathbf{u}}_{s}\widehat{\mathbf{s}}^{\top}_{\sigma,s}$ the same for every $t$.
Since $\mathbf{W}^{*}_{t}$ is linear in $\mathbf{x}^{*}_{t}$, averaging over $t$ gives
\begin{equation*}
    \Big\|\frac{1}{T}\sum_{t=1}^{T}\mathbb{E}^{*}\big(\mathbf{s}^{*}_{\beta,t}\mathbf{s}^{*\top}_{\sigma,t}\mid\mathcal{F}^{*}_{t-1}\big)\Big\|
    \leq K\Big\|\frac{1}{T}\sum_{t=1}^{T}\mathbf{x}^{*}_{t}\Big\|\,\big\|\widehat{\mathbf{m}}\big\|.
\end{equation*}
Collecting the coefficient of each draw in \eqref{eqB:xstar} yields $T^{-1}\sum_{r=1}^{T-1}\big(\sum_{j=1}^{T-r}\widehat{\mathbf{D}}_{j}\big)\mathbf{u}^{*}_{r}$, a sum of independent centered terms whose coefficients are partial sums of the $\widehat{\mathbf{D}}_{j}$.
The cross terms therefore vanish as in Lemma \ref{lem:bootmds}, and bounding each partial sum by $\sum_{j\geq1}\|\widehat{\mathbf{D}}_{j}\| \leq K$ through Lemma \ref{lem:decay} gives
\begin{equation*}
    \mathbb{E}^{*}\Big\|\frac{1}{T}\sum_{t=1}^{T}\mathbf{x}^{*}_{t}\Big\|^{2}
    \leq \frac{\big\|\widehat{\boldsymbol{\Sigma}}_{emp}\big\|}{T^{2}}\sum_{r=1}^{T-1}\Big\|\sum_{j=1}^{T-r}\widehat{\mathbf{D}}_{j}\Big\|^{2}
    \leq \frac{K^{2}\big\|\widehat{\boldsymbol{\Sigma}}_{emp}\big\|}{T}
    = O_{p}\big(T^{-1}\big).
\end{equation*}
Since $\mathbf{s}_{\sigma}$ is quadratic in $\mathbf{U}$, $\|\widehat{\mathbf{m}}\| \leq c\,T^{-1}\sum_{s}(\|\widehat{\mathbf{u}}_{s}\| + \|\widehat{\mathbf{u}}_{s}\|^{3}) = O_{p}(1)$ by Lemma \ref{lem:bootmom}, and the cross block is therefore $O_{p^{*}}(T^{-1/2})$.
\end{proof}

\section{Additional Simulation Results}
\label{app:additional_sims}

\begin{figure}[tbp]
  \centering
  \includegraphics[width=\linewidth]{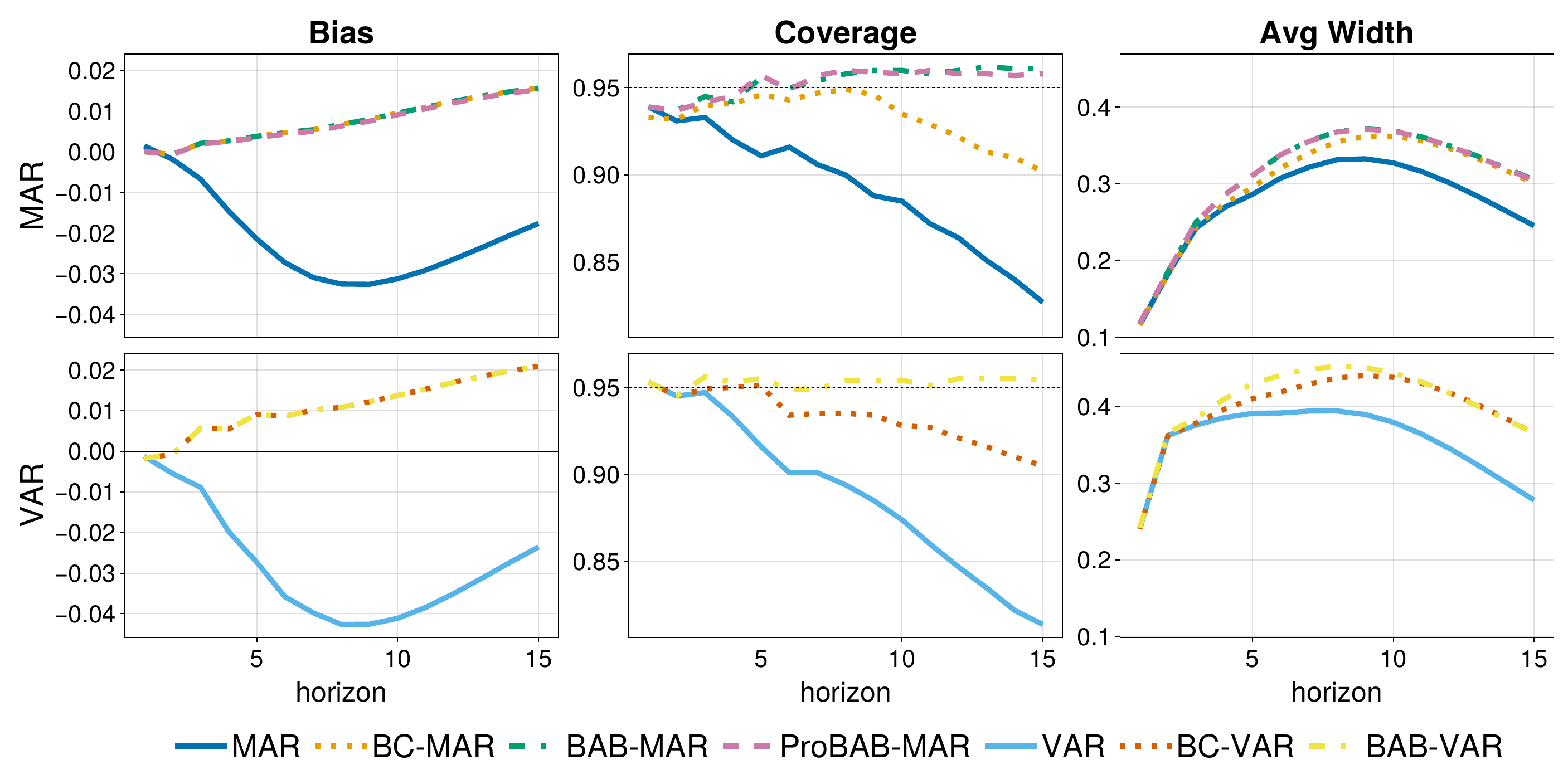}
  \caption{Bias, empirical coverage, and average width of nominal 95\% confidence intervals for impulse responses under the baseline MAR(1) DGP ($N_1 = 3$, $N_2 = 4$, $\rho_{\max} = 0.8$) with $T = 200$, across horizons $h = 1, \ldots, 15$.}
  \label{fig:baseline_obs200}
\end{figure}

\begin{figure}[tbp]
  \centering
  \includegraphics[width=\linewidth]{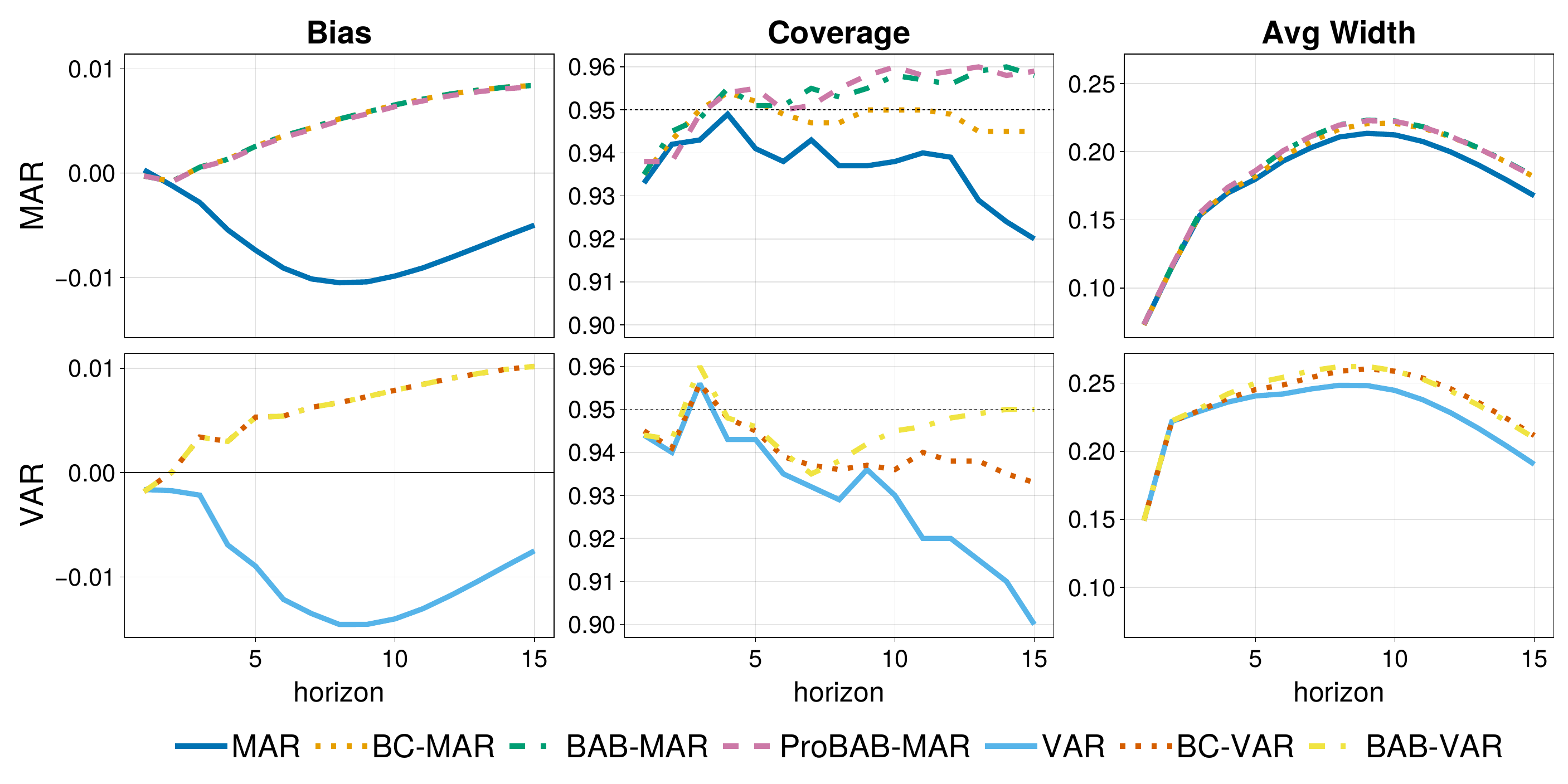}
  \caption{Bias, empirical coverage, and average width of nominal 95\% confidence intervals for impulse responses under the baseline MAR(1) DGP ($N_1 = 3$, $N_2 = 4$, $\rho_{\max} = 0.8$) with $T = 500$, across horizons $h = 1, \ldots, 15$.}
  \label{fig:baseline_obs500}
\end{figure}

\begin{figure}[tbp]
  \centering
  \includegraphics[width=\linewidth]{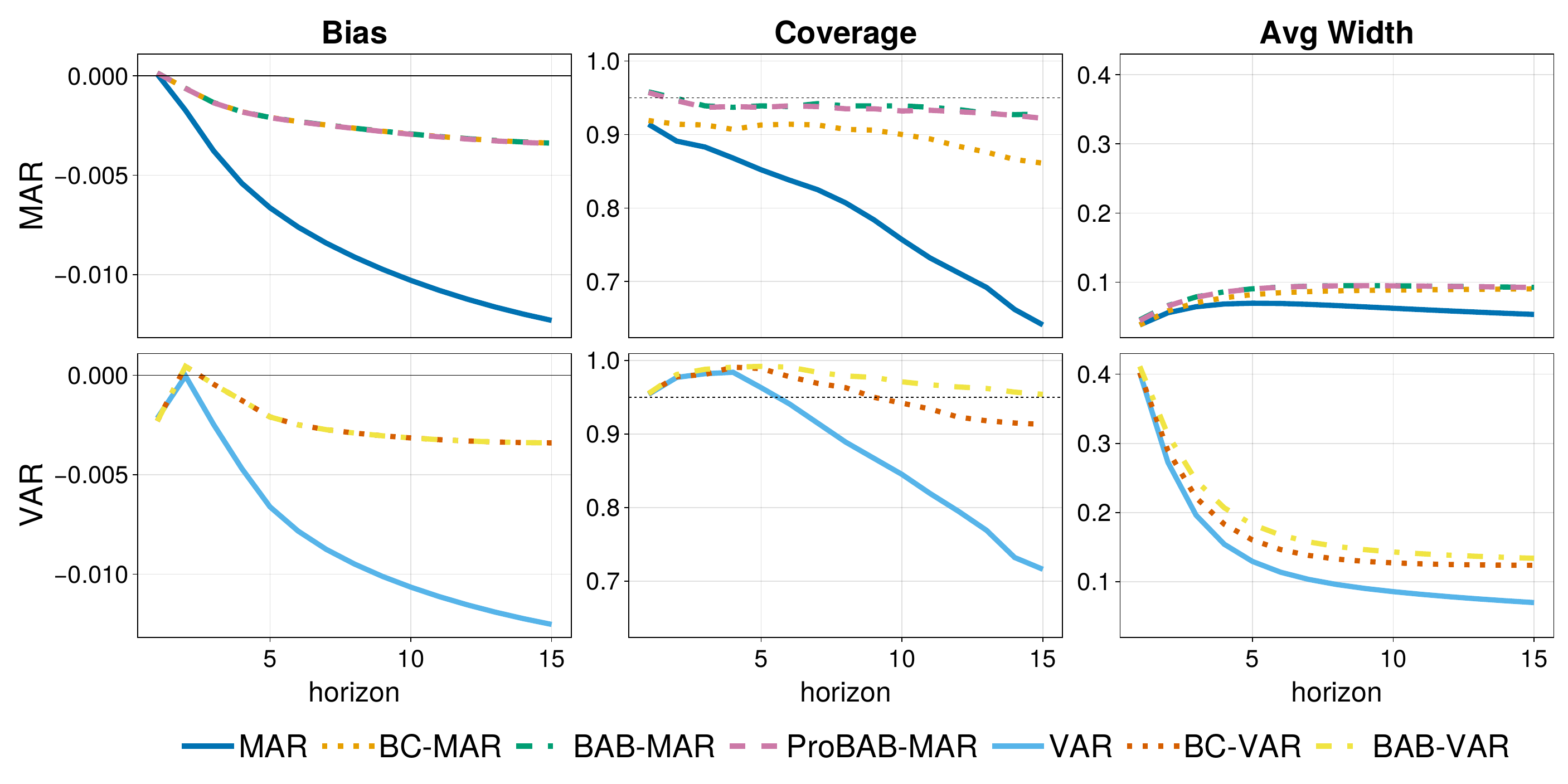}
  \caption{Bias, empirical coverage, and average width of nominal 95\%
  confidence intervals for impulse responses under the MAR(1) DGP
  ($N_1 = 3$, $N_2 = 4$, $T = 100$) with largest companion eigenvalue
  $\rho_{\max} = 0.99$, across horizons $h = 1, \ldots, 15$.}
  \label{fig:persistent_rho99}
\end{figure}

This appendix reports two sets of simulation results that supplement Section \ref{sec:simulations}.
The first extends the baseline design of Section \ref{subsec:sim_samplesize} to the larger samples $T = 200$ and $T = 500$, and the second reports the near-unit-root design of Section \ref{subsec:sim_persistence} at $\rho_{\max} = 0.99$.
All other settings are as in the main text.

Figures \ref{fig:baseline_obs200} and \ref{fig:baseline_obs500} report the results of the baseline design at $T = 200$ and $T = 500$.
The raw MAR estimates, on which the delta-method intervals are centered, become markedly less biased as the sample grows, with the bias at the worst horizon falling from $-0.06$ at $T = 100$ to $-0.03$ at $T = 200$ and $-0.01$ at $T = 500$.
The coverage of the MAR delta-method intervals improves accordingly, from $74\%$ at the worst horizon at $T = 100$ to $83\%$ at $T = 200$ and $92\%$ at $T = 500$.
Their average width also approaches that of the bootstrap intervals, so that the delta-method and bootstrap intervals of the MAR increasingly agree as $T$ grows, consistent with Corollary \ref{thm:bootstrap}, under which both share the limiting distribution \eqref{eq:irf_an}.
The bootstrap results are largely unchanged from those at $T = 100$.
The first-stage correction removes most of the bias, BAB-MAR and ProBAB-MAR remain at the nominal level at both sample sizes, and their intervals remain narrower than those of the unrestricted VAR.

Figure \ref{fig:persistent_rho99} reports the results of the near-unit-root design at $\rho_{\max} = 0.99$, where the bias grows for both the raw and the bias-corrected estimates.
The raw MAR estimates reach a bias of $-0.012$, and the first-stage correction reduces this only to $-0.004$.
The coverage of the MAR delta-method intervals falls to $65\%$, and recentering them at the bias-corrected estimate recovers only part of this gap, leaving coverage between $86\%$ and $90\%$.
BAB-MAR and ProBAB-MAR attain the nominal level at short horizons, but their coverage declines at longer horizons, to $89\%$ at worst, against $94\%$ for BAB-VAR.
The widths compare as at $\rho_{\max} = 0.95$ in Section \ref{subsec:sim_persistence}, with the MAR intervals markedly narrower than their unrestricted VAR counterparts.

\section{Naive Bootstrap Implementation}
\label{app:bab}

Omitting the projection yields BAB-MAR, stated as Algorithm \ref{alg:bab}, which differs from Algorithm \ref{alg:probab} only in lines \ref{lin:bab-delta} and \ref{lin:bab-refit}.
This procedure generates each bootstrap sample from a coefficient matrix that need not lie in the MAR parameter space, while every replication is re-estimated under the Kronecker restriction, so the bootstrap design and the estimator do not agree.
The projection removes this discrepancy at a cost that is asymptotically negligible, since it perturbs the bias-corrected coefficients by $O_p(T^{-1})$ and so leaves the first-order bootstrap distribution unchanged; Section \ref{sec:simulations} confirms that the two deliver similar coverage in finite samples.
What the projection buys, beyond the internal coherence of the procedure, is interpretability: because $\widetilde{\mathbf{C}}$ and each $\widetilde{\mathbf{C}}^{*}$ retain the form $\widetilde{\mathbf{B}}_j \otimes \widetilde{\mathbf{A}}_j$, the bootstrap impulse responses can be read through the Kronecker structure of Section \ref{subsec:irf}, decomposed into separate row and column components, examined through the rank of the matricized response, or reported as bias-corrected factors $\widetilde{\mathbf{A}}_j$ and $\widetilde{\mathbf{B}}_j$ individually.

\begin{algorithm}[t]
\caption{BAB-MAR}\label{alg:bab}
\small
\begin{algorithmic}[1]
\Require Data $\{\mathbf{Y}_t\}_{t=1}^{T}$, lag order $p$, replications $B_1, B_2$, nominal level $\alpha$
\State Fit MAR to obtain $\widehat{\mathbf{C}}_j = \widehat{\mathbf{B}}_j \otimes \widehat{\mathbf{A}}_j$, $\widehat{\mathbf{u}}_t = \mathbf{y}_t - \sum_j \widehat{\mathbf{C}}_j \mathbf{y}_{t-j}$, and center residuals $\widehat{\mathbf{u}}^{c}_t = \widehat{\mathbf{u}}_t - \bar{\mathbf{u}}$
\Statex \textit{Stage 1: bias-corrected DGP}
\For{$m = 1, \ldots, B_1$}
  \State Draw $\mathbf{u}^{*}_{1,t}$ from $\widehat{\mathbf{u}}^{c}_t$, simulate $\mathbf{y}_{t}^{*(m)}$ from $\widehat{\mathbf{C}}$ and $\mathbf{u}^{*}_{1,t}$, and re-fit the MAR to obtain $\widehat{\mathbf{C}}^{*(m)}$
\EndFor
\State Estimate the bias by $\widehat{\boldsymbol{\Psi}} \gets \bar{\mathbf{C}}^{*} - \widehat{\mathbf{C}}$, where $\bar{\mathbf{C}}^{*} = B_1^{-1}\sum_{m} \widehat{\mathbf{C}}^{*(m)}$
\State Set $\delta \gets \max\{\delta \in \mathcal{G} : \rho(\widehat{\mathbf{C}} - \delta\widehat{\boldsymbol{\Psi}}) < 1\}$ and $\widetilde{\mathbf{C}} \gets \widehat{\mathbf{C}} - \delta\widehat{\boldsymbol{\Psi}}$ \label{lin:bab-delta} 
\State \textbf{if} no such $\delta$ \textbf{then} $\widetilde{\mathbf{C}} \gets \widehat{\mathbf{C}}$
\Statex \textit{Stage 2: bootstrap distribution}
\For{$m = 1, \ldots, B_2$}
  \State Draw $\mathbf{u}^{*}_{2,t}$ from $\widehat{\mathbf{u}}^{c}_t$ and generate $\mathbf{y}^{*}_t$ from $\widetilde{\mathbf{C}}$ and $\mathbf{u}^{*}_{2,t}$, starting from $\mathbf{y}^{*}_{-p+1} = \cdots = \mathbf{y}^{*}_{0} = \mathbf{0}$
  \State Re-fit MAR to obtain $\widehat{\mathbf{C}}^{*}$, compute $\delta^{*}$ as above, and set $\widetilde{\mathbf{C}}^{*} \gets \widehat{\mathbf{C}}^{*} - \delta^{*}\widehat{\boldsymbol{\Psi}}$ \label{lin:bab-refit}
  \State Compute $\widehat{\boldsymbol{\Theta}}^{*(m)}_h$ from $\widetilde{\mathbf{C}}^{*}$ via \eqref{eq:mma}
\EndFor
\State \Return percentile interval from $\{\widehat{\theta}^{*(m)}_{ik,h}\}_{m=1}^{B_2}$
\end{algorithmic}
\medskip
{\footnotesize \textit{Notes:} $\mathcal{G} = \{1, 0.99, \ldots, 0.01\}$, $\rho(\cdot)$ is the largest companion eigenvalue in modulus.}
\end{algorithm}

\section{Additional Application Results}
\label{app:additional_application}

In Section \ref{sec:application}, each model is estimated at the lag order selected by the Schwarz and Hannan--Quinn criteria, $p = 2$ for the MAR and $p = 1$ for the VAR.
This appendix re-estimates the VAR at $p = 2$, so that both models share the same lag order and the comparison isolates the effect of the Kronecker restriction.
The VAR then estimates $450$ coefficients, against $66$ for the MAR.
Figure \ref{fig:application_p2} reports the results, computed otherwise as in Section \ref{subsec:results}.
The point estimates again track each other closely, particularly for energy inflation, where the two models differ visibly only in the Netherlands.
The ProBAB-MAR intervals remain narrower than the BAB-VAR intervals, now by a factor of roughly $1.6$ in mean and $1.4$ in median, against $1.4$ and $1.3$ for the VAR(1) in Section \ref{subsec:results}.
This wider gap is consistent with the second lag doubling the number of VAR coefficients while leaving the MAR specification unchanged.

\begin{figure}[htbp]
  \centering
  \includegraphics[width=\linewidth]{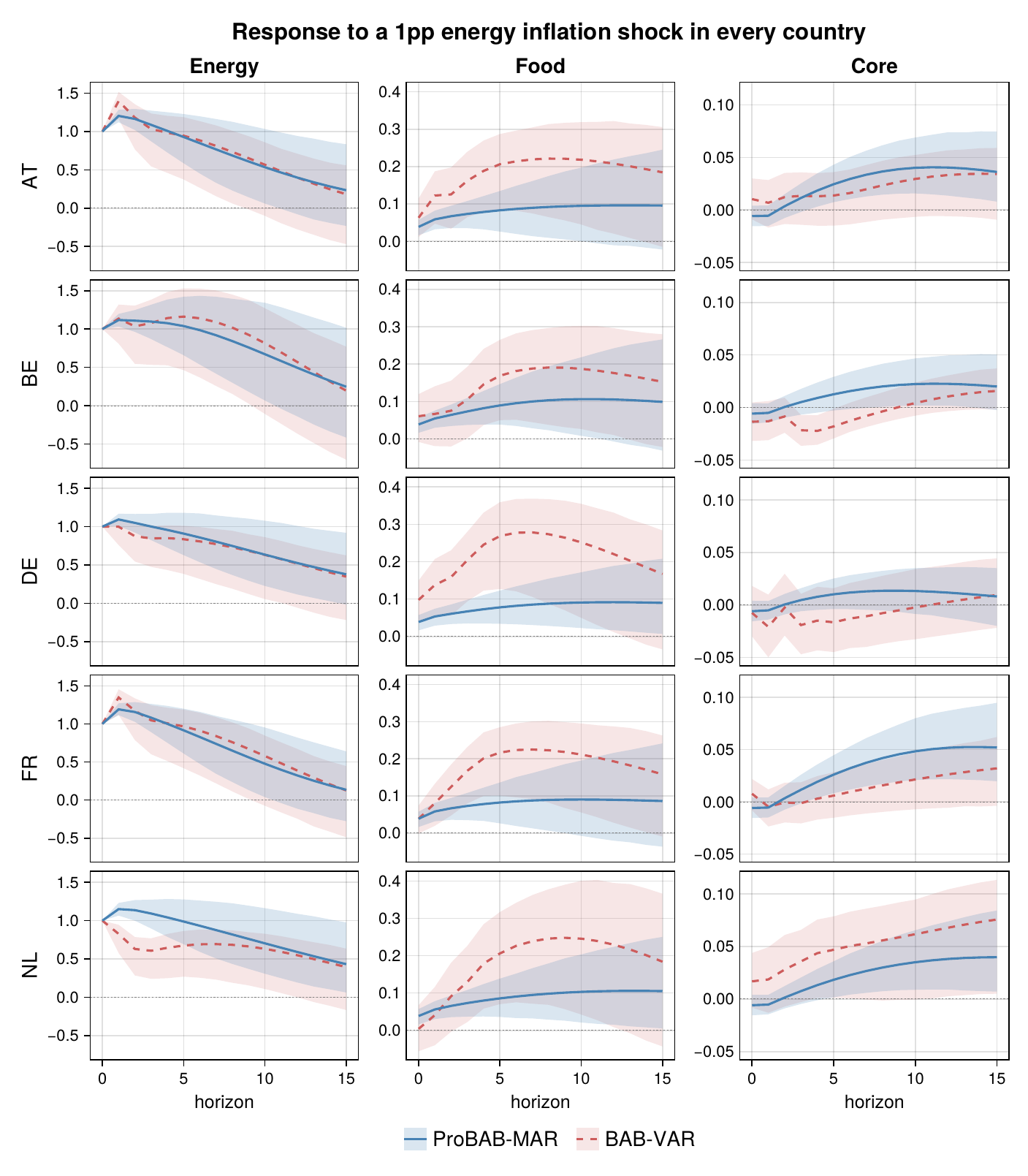}
  \caption{Responses to a shock raising energy inflation by one percentage point on impact in each of the five countries, with $95\%$ bootstrap confidence intervals. Rows index the responding country and columns the responding category, and panels in the same column share a scale. Solid lines with shaded bands are ProBAB-MAR. Dashed lines with bands are BAB-VAR on the vectorized system. Both models are estimated at lag order $p = 2$.}
  \label{fig:application_p2}
\end{figure}

\end{document}